\documentclass[11pt,a4paper]{article}

\usepackage{setspace} 
\usepackage[hmargin=2.5cm,vmargin=1.98cm]{geometry}

\usepackage[T1]{fontenc}
\usepackage[utf8]{inputenc}
\usepackage{amsmath,amssymb,amsthm,amsfonts}
\usepackage{comment}
\usepackage{todonotes}
\usepackage{mathtools}
\usepackage{multirow}
\newcommand{\cmark}{\ding{51}} 
\newcommand{\xmark}{\ding{55}} 
\usepackage{pifont}
\usepackage[colorlinks=true, linkcolor=green, citecolor=blue, urlcolor=blue]{hyperref}
\usepackage{yfonts}
\usepackage{booktabs}
\usepackage{cleveref}
\usepackage{natbib}
\usepackage{todonotes}
\usepackage{tabularx}

\usepackage{array}
\usepackage{caption}

\usepackage{graphicx} 
\usepackage{float} 

\usepackage[math]{blindtext}

\usepackage{graphicx}
\usepackage{dsfont}
\usepackage{units}
\usepackage{enumitem}
\usepackage{booktabs}
\usepackage{color}

\newcommand{\MH}[1]{{\color{red}{#1}}}

\usepackage{tikz}			
\usetikzlibrary{arrows,automata,shapes, trees}

\theoremstyle{plain}
\newtheorem{theorem}{Theorem}[section]
\newtheorem{proposition}[theorem]{Proposition}
\newtheorem{lemma}[theorem]{Lemma}
\newtheorem{corollary}[theorem]{Corollary}
\theoremstyle{definition}
\newtheorem{definition}[theorem]{Definition}
\newtheorem{remark}[theorem]{Remark}
\newtheorem{example}[theorem]{Example}

\theoremstyle{remark}
\numberwithin{equation}{section}
\newtheorem*{empty*}{}

\DeclareMathOperator*{\esssup}{ess\,sup}
\DeclareMathOperator*{\essinf}{ess\,inf}

\newcommand{\e}{\mathrm{e}}
\renewcommand{\theta}{\vartheta}
\renewcommand{\epsilon}{\varepsilon}
\renewcommand{\P}{\mathbb{P}}

\newcommand{\NN}{\mathbb{N}}

\newcommand{\RR}{\mathbb{R}}

\newcommand{\cE}{\mathcal{E}}
\newcommand{\cF}{\mathcal{F}}

\newcommand{\1}{\mathbf{1}}

\newcommand*{\ol}[1]{\bar{#1}}

\newcommand{\ES}{\mathrm{ES}}

\newcommand{\VaR}{\mathrm{VaR}}
\newcommand{\SR}{\mathrm{SR}}
\newcommand{\OCE}{\mathrm{OCE}}
\newcommand{\WC}{\mathrm{WC}}

\newcommand{\LVaR}{\mathrm{LVaR}}
\newcommand{\ALG}{\mathrm{ALG}}
\newcommand{\EW}{\mathrm{EW}}

\newcommand{\SLL}{\mathrm{SLL}}
\newcommand{\ASLL}{\mathrm{ASLL}}
\newcommand{\SE}{\mathrm{SE}}

\def\keywords{\vspace{.5em}
{\noindent\textbf{Keywords}:\,\relax%
}}

\def\JELclassification{\vspace{.5em}
{\noindent\textbf{JEL classification}:\,\relax%
}}

\begin{document}
	
	\title{The Interplay between Utility and Risk in Portfolio Selection\footnote{We are very grateful to Cosimo Munari for stimulating discussions on the material. We also thank the participants of the London-Oxford-Warwick workshop and the Financial Mathematics seminar at King's College London for their comments.}}
	
\author{%
Leonardo Baggiani\thanks{Department of Statistics, University of Warwick, 
\href{mailto:leonardo.baggiani@warwick.ac.uk}{leonardo.baggiani@warwick.ac.uk}}
\and
Martin Herdegen\thanks{Department of Mathematics, University of Stuttgart and Department of Statistics, University of Warwick, 
\href{mailto:martin.herdegen@isa.uni-stuttgart.de}{martin.herdegen@isa.uni-stuttgart.de}}
\and 
Nazem Khan\thanks{Mathematical Institute, University of Oxford, 
\href{mailto:nazem.khan@maths.ox.ac.uk}{nazem.khan@maths.ox.ac.uk}}}

	\date{}
	
	\maketitle
  
\begin{abstract}
We revisit portfolio selection with a utility objective and a risk constraint. Our framework is very general and accommodates non-concave utilities such as $S$-shaped utilities from prospect theory and non-convex risk measures such as Value at Risk. We provide a novel and complete characterization of well-posedness for utility-risk portfolio selection in one period that takes the interplay between utility and risk fully into account. 

In a fixed market, well-posedness is characterized in terms of the behaviour of the utility and risk functionals on the zero-cost payoffs in that market. Across all markets, a necessary and sufficient condition for well-posedness is given by a very simple either-or criterion:\ either the utility or the risk functional need to satisfy the axiom of sensitivity to large losses.

For expected utility, this criterion becomes explicit. Market-independent well-posedness is determined by the asymptotic loss-gain ratio of the utility function and the large-loss sensitivity of the risk functional. For certain $S$-shaped utilities, Value at Risk or Expected Shortfall constraints may  not guarantee  well-posedness across all markets. We therefore return to the fixed-market problem and show that, in elliptical markets, well-posedness reduces to a comparison between the maximal Sharpe ratio and an explicit tail-risk threshold.
\end{abstract}


\JELclassification{G11, C61, D81}


\section{Introduction}

In this paper, we revisit portfolio selection with a utility objective and a risk constraint. An investor trades in a market consisting of one riskless asset and $d$ risky assets. Given a utility functional $\mathcal{U}$, 
a risk functional $\mathcal{R}$, an initial wealth $w > 0$, a risk threshold $R_{\max}$ and denoting by $\bar{S}_0 \coloneqq (S^0_0, \ldots, S^d_0)$ and $\bar{S}_1 \coloneqq (S^0_1, \ldots, S^d_1)$ initial and terminal prices, respectively, the investor solves the optimization problem
\begin{equation}
\label{eq:intro:maximize utility}
\sup_{\bar{\theta} \in \mathbb{R}^{1+d}} \mathcal{U}( \bar{\theta} \cdot \bar{S}_1 ) \quad \text{subject to} \quad
    \begin{cases}
        \bar{\theta} \cdot \bar{S}_0 = w, \\
        \mathcal{R}( \bar{\theta} \cdot \bar{S}_1 ) \leq R_{\max},
    \end{cases}
\end{equation}
To exclude trivially ill-posed or well-posed cases, we assume that the market is arbitrage-free and $\mathcal U$ is nonsatiated. We refer to \eqref{eq:intro:maximize utility} as $(\mathcal U,\mathcal R)$-terminal-wealth portfolio selection, or simply utility-risk terminal-wealth portfolio selection in the sequel.

The central question of this paper is under which conditions utility-risk terminal-wealth portfolio selection is well posed, in the sense that the supremum in \eqref{eq:intro:maximize utility} is attained. This is not automatic. Although the risk constraint is intended to rule out excessive risk-taking, its effectiveness depends on how the risk functional interacts both with the utility objective and the payoffs generated by the market. Indeed, consider an arbitrage-free market with a risk-free bond and a defaultable bond; see Example \ref{exa:mean-var-es-binomial-ill-posed} for details.  If the defaultable bond has positive expected excess return, then an investor maximizing expected terminal wealth may take an increasingly leveraged long position in the bond, financed by shorting the risk-free bond. Expected terminal wealth then diverges to infinity. Nevertheless, if the default event is sufficiently rare, a Value at Risk (VaR) constraint may fail to detect it at the relevant confidence level; and even an Expected Shortfall (ES) constraint may fail to control the position when the upside payoff is sufficiently large relative to the default loss. Thus mean-VaR and mean-ES terminal-wealth portfolio selection can be ill posed in a bounded, arbitrage-free market with strictly positive terminal prices. Ill-posedness is therefore not merely a pathology caused by arbitrage or unrealistic payoffs. It can arise because the chosen risk functional does not control the leveraged loss exposures made attractive by the utility objective.

One possible response is to impose external trading constraints, such as no short-selling constraints or leverage bounds. These restrictions often restore existence by making the feasible set compact. However, they change the nature of the problem. Well-posedness is then enforced by an exogenous constraint on trading strategies, rather than by the utility-risk specification itself. Our aim is different. We seek to understand when the utility objective and the risk constraint, taken together, \emph{themselves} control large downside exposures.

We study this question at two levels. The first is \emph{fixed-market well-posedness}: for a given arbitrage-free market, does \eqref{eq:intro:maximize utility} admit an optimizer for every admissible initial wealth and risk threshold? This is the natural question for an investor or modeller working with a specified market model. The second is \emph{market-independent well-posedness}: for which pairs $(\mathcal U,\mathcal R)$ is \eqref{eq:intro:maximize utility} well posed in every arbitrage-free market? This is a stronger requirement, but it isolates properties of the utility-risk specification that do not depend on a particular market model.

The distinction is important. A pair $(\mathcal U,\mathcal R)$ that fails to be market-independent well posed may still be well posed in a given market, because that market may not contain payoffs along which the utility-risk specification breaks down. Conversely, a pair that is market-independent well posed is robust to the choice of arbitrage-free market model. Our fixed-market results make this distinction precise and show how the well-posedness of a given utility-risk pair can depend on the geometry and distribution of attainable payoffs.

\subsection{Main results}


We distinguish in the sequel between well-posedness and weak well\-posedness. While well-posedness means that the supremum in \eqref{eq:intro:maximize utility} is attained, weak well-posedness means that, although an optimizer need not exist, feasible portfolios cannot drive utility all the way to the bliss point~$\mathcal U(\infty)= \lim_{y \to \infty} \mathcal U(y)$. Thus failure of weak well-posedness is the more severe phenomenon: it means that there exist feasible portfolios whose utility approach the bliss point~$\mathcal U(\infty) $. We refer to such sequences informally as ill-posed sequences.

Our first main result concerns fixed-market well-posedness for a given arbitrage-free market $\bar S$. The relevant objects are the terminal payoffs generated by zero-cost portfolios, that is, payoffs of the form $Z=\bar\eta\cdot \bar S_1$ with $\bar\eta\cdot \bar S_0=0$. The risk functional $\mathcal R$ satisfies \emph{sensitivity to large losses} (SLL) \emph{on $\bar S$} if every zero-cost payoff $Z \neq 0$ is eventually detected by risk when scaled up: $\mathcal R(\lambda Z)>0$ for all sufficiently large $\lambda$. On the utility side, the corresponding condition is \emph{asymptotic sensitivity to large losses} (ASLL) \emph{on $\bar S$}: $\limsup_{\lambda \to \infty} \mathcal{U}(\lambda Z) < \mathcal{U}(\infty)$. In the fixed-market theory we use a \emph{robust} version of this condition, because ill-posed sequences may converge to a limiting zero-cost portfolio rather than lie exactly on it. Robust ASLL therefore requires the same large-loss behaviour to persist under small perturbations of the zero-cost payoff.

Under mild regularity assumptions, Theorem \ref{thm:local wp original characterization} shows that these local large-loss conditions are sufficient for fixed-market well-posedness. Under an additional \emph{alignment} condition, they are also necessary. Alignment means that whenever the risk functional fails to penalize a scaled zero-cost payoff, the utility functional regards the same payoff as asymptotically attractive. This condition is satisfied in many examples. In mean-risk problems, any strictly expectation-bounded positively homogeneous risk functional, such as ES, is aligned with expected return. Or in the case of elliptical returns, alignment holds for $S$-shaped utilities with VaR or ES constraints; see Section \ref{subsec:elliptical implications}.

While fixed-market well-posedness (under alignment) is characterized by the behaviour of $\mathcal U$ and $\mathcal R$ on the zero-cost payoffs generated by the market, the market-independent theory gives an even easier structural criterion. Here, the relevant objects are all payoffs $Y$ in some ambient space $L$ with $\mathbb P[Y<0]>0$. We say that $\mathcal R$ satisfies SLL on $L$ if $\mathcal R(\lambda Y)>0$ for all sufficiently large $\lambda$. Similarly, $\mathcal U$ satisfies SLL on $L$ if $\mathcal U(\lambda Y)<0$ for all sufficiently large $\lambda$. 

Under mild regularity assumptions on $\mathcal U$ or $\mathcal R$, Theorem \ref{thm:global wp original characterization normalized} establishes equivalence between
\begin{itemize}
\item $\mathcal U$ satisfies $\SLL$ on $L$, or $\mathcal R$ satisfies $\SLL$ on $L$;
\item $(\mathcal U,\mathcal R)$-terminal-wealth portfolio selection is market-independent weakly well posed;
\item $(\mathcal U,\mathcal R)$-terminal-wealth portfolio selection is market-independent well posed.
\end{itemize}
Thus,  market-independent well-posedness is characterized by a simple either-or condition. It is enough for one of the two functionals to penalize large losses sufficiently strongly. Conversely, if both fail to do so, then there exists an arbitrage-free market in which feasible portfolios can drive utility to $\mathcal U(\infty)$. 

The property of SLL is studied systematically in \citet{HKM2024}, which provides verifiable criteria for many utility and risk functionals. We use these criteria to make our market-independent condition explicit for expected utility. Let $\mathcal E_u(Y)=\mathbb E[u(Y)]$. If $u(0)=-\infty$, as for logarithmic utility, then sufficiently large losses force expected utility to $-\infty$, and market-independent well-posedness follows for every cash-convex risk functional that satisfies the lower Fatou property. If $u(0)\in\mathbb R$, we normalize $u(0)=0$ and use the asymptotic loss-gain ratio
\begin{equation*}
\ALG(u) := \limsup_{y\to\infty} \frac{|u(-y)|}{u(y)} \in[0,\infty].
\end{equation*}
The quantity $\ALG(u)$ compares the penalty assigned to increasingly large losses with the utility derived from corresponding gains. For a broad class of utility functions, including unbounded utilities that are negatively star-shaped on gains, expected utility satisfies $\SLL$ on $L$ if and only if $\ALG(u)=\infty$. Consequently, under the regularity assumptions above, market-independent well-posedness of $(\mathcal E_u,\mathcal R)$-terminal-wealth portfolio selection reduces to the either-or condition
\begin{equation*}
\ALG(u)=\infty
\quad\text{or}\quad
\mathcal R \text{ satisfies } \SLL \text{ on } L.
\end{equation*}
This covers both classical concave utilities and non-concave $S$-shaped utilities. For
\begin{equation*}
u(y)=
\begin{cases}
y^a, & y\ge 0,\\
-(-y)^b, & y<0,
\end{cases}
\qquad 0<a,b\le 1,
\end{equation*}
we have $\ALG(u)=\infty$ exactly when $a<b$. In that case, the utility functional itself guarantees market-independent well-posedness. By contrast, in the regime $0<b\leq a\leq 1$, including linear utility when $a=b=1$, market-independent well-posedness requires the risk functional to provide the missing large-loss sensitivity.

Table \ref{tab:utility_risk_cases} summarizes the resulting classification for several standard utility-risk pairs. A checkmark indicates market-independent well-posedness; a cross indicates that even market-independent weak well-posedness fails.

\begin{table}[h!]
\centering
\renewcommand{\arraystretch}{1.3}
\begin{tabular}{@{}lcccccc@{}}
\toprule
\textbf{} 
& Mean 
& $S$-shaped $a\ge b$ 
& $S$-shaped $a<b$ 
& Log
& Power 
& Exponential \\
\midrule

No risk constraint
& \xmark & \xmark & \cmark & \cmark & \cmark & \cmark \\

Value at Risk 
& \xmark & \xmark & \cmark & \cmark & \cmark & \cmark \\

Expected Shortfall 
& \xmark & \xmark & \cmark & \cmark & \cmark & \cmark \\

Entropic Risk 
& \cmark & \cmark & \cmark & \cmark & \cmark & \cmark \\

\bottomrule
\end{tabular}
\caption{Market-independent well-posedness of $(\mathcal E_u,\mathcal R)$-terminal-wealth portfolio selection for representative expected utilities and risk functionals. A checkmark denotes market-independent well-posedness; a cross denotes failure of market-independent weak well-posedness.}
\label{tab:utility_risk_cases}
\end{table}

Finally, neither VaR nor ES satisfy $\SLL$ on $L$. This should not be read as saying that VaR and ES are unsuitable risk functionals. Rather, $\SLL$ on $L$ is a strong, market-independent requirement: it asks the risk functional to rule out ill-posed leveraged sequences in every arbitrage-free market. The failure of VaR or ES to satisfy this property means that, when expected utility does not satisfy $\SLL$ these risk functionals alone cannot guarantee well-posedness across all markets.

This is precisely where the fixed-market theory becomes relevant. Once a market model is fixed, one only needs to inspect the zero-cost payoffs generated by that model. We carry this out for elliptical markets. In this setting, the fixed-market criterion simplifies to a comparison between the maximal Sharpe ratio of the market and an explicit threshold determined by the risk functional. If the maximal Sharpe ratio is below the threshold, the risk functional controls all attainable payoffs and the problem is well posed. If it is above the threshold, the market contains zero-cost payoffs along which the risk constraint fails to prevent utility from approaching its bliss point. For Gaussian returns and VaR or ES, Corollary \ref{cor:Gaussian VaR ES Sharpe thresholds} gives explicit thresholds, yielding a directly implementable diagnostic for fixed-market well-posedness.

The remainder of the paper is organized as follows. After the literature review, Section \ref{sec:problem formulation} introduces the model and formalizes the terminal-wealth portfolio selection problem. Section \ref{sec:excess return well posedness} studies the associated excess-return formulation, including fixed-market and market-independent well-posedness. Section \ref{sec:terminal wealth well posedness} translates the main results back to terminal wealth in a simplified setting that applies to the majority of our examples, while the general passage from excess-return to terminal-wealth portfolio selection is developed in Appendix \ref{app:terminal-wealth portfolio selection}. Section 5 discusses related formulations, including risk constraints on changes in wealth and risk minimization under utility constraints. Section \ref{sec:Examples} presents the main examples and applications, with a focus on expected utility and elliptical markets. Section \ref{sec:conclusion} concludes. Appendix \ref{app:counterexamples} contains counterexamples, and Appendix \ref{app:Additional Results and Proofs} collects additional results and proofs.

\subsection{Literature review}

Our work is related to several strands of the portfolio selection literature. We focus on results that are closest to our question: when does utility maximization with (or without) a risk constraint admit an optimizer, and how does this depend on the utility and the risk functional?

The classical starting point is mean-variance portfolio selection. \citet{Markowitz1952} studied the problem of maximizing expected return subject to a variance constraint and showed that the problem admits an explicit well-posed solution. \citet{MasterFundsRockafellar} extended this approach by replacing variance, or rather standard deviation, with more general deviation measures, including lower semideviation and mean absolute deviation. These functionals are well suited to optimization, but they are not monotone risk functionals: a payoff may dominate another almost surely while being assigned larger deviation. This motivates the use of monotone risk functionals as constraints.

A large literature studies mean-risk portfolio selection, where the objective is expected return and the constraint is imposed through a risk functional such as VaR or ES. Important contributions include \citet{campbell2001optimal}, \citet{rockafellar2000optimization,rockafellar2002conditional}, \citet{alexander2002economic}, \citet{bertsimas2004shortfall}, \citet{ciliberti2007feasibility} and \citet{adam2008spectral}. More recent structural results for mean-risk optimization under coherent and star-shaped risk functionals were given by \citet{herdegen2020dual, herdegen2025rho}, who give precise conditions when a risk constraint prevents ill-posed sequences in mean-risk problems. Our paper differs in two respects. First, we allow the objective to be a general utility functional, not only the mean. Second, our primary object is terminal-wealth portfolio selection: utility and risk are evaluated on the investor's final position, and the excess-return formulation is used mainly as a tool to characterize well-posedness of that problem.

The unconstrained case, where the risk constraint is absent, is also closely related. In the one-period expected-utility framework, \citet[\emph{Theorem~3.3}]{follmerschied:2016} show existence and uniqueness under strict concavity and additional conditions such as boundedness from above or the condition $u(a)=-\infty$ for some finite $a < 0$. These assumptions are natural in the classical concave setting, but exclude many non-concave preferences, such as the $S$-shaped utilities associated with prospect theory \citep{kahn1979prospect}. Portfolio choice under cumulative prospect theory is studied by \citet{bernard2010static} and \citet{he2011portfolio} in a one-period model with one risk-free and one risky asset. In particular, \citet[\emph{Section~3.2}]{he2011portfolio} introduce the large-loss aversion degree
$k=\lim_{y\to\infty}\frac{-u(-y)}{u(y)}$
and use it to characterize well-posedness in their setting. When the limit exists, this coincides with the asymptotic loss-gain ratio that appears in our expected-utility applications. Our results recover the same intuition in a different framework, without probability distortion, with multiple risky assets, and in the presence of general risk constraints.

Utility maximization under risk constraints has also been studied beyond mean-risk objectives. \citet{basak2001value} analyze expected utility maximization under a VaR constraint in a continuous-time complete market and show that the constraint may induce undesirable tail behaviour. \citet{gundel2008utility} study expected utility maximization under a shortfall-risk constraint, while \citet{gabih2009utility}, \citet{he2011portfolioquantiles}, \citet{he2015dynamic}, \citet{wei2018risk} and \citet{ghossoub2025risk} consider related risk-constrained utility maximization problems. Closer in spirit to our focus on the interaction between preferences and risk constraints, \citet{armstrong2019risk,armstrong2022coherent,armstrong2024importance} study risk constraints for traders with non-standard preferences and coherent or dynamic risk functionals. These works demonstrate that risk constraints may fail to curb excessive risk-taking, depending on the interaction between the preference functional, the risk functional and the set of admissible trades. Our contribution is to give a sharp well-posedness criterion in a general one-period terminal-wealth model, identifying $\SLL$ as the relevant structural property.

Finally, our paper is connected to the general theory of utility maximization in multi-period and continuous-time markets. \citet{kramkov1999asymptotic} and \citet{schachermayer2001optimal} show that, in continuous time, existence of optimal portfolios is closely linked to asymptotic elasticity conditions on the utility function. In discrete time, \citet{Rasonyi2005} obtain existence results under no-arbitrage and suitable growth assumptions. These results, and their many extensions, of course include one-period models as special cases. Our question is different: we do not seek a general dynamic existence theorem for a fixed utility class, but rather a structural characterization of when a pair $(\mathcal U,\mathcal R)$ rules out ill-posed sequences, both in a fixed one-period market and uniformly over all arbitrage-free one-period markets.

\section{The Problem}
\label{sec:problem formulation}

Throughout the paper, we consider a one-period economy in which uncertainty about the terminal state of the world is modelled by an atomless probability space $(\Omega, \mathcal{F}, \mathbb{P})$. Denote by $L^1$ the space of integrable random variables, and by $L^\infty$ the space of essentially bounded random variables. Let $L$ be a Riesz space such that $L^\infty \subset L \subset L^1$, and assume that it is \emph{law-invariant} in the sense that whenever $Y \in L$ and $Z$ is a random variable on $(\Omega,\mathcal F,\mathbb P)$ with the same law as $Y$, then also $Z \in L$. This is a mild property that is satisfied by all $L^p$ spaces for $p \in [1,\infty]$, and more generally by all Orlicz spaces and Orlicz hearts, which are standard domains in the theory of risk measures; see, for example, \citet{cheridito2009risk} and \citet{gao2018fatou}.

Note that elements of $L$ are interpreted as financial payoffs, not as loss random variables. Accordingly, larger values correspond to more desirable outcomes.

\subsection{Utility and risk functionals}

The reward of a financial payoff is captured by a \emph{utility functional} $\mathcal{U}$, while its riskiness is evaluated using a \emph{risk functional} $\mathcal{R}$. The numerical values assigned by these functionals are not directly important; rather, what matters is the ordering they induce on $L$. We adopt an axiomatic framework. To that end, a functional $\mathcal H:L\to[-\infty,\infty]$ is said to be:
\begin{itemize}
    \item \emph{positive finite} if $\mathcal{H}(c) \in \mathbb{R}$ for all $c \in (0,\infty)$;
    \item \emph{decreasing} if $\mathcal{H}(Y) \geq \mathcal{H}(Z)$ for all $Y, Z \in L$ with $Y \leq Z$ $\mathbb{P}$-a.s.;
    \item \emph{increasing} if $\mathcal{H}(Y) \leq \mathcal{H}(Z)$ for all $Y, Z \in L$ with $Y \leq Z$ $\mathbb{P}$-a.s.
\end{itemize}
Moreover, if $\mathcal{H}$ is increasing or decreasing, then we set $\mathcal{H}(\infty) \coloneqq \lim \limits_{y \to \infty} \mathcal{H}(y) \in [-\infty,\infty]$.

\begin{definition}
A \emph{utility functional} is any increasing functional $\mathcal{U} : L \to [-\infty, \infty)$ that is positive finite and satisfies $\mathcal{U}(Y) < \mathcal{U}(\infty)$ for all $Y \in L$.
A \emph{risk functional} is any decreasing functional $\mathcal{R} : L \to (-\infty, \infty]$ that is positive finite.
\end{definition}

A few comments regarding the above definition are in order. The range restrictions ensure that financial payoffs are not assigned infinite utility or infinitely favourable risk, while still allowing payoffs with utility $-\infty$ or risk $+\infty$. The condition $\mathcal U(Y)<\mathcal U(\infty)$ reflects non-satiation and expresses that no financial payoff $Y\in L$ is as desirable as receiving arbitrarily large deterministic wealth. It will play an important technical role in several of our proofs. By contrast, we do not impose the condition $\mathcal R(Y)>\mathcal R(\infty)$, since this would exclude the important case of a vacuous risk constraint, for instance when $\mathcal R$ is constant; cf.\ Remark \ref{rmk:alternative problems}(b). Finally, positive finiteness is a weak assumption, requiring only that every strictly positive deterministic outcome have finite utility and finite risk. This is natural in our terminal-wealth framework, because the riskless portfolio in \eqref{eq:maximize utility} yields a strictly positive deterministic terminal wealth, and this feasible benchmark should be assigned finite utility and finite risk.

\medskip
Our definitions of utility and risk functionals impose only the minimal structure needed for the terminal-wealth problem. For several results, however, we require additional properties. These are stated for a general functional $\mathcal H:L\to[-\infty,\infty]$. 
\begin{itemize}
    \item \emph{normalization} if $\mathcal{H}(0) = 0$;
    \item \emph{positive homogeneity} if $\mathcal{H}(\lambda Y)=\lambda \mathcal{H}(Y)$ for all $Y \in L$ and $\lambda\in(0,\infty)$;
    \item \emph{negative star-shapedness} if $\mathcal{H}(\lambda Y) \leq \lambda \mathcal{H}(Y)$ for all $Y \in L$ and $\lambda \in (1,\infty)$; or equivalently, if $\mathcal{H}(\lambda Y) \geq \lambda \mathcal{H}(Y)$ for all $Y \in L$ and $\lambda \in (0,1)$;
    \item \emph{positive star-shapedness} if $\mathcal{H}(\lambda Y) \geq \lambda \mathcal{H}(Y)$ for all $Y \in L$ and $\lambda \in (1,\infty)$; or equivalently, if $\mathcal{H}(\lambda Y) \leq \lambda \mathcal{H}(Y)$ for all $Y \in L$ and $\lambda \in (0,1)$;
    \item \emph{cash-additivity} if $\mathcal{H}(Y+c)=\mathcal{H}(Y) + \mathcal{H}(c)$ for all $Y \in L$ and $c \in \mathbb{R}$;
    \item \emph{concavity} if $\mathcal{H}(\lambda Y+(1-\lambda)Z) \geq \lambda \mathcal{H}(Y)+(1-\lambda)\mathcal{H}(Z)$ for all $Y,Z\in L$ and $\lambda\in(0,1)$; it satisfies \emph{strict concavity} if the inequality is strict whenever $\mathbb{P}[Y \neq Z] > 0$;
    \item \emph{convexity} if $\mathcal{H}(\lambda Y+(1-\lambda)Z) \leq \lambda \mathcal{H}(Y)+(1-\lambda)\mathcal{H}(Z)$ for all $Y,Z\in L$ and $\lambda\in(0,1)$; it satisfies \emph{strict convexity} if the inequality is strict whenever $\mathbb{P}[Y \neq Z] > 0$;
     \item \emph{quasi-concavity} if $\mathcal{H}(\lambda Y+(1-\lambda)Z) \geq \min\{\mathcal{H}(Y), \mathcal{H}(Z)\}$ for all $Y,Z\in L$ and $\lambda\in(0,1)$; it satisfies \emph{strict quasi-concavity} if the inequality is strict whenever $\mathbb{P}[Y \neq Z] > 0$;
    \item \emph{quasi-convexity} if $\mathcal{H}(\lambda Y+(1-\lambda)Z) \leq \max \{\mathcal{H}(Y),\mathcal{H}(Z) \}$ for all $Y,Z\in L$ and $\lambda\in(0,1)$; it satisfies \emph{strict quasi-convexity} if the inequality is strict whenever $\mathbb{P}[Y \neq Z] > 0$;
    \item \emph{cash-concavity} if $\mathcal{H}(\lambda Y + (1-\lambda) c) \geq \lambda \mathcal{H}(Y) + (1-\lambda)\mathcal{H}(c)$ for all $\lambda \in (0,1)$, $Y \in L$, and $c \in \RR$;
    \item \emph{cash-convexity} if $\mathcal{H}(\lambda Y + (1-\lambda) c) \leq \lambda \mathcal{H}(Y) + (1-\lambda)\mathcal{H}(c)$ for all $\lambda \in (0,1)$, $Y \in L$, and $c \in \RR$;
    \item \emph{the upper Fatou property} if $\mathcal{H}(Y) \geq \limsup_{n \to \infty} \mathcal{H}(Y_n)$ whenever $Y_n \to Y$ $\mathbb{P}$-a.s., $Y_n, Y \in L$, and $\lvert Y_n \rvert \leq Z$ $\mathbb{P}$-a.s.\ for some $Z \in L$;
    \item \emph{the lower Fatou property} if $\mathcal{H}(Y) \leq \liminf_{n \to \infty} \mathcal{H}(Y_n)$ whenever $Y_n \to Y$ $\mathbb{P}$-a.s., $Y_n, Y \in L$, and $\lvert Y_n \rvert \leq Z$ $\mathbb{P}$-a.s.\ for some $Z \in L$;
    \item \emph{law-invariance} if $\mathcal{H}(Y) = \mathcal{H}(Z)$ whenever $Y, Z \in L$ have the same distribution.
\end{itemize}

Each of the axioms introduced above has a natural economic interpretation and plays a role in our analysis. Normalization is largely a matter of convention whenever the value at the null position is finite. Indeed, if $\widetilde{\mathcal H}:L\to[-\infty,\infty]$ satisfies
$\widetilde{\mathcal H}(0)\in\mathbb R$, then $\mathcal H(Y):=\widetilde{\mathcal H}(Y)-\widetilde{\mathcal H}(0)$ is normalized and induces the same ordering as $\widetilde{\mathcal H}$: for all $Y,Z\in L$, $\widetilde{\mathcal H}(Y)<\widetilde{\mathcal H}(Z)$ if and only if $\mathcal H(Y)<\mathcal H(Z)$. In our setting, normalized functionals also arise naturally in Section~\ref{subsec:Reparametrization}, after rewriting terminal wealth in excess-return coordinates.

Concavity of a utility functional reflects standard risk aversion, while convexity of a risk functional captures preference for diversification. Both are foundational in utility and risk theory; cf.\ \citet[\emph{Chapters 2 and 4}]{follmerschied:2016} and the references therein. Quasi-concavity and quasi-convexity are weaker shape restrictions that still encode diversification. Under normalization, negative and positive star-shapedness are also weaker than concavity and convexity, respectively; they impose diminishing returns to scale and are often more behaviourally plausible than positive homogeneity, which requires exact linear scaling. We refer to \cite{landsberger1990lotteries} for examples of negatively star-shaped utility functionals and to \cite{castagnoli2021star} for a recent presentation of positively star-shaped risk functionals.

Cash-concavity and cash-convexity are weaker than full concavity and convexity but imply star-shapedness by taking the constant equal to $0$. Cash-additivity generalizes the familiar cash-invariance property from the theory of coherent risk measures \cite{artzner1999coherent} (see Proposition \ref{prop:cash-additivity}), but is less common for utility functionals, except in specific frameworks such as Yaari's dual theory \cite{yaari1987dual} and variational preferences \cite{maccheroni2006ambiguity}. Note that cash-additivity together with negative or positive star-shapedness implies cash-concavity or cash-convexity; cf.\ Proposition \ref{prop:cash-convexity criteria}.

The upper and lower Fatou properties are mild regularity conditions ensuring stability under dominated almost-sure limits. In particular, the lower Fatou property is standard in the theory of risk functionals, where it is commonly referred to simply as the Fatou property; see, for example, \citet{delbaen2002coherent}. Finally, law-invariance expresses the idea that only the distribution of outcomes matters, which is a standard assumption in many models of preference.

Important examples of utility and risk functionals are given below. Further examples are discussed in Section \ref{sec:Examples}.

\paragraph{Expected utility.} 
The most prominent class of utility functionals is given by \emph{expected utility}. A function $u:\mathbb{R} \to [-\infty,\infty)$ is called a \emph{utility function}  if it is increasing, $u(c) > -\infty$ for all $c > 0$, $u(\infty) \coloneqq \lim_{y \to \infty}u(y)>u(z)$ for all $z \in \RR$ and $\limsup_{y \to \infty} u(y)/y < \infty$. The corresponding expected utility functional is the map
$\mathcal{E}_u:L\to[-\infty,\infty)$ defined by
\begin{equation*}
    \mathcal E_u(Y)\coloneqq \mathbb E[u(Y)].
\end{equation*}

A utility function captures an agent's preferences over deterministic monetary outcomes. Monotonicity reflects the preference for more wealth. The requirement $u(c)>-\infty$ for all $c>0$ ensures positive finiteness. The non-satiation condition $u(\infty)>u(z)$ for all $z \in \mathbb{R}$ means that no finite deterministic outcome is as desirable as becoming arbitrarily wealthy. Finally, the upper linear-growth condition $\limsup_{y \to \infty} u(y)/y < \infty$ rules out excessively risk-seeking behaviour and ensures that $\mathbb E[u(Y)]$ is well defined as an element of $[-\infty,\infty)$ for every $Y\in L$. Indeed, it implies that there exist $a,y_0>0$ such that $u(y)\leq ay$ for all $y\geq y_0$, and hence the positive part of $u(Y)$ is integrable whenever $Y\in L\subset L^1$.

Our class of utility functions is very rich. It includes concave von Neumann--Morgenstern utility functions, such as exponential, logarithmic and power utility, non-concave
utilities, including the $S$-shaped utility functions used in prospect theory \citet{kahn1979prospect} as well as negatively star-shaped utilities in the sense of \citet{landsberger1990lotteries}.

\paragraph{Value at Risk and Expected Shortfall.}
Examples of risk functionals include \emph{Value at Risk} (VaR) and \emph{Expected Shortfall} (ES), defined for $\alpha\in(0,1)$ by
\begin{equation*}
    \VaR^{\alpha}(Y)
    := \inf \{m\in\mathbb R:\mathbb P[m+Y<0]\leq \alpha\},
    \qquad
    \ES^{\alpha}(Y)
    := \frac{1}{\alpha}\int_0^\alpha \VaR^\beta(Y)\,d\beta,
    \qquad Y\in L.
\end{equation*}
Here $\alpha$ is the tail probability level. VaR disregards losses occurring in the worst $\alpha$-tail and then takes a worst-case view on the remaining scenarios. ES is more conservative. It averages VaR over all tail levels $\beta\in(0,\alpha)$ and, for continuous distributions, coincides with the expected loss in the worst $\alpha$-tail. Smaller values of $\alpha$ correspond to more stringent risk assessment. In particular, as $\alpha\downarrow0$, both $\VaR^\alpha$ and $\ES^\alpha$ converge to the \emph{worst-case risk measure} $\WC:L\to(-\infty,\infty]$, defined by $\mathrm{WC}(Y):=\esssup(-Y)$. Because of their tractability and regulatory relevance, VaR and ES are among the most widely used risk functionals in financial risk management.

\subsection{Market models and excess returns}

A one-period financial market on $L$ with a risk-free asset can be represented by a price process $\bar{S} = (S^0_t,\ldots,S^d_t)_{t \in \{0,1\}}$, where $d \in \NN$, $S^i_0 \in (0,\infty)$ denotes the initial price of asset $i$ and $S^i_1 \in L$ denotes its random terminal price. We assume asset $0$ to be \emph{risk-free} so that $S^0_1 = (1 + r) S^0_0$, where $r \in (-1,\infty)$ denotes the risk-free rate. As we are interested in portfolio selection, we may assume without loss of generality that all assets are \emph{normalized} so that $S^i_0=1$ for all $i$. 

We will only consider market models $\bar S$ that do not admit arbitrage; that is, there does not exist a portfolio $\bar{\theta} \in \mathbb{R}^{1+d}$, expressed in \emph{numbers of shares}, such that
$
\bar{\theta}\cdot \bar{S}_0 \leq 0$, $ \bar{\theta} \cdot \bar{S}_1 \geq 0 \; \mathbb{P}\textnormal{-a.s.}$ and $\mathbb{P} [ \bar{\theta} \cdot \bar{S}_1 > 0 ] > 0.$
Indeed, if arbitrage opportunities were present, the investor could exploit them rather than face a genuine utility-risk trade-off.

As we assume that market models $\bar{S}$ do not admit arbitrage, we may further assume without loss of generality that each market is \emph{non-redundant} in the sense that no risky asset can be synthesized by a linear combination of the other assets. This in turn gives that all the payoffs $S^0_1, \ldots, S^d_1$ are linearly independent as elements of the Riesz space $L$. We denote the collection of all normalized, arbitrage-free and non-redundant market models on $L$ by $\mathbb{S}(L)$.

\medskip
In order to study portfolio selection it is convenient to switch to \emph{excess returns}. For a normalized market model $\ol S = (S^0_t, \ldots, S^d_t)_{t \in \{0, 1\}} \in \mathbb{S}(L)$, we denote the corresponding excess returns by $X = (X^1, \ldots, X^d)$, where \begin{equation*}
    X^i \coloneqq \frac{S^i_1-S^i_0}{S^i_0} - r \in L, \quad i \in \{1,\ldots,d\}.
\end{equation*}
We denote the collection of all excess returns corresponding to $\mathbb{S}(L)$ by $\mathbb{X}(L)$. Note that there exists a bijection between the sets $\mathbb{S}(L)$ and $\mathbb{X}(L) \times (-1, \infty)$ given by $\ol S \mapsto (X, r)$. For this reason and in a slight abuse of notation, we will refer to any element of $\mathbb{X}(L)$ also as a market.

For a given market $X \in \mathbb{X}(L)$, we can describe a portfolio by a vector $\pi = (\pi^1, \dots, \pi^d) \in \mathbb{R}^d$, where $\pi^i$ denotes the fraction of wealth invested in risky asset $i$. The corresponding portfolio \emph{excess return} is
\begin{equation*}
    X_{\pi} := \pi \cdot X.
\end{equation*}

Note that describing portfolio selection in terms of $\pi$ and $X$ is not fully equivalent to describing it in terms of $\ol \theta$ and $\ol S$, as with $\pi$ and $X$ alone, we can neither recover the initial value $w$ of the portfolio nor the interest rate $r$. However, it is easy to check that there exists a bijection between pairs $(\ol \theta, \ol S)$ and quadruples $(w, r, \pi, X)$. This change of coordinates lies at the heart of our analysis.



\begin{remark}
The set $\mathbb{X}(L)$ can be described also without referring first to $\mathbb{S}(L)$ as the set of all excess return vectors $X = (X^1, \ldots, X^d)$ -- where $d \in \NN$ and $X^i \in L$ -- that are arbitrage-free, in the sense that there does not exist $\pi \in \mathbb{R}^{d}$ with $
X_\pi \geq 0 \; \mathbb{P}\textnormal{-a.s.}$ and $\mathbb{P} [ X_\pi > 0 ] > 0$,
and non-redundant, in the sense that $X^1, \ldots, X^d$ are linearly independent.
\end{remark}

Two particularly common classes of one-period market models are elliptical models and empirical models. Both are most naturally specified in terms of excess returns rather than prices since returns are dimensionless, easier to compare across assets, and typically have more stable statistical properties than price levels. Once an excess-return vector $X$ and a risk-free rate $r\in(-1,\infty)$ are specified, the corresponding normalized price process is recovered by
\begin{equation*}
    S^0_0=\cdots=S^d_0=1,\qquad
    S^0_1=1+r,\qquad
    S^i_1=1+r+X^i,\quad i=1,\ldots,d.
\end{equation*}

\paragraph{Elliptical models.} The classical starting point is to assume that the excess-return vector $X=(X^1,\ldots,X^d)$ is multivariate Gaussian. More generally, one may assume that $X$ has an elliptical distribution. We write
    $X\sim \mathcal E_d(\mu,\Sigma,\psi)$
if its characteristic function is of the form
    $\mathbb E[e^{i t^\top X}]
    =
    \exp(i t^\top \mu)\,\psi(t^\top \Sigma t),
  $, $t\in\mathbb R^d$,
where $\mu\in\mathbb R^d$, $\Sigma\in\mathbb R^{d\times d}$ is symmetric positive semidefinite, and $\psi$ is a real-valued function on $[0,\infty)$. The vector $\mu$ is uniquely determined and is called the location vector, while $\Sigma$ is a dispersion matrix and $\psi$ is the characteristic generator. If $X$ has finite first moments, then $\mu$ coincides with the mean vector. If $X$ has finite second moments, then $\Sigma$ is proportional to the covariance matrix and, after a normalization of the generator, may be taken to be the covariance matrix itself.

Elliptical models are widely used in portfolio theory because they retain much of the tractability of Gaussian models while allowing for heavier tails. In applications, the parameters $\mu$ and $\Sigma$ are typically estimated from historical returns. For elliptical models whose law has full support, $\operatorname{supp}(X)=\mathbb R^d$, such as Gaussian or Student models with positive definite dispersion matrix $\Sigma$, we have $X\in\mathbb X(L)$ provided $X^i\in L$ for each $i$. The use of elliptical distributions in portfolio theory goes back at least to \citet{owen1983class}; see also \citet{chamberlain1983characterization} and the more recent discussion in \citet{schuhmacher2021justifying}.

\paragraph{Empirical models.}
A second common class is obtained directly from empirical return distributions. Suppose that $x_1,\ldots,x_N\in\mathbb R^d$ are observed excess-return vectors, for instance the excess returns of $d$ assets over the past $N$ trading days. The empirical model assigns probability $1/N$ to each observation:
\begin{equation*}
    \mathbb P[X=x_j]=\frac{1}{N},\qquad j=1,\ldots,N.
\end{equation*}
More generally, one may assign probabilities $p_1,\ldots,p_N>0$ with $\sum_{j=1}^N p_j=1$. Such finite-support models are the natural one-period version of historical simulation and can be realized within our atomless framework. Indeed, since $(\Omega,\mathcal F,\mathbb P)$ is atomless, there exist disjoint sets $A_1,\ldots,A_N$ with $\mathbb P[A_j]=p_j$, and so $X=\sum_{j=1}^N x_j\mathbf 1_{A_j}$ has the desired empirical distribution. 

For the empirical model to belong to $\mathbb X(L)$, the support points must satisfy the finite scenario analogues of no-arbitrage and non-redundancy. Non-redundancy means that there is no $\pi\in\mathbb R^d\setminus\{0\}$ such that $\pi\cdot x_m=0$ for all $m$. No-arbitrage means that there is no $\pi\in\mathbb R^d\setminus\{0\}$ such that $\pi\cdot x_m\geq 0$ for all $m$ with strict inequality for at least one $m$. Empirical return distributions are commonly used in portfolio optimization under VaR, ES and more general law-invariant risk constraints; see, for example, \citet{rockafellar2000optimization, rockafellar2002conditional}, \citet{bertsimas2004shortfall} and \citet{adam2008spectral}.

\subsection{Terminal-wealth portfolio selection}

We consider an investor who trades in a financial market and seeks to maximize utility subject to a risk constraint. Let $\mathcal{U}$ be a utility functional and $\mathcal{R}$ a risk functional. Given a market $\bar S\in\mathbb S(L)$, an initial wealth $w\in(0,\infty)$, and a maximal risk threshold $R_{\max}\in[\mathcal R(w(1+r)),\infty)$, we define $(\mathcal U,\mathcal R)$-terminal-wealth portfolio selection as the problem
\begin{equation}
\label{eq:maximize utility}
\sup_{\bar\theta\in\mathbb R^{1+d}} \mathcal U(\bar\theta\cdot\bar S_1)
\quad\text{subject to}\quad
\begin{cases}
\bar\theta\cdot\bar S_0=w,\\
\mathcal R(\bar\theta\cdot\bar S_1)\le R_{\max}.
\end{cases}
\end{equation}

It is natural to take $R_{\max}\in[\mathcal R(w(1+r)),\infty)$, since this is the largest interval that guarantees feasibility of the riskless portfolio $\bar\theta=(w,0,\ldots,0)$, which invests all initial wealth in the riskless asset. This portfolio has cost $w$ at time $0$ and terminal wealth $w(1+r)$, hence risk $\mathcal R(w(1+r))$.

\begin{remark}
\label{rmk:alternative problems}
(a) In \eqref{eq:maximize utility}, both utility and risk are evaluated on terminal wealth. If $\mathcal R$ is cash-additive, imposing the risk constraint on terminal wealth is equivalent, after shifting the threshold, to imposing it on the \emph{change in wealth}. Without cash-additivity, however, the choice of reference point is substantive. We discuss the corresponding change-in-wealth formulation, along with other related problems, separately in Section \ref{sec:related problems}.

We emphasize that the formulation \eqref{eq:maximize utility}, in which the risk functional is applied directly to terminal wealth, is standard and economically meaningful. It evaluates the acceptability of the investor's final position itself, rather than its deviation from a reference wealth level.

(b) If $\mathcal R$ is constant, then the risk constraint is vacuous. In this case, \eqref{eq:maximize utility} reduces to unconstrained utility maximization subject only to the budget constraint:
\begin{equation*}
\sup_{\bar\theta\in\mathbb R^{1+d}} \mathcal U(\bar\theta\cdot\bar S_1)
\quad\text{subject to}\quad
\bar\theta\cdot\bar S_0=w.
\end{equation*}
\end{remark}

Even in an arbitrage-free and non-redundant market, the risk constraint in \eqref{eq:maximize utility} need not prevent investors from taking increasingly large leveraged positions. The following elementary example illustrates this point.

\begin{example}
\label{exa:mean-var-es-binomial-ill-posed}
Let $L=L^1$ and $\mathcal U(Y):=\mathbb E[Y]$. Fix $\alpha\in(0,1)$. Consider a market with a risk-free bond and a defaultable bond, normalized so that $S^0_0=S^1_0=1$ and $S^0_1=1+r$ for some $r\in(-1,\infty)$. Let $A$ denote the default event, with $\mathbb P[A]=p\in(0,\alpha)$, and suppose that the defaultable bond has excess return $X:=a\mathbf 1_{A^c}-b\mathbf 1_A$, where $b\in(0,1+r)$ and $a>0$. Thus
    $S^1_1=1+r+X
    =
    (1+r+a)\mathbf 1_{A^c}+(1+r-b)\mathbf 1_A.$
The bond pays more than the risk-free bond if default does not occur, and suffers a loss relative to the risk-free bond on the default event. Since $b<1+r$, its terminal price is strictly positive. Moreover, the market is arbitrage-free and non-redundant. For $\lambda>0$, consider the strategy that invests the fraction $\lambda$ of initial wealth in the defaultable bond and the remaining fraction $1-\lambda$ in the risk-free bond. When $\lambda>1$, this is a leveraged position financed by shorting the risk-free bond. Its terminal wealth is
\begin{equation*}
    W_\lambda=w(1+r)+w\lambda X.
\end{equation*}
If $a>pb/(1-p)$, then $\mathbb E[X]>0$, and hence $\mathbb E[W_\lambda] = w(1+r)+w\lambda\mathbb E[X] \to\infty$ as $\lambda \to \infty$. However, this unbounded increase in expected terminal wealth need not be ruled out by standard tail-risk constraints. Since $p<\alpha$, the default event is too small to be detected by $\VaR^\alpha$ in the relevant quantile calculation, and $\VaR^\alpha(X)=-a$. Therefore, by cash-additivity and positive homogeneity,
\begin{equation*}
    \VaR^\alpha(W_\lambda)
    =
    -w(1+r)-w\lambda a
    \leq
    \VaR^\alpha(w(1+r))
    \leq R_{\max}
\end{equation*}
for every $R_{\max}\in[\VaR^\alpha(w(1+r)),\infty)$ and every $\lambda>0$. Thus the mean-VaR problem is ill posed since expected terminal wealth tends to infinity along feasible leveraged positions.

The same excessive-risk-taking phenomenon occurs for ES if the upside payoff is sufficiently large. For the above two-point excess return, $\ES^\alpha(X) = (pb-(\alpha-p)a)/\alpha$. Hence, if $a>pb/(\alpha - p)$, then $\ES^\alpha(X)<0$, and so by cash-additivity and positive homogeneity of ES,
\begin{equation*}
    \ES^\alpha(W_\lambda)
    =
    -w(1+r)+w\lambda\ES^\alpha(X)
    \leq
    \ES^\alpha(w(1+r))
    \leq R_{\max}
\end{equation*}
for every $R_{\max}\in[\ES^\alpha(w(1+r)),\infty)$ and every $\lambda>0$. Consequently, the mean-ES problem is also ill posed.
\end{example}

Example \ref{exa:mean-var-es-binomial-ill-posed} shows that no-arbitrage alone is not enough to guarantee existence of an optimizer. Moreover, ill-posedness is not caused by a pathological market. Rather, VaR and ES may fail to penalize rare default losses in a way that controls leverage. The resulting risk constraint can therefore permit, and even appear to reward, increasingly large positions in a defaultable bond whose expected excess return is positive.

One way to avoid this difficulty is to impose additional portfolio constraints, such as no short-selling or leverage bounds. Such constraints make the feasible set bounded and therefore restore existence in many cases. However, this changes the nature of the problem. Well-posedness is then enforced by an external trading constraint rather than by the interaction between the utility objective and the risk constraint. Our aim is instead to understand when the pair $(\mathcal U,\mathcal R)$ \emph{on its own} guarantees well-posedness.

This leads to the following notions. We first fix a market and ask whether the optimization problem admits a solution for every admissible initial wealth and risk threshold.

\begin{definition}
\label{defn:local well posedness}
$(\mathcal U,\mathcal R)$-terminal-wealth portfolio selection is \emph{well posed for the market} $\bar S\in\mathbb S(L)$ if, for every initial wealth $w\in(0,\infty)$ and every risk threshold $R_{\max}\in[\mathcal R(w(1+r)),\infty)$, the optimization problem \eqref{eq:maximize utility} admits at least one solution.
\end{definition}

Once existence is known, uniqueness follows under standard strict quasi-concavity/convexity conditions.

\begin{proposition}
\label{prop:terminal-wealth uniqueness}
Let $\bar S\in\mathbb S(L)$. Assume that $(\mathcal U,\mathcal R)$-terminal-wealth portfolio selection is well posed for $\bar S$. If $\mathcal U$ is strictly quasi-concave and $\mathcal R$ is quasi-convex, then for every $w\in(0,\infty)$ and $R_{\max}\in[\mathcal R(w(1+r)),\infty)$, the optimization problem \eqref{eq:maximize utility} admits a unique maximizer.
\end{proposition}

It is also useful to isolate a weaker property. Even when an optimizer does not exist, one may ask whether feasible portfolios can drive utility all the way to the bliss point $\mathcal U(\infty)$.

\begin{definition}
\label{defn:local weak well posedness}
$(\mathcal U,\mathcal R)$-terminal-wealth portfolio selection is \emph{weakly well posed for the market} $\bar S\in\mathbb S(L)$ if, for every initial wealth $w\in(0,\infty)$ and every risk threshold $R_{\max}\in[\mathcal R(w(1+r)),\infty)$, the supremum in \eqref{eq:maximize utility} is strictly less than $\mathcal U(\infty)$.
\end{definition}

Well-posedness implies weak well-posedness. The converse need not hold, even if $\mathcal U$ satisfies the upper Fatou property and $\mathcal R$ satisfies the lower Fatou property; see Example \ref{exa:no (U,R) arbitrage, but portfolio optimization ill posed}.

We will study both fixed-market and market-independent versions of these notions. The fixed-market problem asks when $(\mathcal U,\mathcal R)$-terminal-wealth portfolio selection is (weakly) well posed for a given market $\bar S\in\mathbb S(L)$. The market-independent problem asks for structural conditions on $\mathcal U$ and $\mathcal R$ that guarantee this uniformly over all markets.

\begin{definition}
\label{defn:global well posedness}
$(\mathcal U,\mathcal R)$-terminal-wealth portfolio selection is \emph{market-independent (weakly) well posed} if it is (weakly) well posed for every market $\bar S\in\mathbb S(L)$.
\end{definition}

Thus the two central questions are the following. For a fixed market $\bar S\in\mathbb S(L)$, when is $(\mathcal U,\mathcal R)$-terminal-wealth portfolio selection well posed? And when do conditions on $\mathcal U$ and $\mathcal R$ alone ensure well-posedness across all markets?

\begin{remark}
From an investor's perspective, the fixed-market problem is the more natural one, since the optimization problem is solved in a given market model. From a regulatory or model-robustness perspective, however, the market-independent question is more relevant. It asks whether the utility/risk specification itself rules out ill-posed leveraged sequences, independently of the chosen arbitrage-free market model.
\end{remark}

\subsection{Reparametrization and the excess-return problem}
\label{subsec:Reparametrization}

The terminal-wealth problem \eqref{eq:maximize utility} is the main object of the paper. For the analysis, however, it is convenient to rewrite it in excess-return coordinates.

Let $\bar S \in \mathbb S(L)$ be a normalized market with risk-free rate $r \in (-1,\infty)$, and let $X \in \mathbb X(L)$ denote the corresponding excess-return vector. If $\bar\theta\in\mathbb R^{1+d}$ satisfies $\bar\theta\cdot \bar S_0=w \in (0,\infty)$, then the associated portfolio fractions in the risky assets are given by
\begin{equation*}
    \pi^i:=\frac{\theta^i}{w}, \qquad i=1,\ldots,d.
\end{equation*}
Since $S^i_1=S^0_1+X^i=(1+r)+X^i$ and $S^i_0=1$ for all $i$, we obtain
\begin{equation*}
    \bar\theta\cdot \bar S_1
=\theta^0(1+r)+\sum_{i=1}^d \theta^i\bigl((1+r)+X^i\bigr)
=(1+r)\sum_{i=0}^d \theta^i+\sum_{i=1}^d \theta^i X^i
=w(1+r)+wX_\pi.
\end{equation*}
Thus terminal wealth decomposes into the deterministic benchmark $w(1+r)$ and the risky excess-return payoff $wX_\pi$.

This motivates the following functionals on excess-return payoffs. For $w\in(0,\infty)$ and $r\in(-1,\infty)$, we define the \emph{$(w,r)$-excess-return functionals} associated with the terminal-wealth
functionals $\mathcal U$ and $\mathcal R$ by
\begin{align*}
\mathcal U_{w,r}(Y)
&:= \mathcal U(w(1+r)+wY)-\mathcal U(w(1+r)),\\
\mathcal R_{w,r}(Y)
&:= \mathcal R(w(1+r)+wY)-\mathcal R(w(1+r)).
\end{align*}
These functionals evaluate an excess-return payoff $Y$ by first converting it into terminal wealth via the map $Y\mapsto w(1+r)+wY$, and then subtracting the baseline values $\mathcal U(w(1+r))$ and $\mathcal R(w(1+r))$. In
particular, they are normalized. Moreover, $\mathcal U_{w,r}$ and $\mathcal R_{w,r}$ inherit the key structural properties of $\mathcal U$ and $\mathcal R$; see Proposition \ref{prop:relating U with tilde U} in Appendix \ref{app:Additional Results and Proofs}.

By Proposition \ref{prop:reparameterization equivalent}, for every initial wealth $w\in(0,\infty)$ and every maximal risk threshold $R_{\max}\in[\mathcal R(w(1+r)),\infty)$, the $(\mathcal U,\mathcal R)$-terminal-wealth portfolio selection problem \eqref{eq:maximize utility} is equivalent to the excess-return problem
\begin{equation}
\label{eq:maximize utility reparameterized}
\sup_{\pi\in\mathbb R^d}\mathcal U_{w,r}(X_\pi)
\quad\text{subject to}\quad
\mathcal R_{w,r}(X_\pi)\le \widetilde R_{\max},
\end{equation}
where $\widetilde R_{\max}:=R_{\max}-\mathcal R(w(1+r))\in[0,\infty)$. We refer to \eqref{eq:maximize utility reparameterized} as $(\mathcal U_{w,r},\mathcal R_{w,r})$-excess-return portfolio selection.

\medskip
We now introduce the corresponding notions of (weak) well-posedness for $(\mathcal U_{w,r},\mathcal R_{w,r})$-excess-return portfolio selection, both for a fixed market and in the market-independent setting.

\begin{definition}
\label{defn:local well posedness reparameterized}
Let $w \in (0,\infty)$, $r \in (-1,\infty)$ and $X \in \mathbb X(L)$. We say that $(\mathcal U_{w,r},\mathcal R_{w,r})$-excess-return portfolio selection is 
\begin{itemize}
    \item \emph{well posed for the market $X$} if, for every $\widetilde R_{\max} \in [0,\infty)$, there exists at least one solution to the optimization problem \eqref{eq:maximize utility reparameterized}.
    \item \emph{weakly well posed for the market $X$} if, for every $\widetilde R_{\max} \in [0,\infty)$, the supremum in \eqref{eq:maximize utility reparameterized} is less than $\mathcal U_{w,r}(\infty)$.
\end{itemize}

\end{definition}

\begin{definition}
\label{defn:global well posedness reparameterized}
Let $w \in (0,\infty)$ and $r \in (-1,\infty)$. Then $(\mathcal{U}_{w,r}, \mathcal{R}_{w,r})$-excess-return portfolio selection is \emph{market-independent (weakly) well posed} if it is (weakly) well posed for every market $X \in \mathbb{X}(L)$. 
\end{definition}

\noindent Analogously to Proposition \ref{prop:terminal-wealth uniqueness}, under well-posedness, we have uniqueness if $\mathcal{U}_{w,r}$ is strictly quasi-concave and $\mathcal{R}_{w,r}$ is quasi-convex.

\begin{proposition}
\label{prop:excess-return uniqueness}
Let $w \in (0,\infty)$, $r \in (-1,\infty)$ and $X \in \mathbb X(L)$. Assume $(\mathcal U_{w,r},\mathcal R_{w,r})$-excess-return portfolio selection is well posed for $X$. If $\mathcal U_{w,r}$ is strictly quasi-concave, and $\mathcal R_{w,r}$ is quasi-convex, then for any $\tilde{R}_{\max}\in[ 0,\infty)$, the optimization problem \eqref{eq:maximize utility reparameterized} admits a unique maximizer.
\end{proposition}

The key observation is that both fixed-market and market-independent (weak) well-posedness of terminal-wealth portfolio selection can be expressed in terms of the associated excess-return problems. For a fixed market $\bar S$, the interest rate $r$ and excess-return vector $X$ are fixed, while the initial wealth $w$ still varies. In the market-independent setting one must in addition quantify over all markets, and hence over all admissible pairs $(X,r)$. The next two lemmas make this reduction precise.

\begin{lemma}
\label{lem:local equivalence reparametrization}
Let $\bar S \in \mathbb S(L)$ be a market with risk-free rate $r \in (-1,\infty)$ and corresponding excess-return vector $X \in \mathbb X(L)$. Then the following are equivalent:
\begin{enumerate}[label=\textnormal{(\alph*)}]
    \item $(\mathcal U,\mathcal R)$-terminal-wealth portfolio selection is (weakly) well posed for $\bar S$.
    \item For every $w \in (0,\infty)$, $(\mathcal U_{w,r},\mathcal R_{w,r})$-excess-return portfolio selection is (weakly) well posed for $X$.
\end{enumerate}
\end{lemma}

\begin{lemma}
\label{lemma:global well-posedness}
The following are equivalent:
\begin{enumerate}[label=\textnormal{(\alph*)}]
    \item $(\mathcal U,\mathcal R)$-terminal-wealth portfolio selection is market-independent (weakly) well posed.
    \item For every $w \in (0,\infty)$ and $r \in (-1,\infty)$, $(\mathcal U_{w,r},\mathcal R_{w,r})$-excess-return portfolio selection is market-independent (weakly) well posed.
\end{enumerate}
\end{lemma}

The remainder of the theoretical analysis follows the structure suggested by Lemmas \ref{lem:local equivalence reparametrization} and \ref{lemma:global well-posedness}. In Section \ref{sec:excess return well posedness}, we study $(\mathcal U_{w,r}, \mathcal R_{w,r})$-excess-return portfolio selection, keeping $w$ and $r$ fixed, and analyze both fixed-market and market-independent well-posedness. In Section \ref{sec:terminal wealth well posedness}, we then translate these results back into conditions on the terminal-wealth functionals $\mathcal{U}$ and $\mathcal{R}$.

\begin{remark}
\label{rmk:reparameterized benchmark problem}
Although the terminal-wealth formulation \eqref{eq:maximize utility} is the main object of this paper, the excess-return problem \eqref{eq:maximize utility reparameterized} is also of independent interest. If the subscripts $(w,r)$ are suppressed, it takes the form
\begin{equation*}
    \sup_{\pi\in\mathbb R^d}\mathcal U(X_\pi)
    \quad\text{subject to}\quad
    \mathcal R(X_\pi)\le \widetilde R_{\max},
\end{equation*}
where $\mathcal U$ and $\mathcal R$ are normalized, or merely finite at zero and then normalized. This excess-return formulation covers many classical portfolio selection problems, including mean-risk problems of Markowitz type; cf.\ Corollary \ref{cor:portfolio selection a la markowitz}. In this paper, however, we primarily use the excess-return formulation as a coordinate system for analyzing terminal-wealth portfolio selection. 
\end{remark}

\section{Well-Posedness for Excess-Return Portfolio Selection}
\label{sec:excess return well posedness}

In this section we study the excess-return problem \eqref{eq:maximize utility reparameterized}. Throughout, we fix $w\in(0,\infty)$ and $r\in(-1,\infty)$, and let $\mathcal U_{w,r}$ and $\mathcal R_{w,r}$ denote the corresponding $(w,r)$-excess-return functionals associated with the terminal-wealth functionals $\mathcal U$ and $\mathcal R$. Recall that these excess-return functionals are normalized.

Our goal is to identify conditions on $\mathcal U_{w,r}$ and $\mathcal R_{w,r}$ that guarantee (weak) well-posedness of $(\mathcal U_{w,r},\mathcal R_{w,r})$-excess-return portfolio selection. The corresponding results for the terminal-wealth functionals $\mathcal U$ and $\mathcal R$ will then be obtained in Section \ref{sec:terminal wealth well posedness} by translating these conditions back to the original problem via Lemmas \ref{lem:local equivalence reparametrization} and \ref{lemma:global well-posedness}.


\subsection{Fixed-market well-posedness}
\label{subsec:excess return well posedness:fixed-market}

We first study $(\mathcal U_{w,r},\mathcal R_{w,r})$-excess-return portfolio selection for a fixed market $X\in\mathbb X(L)$. We begin with weak well-posedness. Since the risk threshold in \eqref{eq:maximize utility reparameterized} is allowed to vary, the relevant obstruction is the existence of a portfolio sequence along which utility tends to its bliss point while risk remains uniformly controlled. This leads to the following concept.

\begin{definition}
 We say that the market $X$ admits \emph{$(\mathcal U_{w,r},\mathcal R_{w,r})$-arbitrage} if there exist a sequence of portfolios $(\pi_n)_{n\ge1}\subset \mathbb R^d$ and a constant $\tilde R \in [0,\infty)$ such that
\begin{equation*}
\mathcal U_{w,r}(X_{\pi_n}) \to \mathcal U_{w,r}(\infty)
\qquad\text{and}\qquad
\mathcal R_{w,r}(X_{\pi_n}) \le \tilde R
\quad \text{for all } n \geq 1.
\end{equation*}
\end{definition}


Any sequence generating $(\mathcal U_{w,r},\mathcal R_{w,r})$-arbitrage must be unbounded by Proposition \ref{prop:u rho arbitrage implies unbounded sequence}. Combined with Proposition \ref{prop: (U,R)-arbitrage can be globally centered at 0}, this shows that under the lower Fatou property and positive star-shapedness, such an unbounded sequence gives rise to a single non-zero portfolio whose risk remains non-positive under arbitrary scaling.

\begin{proposition}
\label{prop:U R arbitrage implies not SLL for R}
Assume $\mathcal R_{w,r}$ satisfies the lower Fatou property and positive star-shapedness. If $X$ admits $(\mathcal U_{w,r},\mathcal R_{w,r})$-arbitrage, then there exists $\pi \in \mathbb R^d \setminus \{\mathbf 0\}$ such that
\begin{equation*}
\mathcal R_{w,r}(\lambda X_\pi) \leq 0
\quad \text{for all } \lambda \in (0,\infty).
\end{equation*}
\end{proposition}

\begin{remark}
The terminology is consistent with existing notions of arbitrage. If $\mathcal U_{w,r}$ is the mean and one requires $\widetilde R=0$, the above definition coincides with $\rho$-arbitrage; see \citet[\textit{Definition~3.3}]{herdegen2025rho}. If, in addition, $\mathcal R_{w,r}=\WC$, then $(\mathcal U_{w,r},\mathcal R_{w,r})$-arbitrage is equivalent to ordinary arbitrage. Indeed, Proposition \ref{prop:U R arbitrage implies not SLL for R} and the equivalence
\begin{equation*}
\WC(X_\pi)\leq0
\quad\Longleftrightarrow\quad
X_\pi\geq0\ \text{almost surely}
\end{equation*}
yield one implication, while the converse follows by scaling an arbitrage portfolio.
\end{remark}

This motivates \emph{sensitivity to large losses}: every non-zero portfolio payoff should eventually be detected by the risk functional under sufficiently large leverage.

\begin{definition}
\label{def: ssl local}
We say that $\mathcal R_{w,r}$ satisfies \emph{sensitivity to large losses $(\SLL)$ on $X$} if, for every $\pi\in\mathbb R^d\setminus\{\mathbf 0\}$, there exists $\lambda_\pi\in(0,\infty)$ such that
\begin{equation*}
\mathcal R_{w,r}(\lambda X_\pi)>0
\quad\text{for all }\lambda\in(\lambda_\pi,\infty).
\end{equation*}
\end{definition}

Since $X$ is arbitrage-free, every non-zero portfolio $\pi$ satisfies $\mathbb P[X_\pi<0]>0$. Thus $\SLL$ on $X$ means that if such a portfolio is leveraged strongly enough, then eventually the losses become large enough for the risk functional to turn positive. In other words, the risk constraint cannot remain blind to large leveraged losses. Proposition \ref{prop:U R arbitrage implies not SLL for R} shows that, under the lower Fatou property and positive star-shapedness, $\SLL$ on $X$ is sufficient on the risk side to exclude $(\mathcal U_{w,r},\mathcal R_{w,r})$-arbitrage and hence to guarantee weak well-posedness. 

On the utility side, the situation is more delicate. A sequence generating $(\mathcal U_{w,r},\mathcal R_{w,r})$-arbitrage need not yield a fixed portfolio whose scaled payoff drives utility to its bliss point. Instead, this behaviour may appear only along payoffs that are small perturbations of a limiting portfolio payoff. This is the content of the next proposition.

\begin{proposition}
\label{prop:unbounded sequence of acceptable portfolios}
If $X$ admits $(\mathcal U_{w,r},\mathcal R_{w,r})$-arbitrage, then there exist $\pi\in\mathbb R^d\setminus\{\mathbf 0\}$ and $Y \in L$ such that, for every $\varepsilon>0$, there exists a sequence $(\lambda_n)_{n\ge1}\subset(0,\infty)$ with $\lambda_n\to\infty$ and
\begin{equation*}
\mathcal U_{w,r}\bigl(\lambda_n(X_\pi+\varepsilon Y)\bigr)
\to
\mathcal U_{w,r}(\infty).
\end{equation*}
\end{proposition}

Thus, the asymptotic utility behaviour generated by $(\mathcal U_{w,r},\mathcal R_{w,r})$-arbitrage need not occur along the limiting portfolio payoff $X_\pi$ itself. What is guaranteed is the existence of a perturbation $Y\in L$ such that, for every $\varepsilon>0$, utility approaches its bliss point under unbounded scaling of the perturbed payoff $X_\pi+\varepsilon Y$. Ordinary \emph{asymptotic sensitivity to large losses} controls the exact portfolio payoffs generated by the market, whereas its \emph{robust} version rules out precisely this phenomenon for nearby payoffs.
 
\begin{definition}
\label{def:ASLL local}
We say that $\mathcal U_{w,r}$ satisfies \emph{asymptotic sensitivity to large losses $(\ASLL)$ on $X$} if, for every $\pi\in\mathbb R^d\setminus\{\mathbf 0\}$,
\begin{equation*}
\limsup_{\lambda\to\infty}\mathcal U_{w,r}(\lambda X_\pi)
<
\mathcal U_{w,r}(\infty).
\end{equation*}
We say that $\mathcal U_{w,r}$ satisfies \emph{robust $\ASLL$ on $X$} if, for every $\pi\in\mathbb R^d\setminus\{\mathbf 0\}$ and every $Y\in L$, there exists $\varepsilon_{\pi,Y}>0$ such that
\begin{equation*}
\limsup_{\lambda\to\infty}
\mathcal U_{w,r}\bigl(\lambda(X_\pi+\varepsilon Y)\bigr)
<
\mathcal U_{w,r}(\infty)
\quad\text{for all }\varepsilon\in(0,\varepsilon_{\pi,Y}).
\end{equation*}
\end{definition}

$\ASLL$ requires utility to remain strictly below its bliss point under large scaling of every non-zero portfolio payoff. Robust $\ASLL$ requires this property to persist under all sufficiently small perturbations of such payoffs, which may be interpreted as model misspecification or approximation error.

Combining the preceding utility-side and risk-side observations yields two alternative sufficient conditions for weak well-posedness.

\begin{proposition}
\label{prop:local SLL implies no (U,R) arbitrage}
$(\mathcal U_{w,r},\mathcal R_{w,r})$-excess-return portfolio selection is weakly well posed for $X$ if one of the following conditions is satisfied:
\begin{enumerate}[label=\textnormal{(\alph*)}]
    \item $\mathcal U_{w,r}$ satisfies robust $\ASLL$ on $X$;
    \item $\mathcal R_{w,r}$ satisfies the lower Fatou property, positive star-shapedness, and $\SLL$ on $X$.
\end{enumerate}
\end{proposition}

We now pass from weak well-posedness to well-posedness. Weak well-posedness excludes feasible sequences whose utility approaches the bliss point, but does not by itself guarantee that the supremum is attained. Under the relevant Fatou properties, attainment follows once every maximizing sequence is shown to be bounded. On the risk side, this boundedness is again provided by positive star-shapedness and $\SLL$ on $X$. On the utility side, it is provided by the following stronger robust version of SLL

\begin{definition}
\label{def: robust SLL local}
We say that $\mathcal U_{w,r}$ satisfies \emph{robust $\SLL$ on $X$} if, for every $\pi\in\mathbb R^d\setminus\{\mathbf 0\}$ and every $Y \in L$, there exist $\varepsilon_{\pi,Y}>0$ and $\lambda_{\pi,Y}\in(0,\infty)$ such
that 
\begin{equation*}
\mathcal U_{w,r}\bigl(\lambda(X_\pi+\varepsilon Y)\bigr)<0 \quad \textnormal{for all } \lambda\in(\lambda_{\pi,Y},\infty) \textnormal{ and } \varepsilon\in(0,\varepsilon_{\pi,Y}). 
\end{equation*}

\end{definition}

With this we obtain the following sufficient conditions for well-posedness of $(\mathcal U_{w,r},\mathcal R_{w,r})$-excess-return portfolio selection for the fixed market $X$.

\begin{proposition}
\label{prop:sufficient local well posedness}
Assume that $\mathcal U_{w,r}$ satisfies the upper Fatou property and $\mathcal R_{w,r}$ satisfies the lower Fatou property. Then $(\mathcal U_{w,r},\mathcal R_{w,r})$-excess-return portfolio selection is well posed for $X$ if one of the following conditions is satisfied:
\begin{enumerate}[label=\textnormal{(\alph*)}]
    \item $\mathcal U_{w,r}$ satisfies robust $\SLL$ on $X$;
    \item $\mathcal R_{w,r}$ satisfies positive star-shapedness and $\SLL$ on $X$.
\end{enumerate}
\end{proposition}

We have now obtained sufficient conditions for (weak) well-posedness of $(\mathcal U_{w,r},\mathcal R_{w,r})$-excess-return portfolio selection for the fixed market $X$. These conditions are not necessary in general; see Example \ref{exa:local weak well posedness}. Indeed, the failure of the utility-side condition and the failure of the risk-side condition may occur for different portfolio payoffs. A payoff along which utility can approach its bliss point may be controlled by the risk functional, while a payoff that remains invisible to the risk functional may be unattractive to the utility functional. To obtain a converse in the fixed-market setting, one must exclude this mismatch.

\begin{definition}
\label{def:aligned local}
We say that $\mathcal R_{w,r}$ is \emph{aligned with $\mathcal U_{w,r}$ on $X$} if, for every $\pi\in\mathbb R^d\setminus\{\mathbf 0\}$,
\begin{equation*}
\mathcal R_{w,r}(\lambda X_\pi)\le 0
\ \text{for all }\lambda>0
\implies
\lim_{\lambda\to\infty}\mathcal U_{w,r}(\lambda X_\pi)
=
\mathcal U_{w,r}(\infty).
\end{equation*}
\end{definition}
Alignment requires every non-zero portfolio payoff that remains undetected by the risk constraint under arbitrary scaling also to drive utility to its bliss point. It therefore ensures that a failure of risk control occurs for a payoff that is relevant for utility maximization. 
We proceed to illustrate the condition of alignment in the classical mean-risk case. A further important example is given in the proof of Theorem~\ref{thm:Gaussian fixed market characterization}.
\begin{example}
Suppose that $\mathcal U_{w,r}(Y)=\mathbb E[Y]$ and  $\mathcal R_{w,r}$ is positively homogeneous and \emph{strictly expectation bounded} on the payoffs generated by the market, in the sense that
\begin{equation*}
\mathcal R_{w,r}(X_\pi)>\mathbb E[-X_\pi], \quad \textnormal{for all }\pi \neq \mathbf{0}.
\end{equation*}
If $\mathcal R_{w,r}(\lambda X_\pi)\le0$ for all $\lambda > 0$, then positive homogeneity gives $\mathcal R_{w,r}(X_\pi)\le0$. Strict expectation boundedness implies $\mathbb E[X_\pi]>0$, and hence
\begin{equation*}
\mathcal U_{w,r}(\lambda X_\pi)
=
\lambda\mathbb E[X_\pi]
\to\infty
=
\mathcal U_{w,r}(\infty).
\end{equation*}
Thus $\mathcal R_{w,r}$ is aligned with expected return on $X$.
\end{example}

Under the assumption of alignment, one obtains a complete characterization of (weak) well-posedness for $(\mathcal U_{w,r},\mathcal R_{w,r})$-excess-return portfolio selection for the fixed market $X$.

\begin{theorem}
\label{thm:local weak wp reparametrized characterization}
Assume that $\mathcal U_{w,r}$ satisfies the upper Fatou property, and that $\mathcal R_{w,r}$ satisfies positive star-shapedness and the lower Fatou property. Consider the following statements:
\begin{enumerate}[label=\textnormal{(\alph*)}]
    \item $\mathcal U_{w,r}$ satisfies robust $\ASLL$ on $X$, or
    $\mathcal R_{w,r}$ satisfies $\SLL$ on $X$;
    \item $(\mathcal U_{w,r},\mathcal R_{w,r})$-excess-return portfolio selection is weakly well posed for $X$;
    \item $\mathcal U_{w,r}$ satisfies robust $\SLL$ on $X$, or $\mathcal R_{w,r}$ satisfies $\SLL$ on $X$;
    \item $(\mathcal U_{w,r},\mathcal R_{w,r})$-excess-return portfolio selection is well posed for $X$.
\end{enumerate}
Then \textnormal{(a)} implies \textnormal{(b)}, and \textnormal{(c)} implies
\textnormal{(d)}. If, in addition, $\mathcal R_{w,r}$ is aligned with $\mathcal U_{w,r}$ on $X$, then all four statements are equivalent.
\end{theorem}

The theorem has a simple interpretation. For a fixed market, ill-posed behaviour can only arise through leveraged portfolio payoffs. Either the risk functional eventually detects these payoffs under scaling, or the utility functional prevents them from approaching the bliss point. Alignment is needed only for the converse: without it, the payoff for which risk fails to control leverage may differ from the payoff along which utility can become arbitrarily large.

\subsection{Market-independent well-posedness}
\label{subsec:excess return well posedness:market-independent}

We now turn to market-independent (weak) well-posedness of
$(\mathcal U_{w,r},\mathcal R_{w,r})$-excess-return portfolio selection. In the fixed-market setting of Section \ref{subsec:excess return well posedness:fixed-market}, the relevant large-loss conditions are formulated relative to the portfolio payoffs generated by a given market $X$. In the market-independent setting, the market is allowed to vary, so the corresponding conditions are formulated on all downside-bearing payoffs in the ambient space $L$. This leads to the following global notions.

\begin{definition}
Let $\mathcal R_{w,r}$ be a risk functional and $\mathcal U_{w,r}$ be a utility functional. Then
\begin{itemize}
 \item    $\mathcal R_{w,r}$ satisfies \emph{$\SLL$ on $L$} if, for every $Y \in L$ with $\mathbb P[Y<0]>0$, there exists $\lambda_Y \in (0,\infty)$ such that
\begin{equation*}
\mathcal R_{w,r}(\lambda Y) > 0
\quad \text{for all } \lambda \in (\lambda_Y,\infty).
\end{equation*}
\item $\mathcal U_{w,r}$ satisfies \emph{$\ASLL$ on $L$} if, for every $Y \in L$ with $\mathbb P[Y<0]>0$,
\begin{equation*}
\limsup_{\lambda\to\infty}\mathcal U_{w,r}(\lambda Y)
<
\mathcal U_{w,r}(\infty).
\end{equation*}
\item $\mathcal U_{w,r}$ satisfies \emph{$\SLL$ on $L$} if, for every $Y \in L$ with $\mathbb P[Y<0]>0$, there exists $\lambda_Y \in (0,\infty)$ such that
\begin{equation*}
\mathcal U_{w,r}(\lambda Y) < 0
\quad \text{for all } \lambda \in (\lambda_Y,\infty).
\end{equation*}
\item $\mathcal R_{w,r}$ is \emph{aligned with $\mathcal U_{w,r}$ on $L$} if, for every $Y \in L$ with $\mathbb P[Y<0]>0$,
\begin{equation*}
\mathcal R_{w,r}(\lambda Y) \le 0 \ \text{for all } \lambda>0
\implies
\lim_{\lambda\to\infty}\mathcal U_{w,r}(\lambda Y)=\mathcal U_{w,r}(\infty).
\end{equation*}
\end{itemize}
\end{definition}

One key simplification occurs in the market-independent setting. In the fixed-market theory, the utility-side conditions had to be formulated robustly, because one only tests portfolio payoffs of the form $X_\pi$, while perturbations $X_\pi+\varepsilon Y$ need not themselves be generated by the market. On the whole space $L$, by contrast, such perturbations are simply further elements of $L$. Thus robustness is automatic; see Proposition \ref{prop:robust and nonrobust global coincide}. As a consequence, the market-independent theory can be formulated directly in terms of $\ASLL$ and $\SLL$, without separate robust variants. The next result is the global analogue of Theorem \ref{thm:local weak wp reparametrized characterization}.

\begin{proposition}
\label{prop:global wp reparametrized characterization aligned}
Assume that $\mathcal U_{w,r}$ satisfies the upper Fatou property, and that $\mathcal R_{w,r}$ satisfies positive star-shapedness and the lower Fatou property. Consider the following statements:
\begin{enumerate}[label=\textnormal{(\alph*)}]
    \item $\mathcal U_{w,r}$ satisfies $\ASLL$ on $L$, or
    $\mathcal R_{w,r}$ satisfies $\SLL$ on $L$;
    \item $(\mathcal U_{w,r},\mathcal R_{w,r})$-excess-return portfolio selection is market-independent weakly well posed;
    \item $\mathcal U_{w,r}$ satisfies $\SLL$ on $L$, or
    $\mathcal R_{w,r}$ satisfies $\SLL$ on $L$;
    \item $(\mathcal U_{w,r},\mathcal R_{w,r})$-excess-return portfolio selection is market-independent well posed.
\end{enumerate}
Then \textnormal{(a)} implies \textnormal{(b)}, and \textnormal{(c)} implies
\textnormal{(d)}. If, in addition, $\mathcal R_{w,r}$ is aligned with $\mathcal U_{w,r}$ on $L$, then all four statements are equivalent.
\end{proposition}

Proposition \ref{prop:global wp reparametrized characterization aligned} shows that the distinction between robust and non-robust large-loss conditions disappears in the market-independent setting. Alignment remains as a compatibility condition for the converse, coupling the utility and risk functionals on the same payoff in $L$. Under law-invariance, however, alignment can be dispensed with if the utility functional satisfies the following \emph{sensitivity equivalence} property.

\begin{definition}
\label{def:sensitivity equivalent global}
We say that $\mathcal U_{w,r}$ satisfies \emph{sensitivity equivalence} ($\SE$) if it satisfies $\ASLL$ on $L$ if and only if it satisfies $\SLL$ on $L$.
\end{definition}

SE does not require either property to hold. It requires only that the weak and strong utility-side conditions be equivalent. The class of utility functionals satisfying $\SE$ is broad. It includes all cash-additive and cash-concave utility functionals; see Proposition \ref{prop:WSLL equivalent to SLL}. It also contains many expected utility functionals, including a wide range of concave and non-concave
examples; see Proposition \ref{prop:expected utility sensitivity equivalent}.

The next theorem gives the resulting complete characterization. It replaces alignment, which was needed for the converse in the fixed-market theory, by the more natural assumptions of $\SE$ and law-invariance. Thus, market-independent (weak) well-posedness of $(\mathcal U_{w,r},\mathcal R_{w,r})$-excess-return portfolio selection can be characterized entirely through structural properties of the excess-return functionals $\mathcal U_{w,r}$ and $\mathcal R_{w,r}$.

\begin{theorem}
\label{thm:global wp reparametrized characterization}
Assume that $\mathcal U_{w,r}$ satisfies the upper Fatou property and $\SE$, and that $\mathcal R_{w,r}$ satisfies positive star-shapedness and the lower Fatou property. In addition, suppose that $\mathcal U_{w,r}$ or $\mathcal R_{w,r}$ is law-invariant. Then the following statements are
equivalent:
\begin{enumerate}[label=\textnormal{(\alph*)}]
    \item $\mathcal U_{w,r}$ satisfies $\SLL$ on $L$, or
    $\mathcal R_{w,r}$ satisfies $\SLL$ on $L$;
    \item $(\mathcal U_{w,r},\mathcal R_{w,r})$-excess-return portfolio selection is market-independent weakly well posed;
    \item $(\mathcal U_{w,r},\mathcal R_{w,r})$-excess-return portfolio selection is market-independent well posed.
\end{enumerate}
\end{theorem}

The theorem makes the utility-risk interplay particularly transparent. Market-independent well-posedness holds precisely when at least one of the two functionals detects large losses strongly enough. No separate alignment condition is required.

We conclude with the important special case of Markowitz-type mean-risk portfolio selection on excess returns, for which we suppress the subscripts $(w,r)$. The mean functional is law-invariant, satisfies the upper Fatou property and $\SE$, but does not satisfy $\SLL$ on $L$. Consequently, the preceding theorem reduces market-independent well-posedness entirely to $\SLL$ of the risk functional.

\begin{corollary}
\label{cor:portfolio selection a la markowitz}
Let $\mathcal U(Y):=\mathbb E[Y]$, and let $\mathcal R$ be a positively star-shaped normalized risk functional satisfying the lower Fatou property. Then the Markowitz-type optimization problem
\begin{equation*}
\sup_{\pi\in\mathbb R^d}\mathbb E[X_\pi]
\quad\textnormal{subject to}\quad
\mathcal R(X_\pi)\leq \widetilde R_{\max},
\end{equation*}
admits a solution for every $\widetilde R_{\max}\in[0,\infty)$ and every market $X\in\mathbb X(L)$ if and only if $\mathcal R$ satisfies $\SLL$ on $L$.
\end{corollary}

\section{Well-Posedness for Terminal-Wealth Portfolio Selection}
\label{sec:terminal wealth well posedness}

We now return to the main problem of the paper, namely $(\mathcal U,\mathcal R)$-terminal-wealth portfolio selection. By Lemmas \ref{lem:local equivalence reparametrization} and \ref{lemma:global well-posedness}, the results of Section \ref{sec:excess return well posedness} can be transferred back to the terminal-wealth formulation once the excess-return large-loss conditions are translated into conditions on the original functionals $\mathcal U$ and $\mathcal R$.

The full terminal-wealth theory is slightly more cumbersome than its excess-return counterpart, because without additional assumptions the relevant large-loss conditions must be formulated relative to strictly positive cash baselines. We therefore proceed in two steps. In the main body, we present simplified fixed-market and market-independent results under standard assumptions that recover statements closely parallel to those of Section \ref{sec:excess return well posedness}. The corresponding general theory, valid without these simplifying assumptions, is given in Appendix \ref{app:terminal-wealth portfolio selection}.

\subsection{Fixed-market well-posedness}
\label{subsec:terminal wealth Fixed-market well-posedness}

We first study $(\mathcal U,\mathcal R)$-terminal-wealth portfolio selection for a fixed market. Throughout this subsection, assume $\mathcal R$ is cash-additive and $\mathcal U(0)\in\mathbb R$. By subtracting the finite constant $\mathcal U(0)$, we may assume without loss of generality that $\mathcal U(0)=0$. Since cash-additivity implies normalization, we also have $\mathcal R(0)=0$. The general theory can be found in Appendix \ref{app:subsec:terminal wealth Fixed-market well-posedness}.

Let $\bar S\in\mathbb S(L)$ be a fixed market, with risk-free rate $r\in(-1,\infty)$ and corresponding excess-return vector $X\in\mathbb X(L)$. If $\pi\in\mathbb R^d$ and $\bar\theta$ denotes the corresponding portfolio in numbers of shares with initial wealth $w$, then $\bar\theta\cdot\bar S_1=w(1+r)+wX_\pi$. Equivalently, since $\bar S$ is normalized, every $\pi\in\mathbb R^d$ determines a unique zero-cost portfolio
\begin{equation*}
\bar\eta^\pi:=\Bigl(-\sum_{i=1}^d\pi^i,\pi^1,\ldots,\pi^d\Bigr)
\qquad\text{with}\qquad
\bar\eta^\pi\cdot\bar S_1=X_\pi.
\end{equation*}
Thus, in the fixed-market terminal-wealth problem, the relevant risky payoffs are precisely the terminal payoffs of non-zero zero-cost portfolios. Under the standing assumptions of this subsection, the corresponding large-loss notions simplify to the following normalized versions; the general cash-based formulations are recorded in Appendix \ref{app:subsec:terminal wealth Fixed-market well-posedness}.

\begin{definition}
Let $\mathcal R$ be a risk functional and $\mathcal U$ be a utility function. Then
\begin{itemize}
    \item $\mathcal R$ satisfies \emph{$\SLL$ on $\bar S$} if, for every zero-cost portfolio $\bar\eta\in\mathbb R^{1+d}\setminus\{\mathbf 0\}$ with $\bar\eta\cdot\bar S_0=0$, there exists $\lambda_{\bar\eta}\in(0,\infty)$
such that
\begin{equation*}
\mathcal R\bigl(\lambda\,\bar\eta\cdot \bar S_1\bigr)>0
\quad\text{for all }\lambda\in(\lambda_{\bar\eta},\infty).
\end{equation*}
\item $\mathcal U$ satisfies \emph{robust $\ASLL$ on $\bar S$} if, for every zero-cost portfolio $\bar\eta\in\mathbb R^{1+d}\setminus\{\mathbf 0\}$ with $\bar\eta\cdot\bar S_0=0$ and every $Y\in L$, there exists $\varepsilon_{\bar\eta,Y}>0$ such that
\begin{equation*}
\limsup_{\lambda\to\infty}
\mathcal U\bigl(\lambda(\bar\eta\cdot \bar S_1+\varepsilon Y)\bigr)
<
\mathcal U(\infty)
\quad\text{for all }\varepsilon\in(0,\varepsilon_{\bar\eta,Y}).
\end{equation*}
\item $\mathcal U$ satisfies \emph{robust $\SLL$ on $\bar S$} if, for every zero-cost portfolio $\bar\eta\in\mathbb R^{1+d}\setminus\{\mathbf 0\}$ with $\bar\eta\cdot\bar S_0=0$ and every $Y\in L$, there exist $\varepsilon_{\bar\eta,Y}>0$ and $\lambda_{\bar\eta,Y}\in(0,\infty)$ such that
\begin{equation*}
\mathcal U\bigl(\lambda(\bar\eta\cdot \bar S_1+\varepsilon Y)\bigr)<0 \quad \textnormal{for all } \lambda\in(\lambda_{\bar\eta,Y},\infty) \textnormal{ and } \varepsilon\in(0,\varepsilon_{\bar\eta,Y}).
\end{equation*}
\item $\mathcal R$ is \emph{aligned with $\mathcal U$ on $\bar S$} if, for every zero-cost portfolio $\bar\eta\in\mathbb R^{1+d}\setminus\{\mathbf 0\}$ with $\bar\eta\cdot\bar S_0=0$,
\begin{equation*}
\mathcal R\bigl(\lambda\,\bar\eta\cdot \bar S_1\bigr)\le 0
\ \text{for all }\lambda>0
\implies
\lim_{\lambda\to\infty}\mathcal U\bigl(\lambda\,\bar\eta\cdot \bar S_1\bigr)
=
\mathcal U(\infty).
\end{equation*}
\end{itemize}

\end{definition}

The next theorem is the simplified fixed-market terminal-wealth theorem corresponding to Theorem \ref{thm:local weak wp reparametrized characterization}.

\begin{theorem}
\label{thm:local wp original characterization}
Assume that $\mathcal U$ satisfies normalization and the upper Fatou property, and that $\mathcal R$  satisfies cash-additivity, positive star-shapedness and the lower Fatou property. Consider the following statements:
\begin{enumerate}[label=\textnormal{(\alph*)}]
    \item $\mathcal U$ satisfies robust $\ASLL$ on $\bar S$, or $\mathcal R$ satisfies $\SLL$ on $\bar S$.
    \item $(\mathcal U,\mathcal R)$-terminal-wealth portfolio selection is weakly well posed for $\bar S$.
    \item $\mathcal U$ satisfies robust $\SLL$ on $\bar S$, or $\mathcal R$ satisfies $\SLL$ on $\bar S$.
    \item $(\mathcal U,\mathcal R)$-terminal-wealth portfolio selection is well posed for $\bar S$.
\end{enumerate}
Then \textnormal{(a)} implies \textnormal{(b)}, and \textnormal{(c)} implies \textnormal{(d)}. If, in addition, $\mathcal R$ is aligned with $\mathcal U$ on $\bar S$, then all four statements are equivalent.
\end{theorem}

Theorem \ref{thm:local wp original characterization} shows that fixed-market (weak) well-posedness is determined by how $\mathcal U$ and $\mathcal R$ behave along the large zero-cost trading payoffs available in the market $\bar S$. This local viewpoint is especially useful for risk functionals such as VaR and ES, which are not sensitive to large losses on the whole space $L$. In elliptical markets, for instance, the relevant condition can be expressed in terms of the maximal Sharpe ratio of the market; see Section~\ref{subsec:elliptical implications}.

\subsection{Market-independent well-posedness}
\label{subsec:terminal wealth market-independent well-posedness}

We now turn to the market-independent terminal-wealth problem. As before, we assume $\mathcal U(0)=\mathcal R(0)=0$ with the general theory to be found in Appendix \ref{app:terminal wealth market-independent well-posedness}.

\begin{definition}
Let $\mathcal R$ be a risk functional and $\mathcal U$ be a utility function. Then
\begin{itemize}
    \item $\mathcal R$ satisfies \emph{$\SLL$ on $L$} if, for every $Y\in L$ with $\mathbb P[Y<0]>0$, there is $\lambda_Y\in(0,\infty)$
with
\begin{equation*}
\mathcal R(\lambda Y)>0
\quad\text{for all }\lambda\in(\lambda_Y,\infty).
\end{equation*}
\item $\mathcal U$ satisfies \emph{$\ASLL$ on $L$} if, for every $Y\in L$ with $\mathbb P[Y<0]>0$,
\begin{equation*}
\limsup_{\lambda\to\infty}\mathcal U(\lambda Y)<\mathcal U(\infty).
\end{equation*}
\item $\mathcal U$ satisfies \emph{$\SLL$ on $L$} if, for every $Y\in L$ with $\mathbb P[Y<0]>0$, there is $\lambda_Y\in(0,\infty)$
with
\begin{equation*}
\mathcal U(\lambda Y)<0
\quad\text{for all }\lambda\in(\lambda_Y,\infty).
\end{equation*}
\item $\mathcal U$ satisfies $\SE$ if $\mathcal U$ satisfies $\ASLL$ on $L$ if and only if it satisfies $\SLL$ on $L$.
\end{itemize}

\end{definition}

The next theorem is the market-independent terminal-wealth counterpart of Theorem~\ref{thm:global wp reparametrized characterization}.

\begin{theorem}
\label{thm:global wp original characterization normalized}
Assume that $\mathcal U$ is normalized and satisfies the upper Fatou property and $\SE$, that $\mathcal R$ is normalized, cash-convex and satisfies the lower Fatou property, and that $\mathcal U$ or $\mathcal R$ is law-invariant. Then the following statements are equivalent:
\begin{enumerate}[label=\textnormal{(\alph*)}]
    \item $\mathcal U$ satisfies $\SLL$ on $L$, or $\mathcal R$ satisfies $\SLL$ on $L$;
    \item $(\mathcal U,\mathcal R)$-terminal-wealth portfolio selection is market-independent weakly well posed;
    \item $(\mathcal U,\mathcal R)$-terminal-wealth portfolio selection is market-independent well posed.
\end{enumerate}
\end{theorem}

Theorem \ref{thm:global wp original characterization normalized} shows that market-independent (weak) well-posedness is governed by the same dichotomy as in the excess-return formulation: either the utility functional or the risk functional must be sensitive to large losses on the entire space $L$. The strength of this condition reflects the strength of the conclusion: well-posedness is required to hold for every arbitrage-free market, not only for a fixed market or a restricted model class. In particular, the failure of VaR or ES to satisfy $\SLL$ on $L$ does not prevent well-posedness in a fixed market, but it does show that these risk functionals cannot by themselves guarantee market-independent well-posedness.

In Section~\ref{sec:Examples} we will apply our theory to concrete utility-risk pairs and show how the preceding criteria apply to elliptical markets. Before that, we discuss some formulations related to the terminal-wealth problem \eqref{eq:maximize utility}.

\section{Related Problems}
\label{sec:related problems}

Our main object in this paper is $(\mathcal U,\mathcal R)$-terminal-wealth portfolio selection,
\begin{equation*}
\sup_{\bar\theta\in\mathbb R^{1+d}} \mathcal U(\bar\theta\cdot\bar S_1)
\quad\text{subject to}\quad
\begin{cases}
\bar\theta\cdot\bar S_0=w,\\
\mathcal R(\bar\theta\cdot\bar S_1)\le R_{\max}.
\end{cases}
\end{equation*}
As noted in Remark \ref{rmk:alternative problems}, one may instead apply the risk functional to the change in wealth rather than to terminal wealth itself. We first discuss this variant, and then consider the dual problem of minimizing risk subject to a utility constraint.

\subsection{Risk constraints on changes in wealth}

A natural alternative is to keep utility on terminal wealth, while applying the risk functional to the \emph{change in wealth}:
\begin{equation}
\label{eq:change in wealth risk}
\sup_{\bar\theta\in\mathbb R^{1+d}} \mathcal U(\bar\theta\cdot\bar S_1)
\quad\text{subject to}\quad
\begin{cases}
\bar\theta\cdot\bar S_0=w,\\
\mathcal R(\bar\theta\cdot\bar S_1-w)\le R_{\max}.
\end{cases}
\end{equation}
We refer to \eqref{eq:change in wealth risk} as $(\mathcal U,\mathcal R)$-change-in-wealth portfolio selection. If $\mathcal R$ is cash-additive, then this formulation is equivalent to the terminal-wealth formulation after shifting the risk threshold. Without cash-additivity, however, the choice of reference point matters, and \eqref{eq:change in wealth risk} is a genuinely different problem.

The excess-return reduction used throughout the paper applies equally to \eqref{eq:change in wealth risk}. Let $\bar S\in\mathbb S(L)$ have risk-free rate $r$ and excess-return vector $X\in\mathbb X(L)$. Then every admissible portfolio can be written as $\bar\theta\cdot\bar S_1=w(1+r)+wX_\pi$, for some $\pi\in\mathbb R^d$, and hence $\bar\theta\cdot\bar S_1-w=wr+wX_\pi$. Assume that $\mathcal R(wr)\in\mathbb R$. This is automatic when $r>0$ by positive finiteness. If $r=0$, it requires $\mathcal R(0)\in\mathbb R$, and if $r<0$, then requiring this for all $w\in(0,\infty)$ amounts to finiteness of $\mathcal R$ on all negative deterministic payoffs. Define
\begin{align*}
\mathcal U_{w,r}(Y)
&:=
\mathcal U(w(1+r)+wY)-\mathcal U(w(1+r)),\\
\mathcal R^\Delta_{w,r}(Y)
&:=
\mathcal R(wr+wY)-\mathcal R(wr).
\end{align*}
Then \eqref{eq:change in wealth risk} is equivalent, after shifting the objective and the risk threshold, to
\begin{equation*}
\sup_{\pi\in\mathbb R^d}\mathcal U_{w,r}(X_\pi)
\quad\text{subject to}\quad
\mathcal R^\Delta_{w,r}(X_\pi)\le \widetilde R_{\max}.
\end{equation*}
Thus the excess-return results of Section~\ref{sec:excess return well posedness} apply with $\mathcal R_{w,r}$ replaced by $\mathcal R^\Delta_{w,r}$. Consequently, Theorem~\ref{thm:local wp original characterization} remains valid for the change-in-wealth formulation \eqref{eq:change in wealth risk}, with $(\mathcal U,\mathcal R)$-terminal-wealth portfolio selection replaced throughout by $(\mathcal U,\mathcal R)$-change-in-wealth portfolio selection. In the market-independent setting, the same excess-return reduction and cash-convexity argument show that Theorem~\ref{thm:global wp original characterization normalized} also remains valid with the same replacement.

\subsection{Risk minimization under a utility constraint}

Another related formulation reverses the roles of the objective and the constraint. Given a target utility level $U_{\min}$, consider the $(\mathcal R,\mathcal U)$-risk-minimization problem
\begin{equation}
\label{eq:minimize risk}
\inf_{\bar\theta\in\mathbb R^{1+d}} \mathcal R(\bar\theta\cdot\bar S_1)
\quad\text{subject to}\quad
\begin{cases}
\bar\theta\cdot\bar S_0=w,\\
\mathcal U(\bar\theta\cdot\bar S_1)\ge U_{\min}.
\end{cases}
\end{equation}
This formulation can be obtained from \eqref{eq:maximize utility} by replacing the utility functional by $-\mathcal R$ and the risk functional by $-\mathcal U$. The only additional point is feasibility: if $U_{\min}$ is too high, the constraint set may be empty, which is unrelated to well-posedness. Thus one should restrict attention to targets for which the feasible set is non-empty. With this convention, the preceding theory applies directly. For example, the market-independent characterization becomes the following sign-reversed version of Theorem \ref{thm:global wp original characterization normalized}.

\begin{theorem}
\label{thm:risk minimization market independent characterization}
Assume that $\mathcal R$ is normalized and satisfies the lower Fatou property, that $-\mathcal R$ satisfies $\SE$, that $\mathcal U$ is normalized, cash-concave and satisfies the upper Fatou property, and that $\mathcal U$ or $\mathcal R$ is law-invariant. Then the following statements are equivalent:
\begin{enumerate}[label=\textnormal{(\alph*)}]
    \item $\mathcal U$ satisfies $\SLL$ on $L$, or $\mathcal R$ satisfies $\SLL$ on $L$;
    \item $(\mathcal R,\mathcal U)$-risk-minimization portfolio selection is market-independent weakly well posed;
    \item $(\mathcal R,\mathcal U)$-risk-minimization portfolio selection is market-independent well posed.
\end{enumerate}
\end{theorem}

\section{Examples and Applications}
\label{sec:Examples}

We now turn the abstract criteria into concrete tests. We first identify when expected utility and several standard risk functionals satisfy the large-loss sensitivity and regularity properties required by the market-independent theory. Combining these criteria gives a classification of market-independent well-posedness for expected utility-risk pairs.

We then consider fixed elliptical markets in cases where the market-independent criterion fails. In this setting, well-posedness is no longer determined by sensitivity on the whole space $L$, but by the interaction between the functionals and the zero-cost payoffs generated by the market. For elliptical returns, this interaction reduces to a comparison between the market's maximal Sharpe ratio and a tail-risk threshold.

\subsection{Large-loss sensitivity of expected utility and risk functionals}
\label{subsec:sensitivity criteria}

The main theorems reduce well-posedness to large-loss sensitivity, together with regularity properties such as the Fatou properties, star-shapedness, $\SE$ and law-invariance. We now verify these properties for the expected utility and risk functionals used below.

\subsubsection{Expected utility and the asymptotic loss-gain ratio}

Recall that a utility function is an increasing map $u:\mathbb R\to[-\infty,\infty)$ such that $u(c)\in\mathbb R$ for every $c>0$, $u(\infty) \coloneqq \lim_{y \to \infty}u(y)>u(z)$ for all $z \in \RR$ and $\limsup_{y \to \infty} u(y)/y < \infty$. The associated expected utility functional is $\mathcal E_u(Y)\coloneqq \mathbb E[u(Y)]$ for $ Y\in L$. This is law-invariant. If $u$ is upper semicontinuous, then we further have the upper Fatou property required for our well-posedness results.

\begin{proposition}
\label{prop:expected utility usc}
If $u$ is an upper semicontinuous utility function, then $\mathcal E_u$ satisfies the upper Fatou property.
\end{proposition}

We first handle the case in which utility is negatively infinite at zero, which applies for example to logarithmic utility
\begin{equation*}
u(y)=
\begin{cases}
\log(y), & y > 0,\\
-\infty,          & y \leq 0.
\end{cases}
\end{equation*}
In this case, sufficiently large losses force expected utility to drop to $-\infty$, so the utility functional itself rules out the relevant ill-posed behaviour.

\begin{proposition}
\label{prop:expected utility minus infinity at zero}
Let $u$ be an upper semicontinuous utility function satisfying $u(0)=-\infty$, and let $\mathcal R$ be a cash-convex risk functional satisfying the lower Fatou property. Then $(\mathcal E_u,\mathcal R)$-terminal-wealth portfolio selection is market-independent well posed.
\end{proposition}
 
We therefore restrict attention henceforth to the case $u(0)\in\mathbb R$. Since adding a constant to $u$ does not change the portfolio problem, we normalize $u(0)=0$. We next provide a sufficient condition for $\SE$.

\begin{proposition}
\label{prop:expected utility sensitivity equivalent}
If $u$ is an unbounded utility function with $u(0) = 0$ that is negatively star-shaped on $\mathbb R_+$, then $\mathcal E_u$ satisfies $\SE$.
\end{proposition}

In particular, the proposition applies when $u$ is unbounded and concave on $\mathbb R_+$. Bounded expected utility functionals need not satisfy $\SE$; cf.\ Example \ref{example:sensitivity equivalence}.

It remains to determine when $\mathcal E_u$ satisfies $\SLL$. To this end, we define the \emph{asymptotic loss-gain ratio} of $u$ by 
\begin{equation*}
\mathrm{ALG}(u) := \limsup_{y\to\infty}
\frac{|u(-y)|}{u(y)}
\in[0,\infty].
\end{equation*}
This quantity compares the penalty assigned to increasingly large losses with the utility derived from corresponding gains. By \citet[\emph{Theorem~5.1}]{HKM2024}, if $u$ is negatively star-shaped on $\mathbb R_+$ and $u(0) = 0$, then $\mathcal E_u\text{ satisfies $\SLL$ on }L$ if and only if 
\begin{equation}
\label{eq:expected utility SLL ALG}
\mathcal E_u\text{ satisfies $\SLL$ on }L
\quad\Longleftrightarrow\quad
\mathrm{ALG}(u)=\infty.
\end{equation}
Thus, expected utility is sensitive to large losses precisely when the disutility of large losses asymptotically dominates the utility of corresponding gains.

\begin{remark}
Several standard utility functions illustrate the criterion \eqref{eq:expected utility SLL ALG}. Indeed, $\mathrm{ALG}(u)=\infty$ for exponential and for power utility.
For the $S$-shaped utility 
\begin{equation}
\label{eq:S-shaped utility}
u(y)=
\begin{cases}
y^a,    & y\ge0,\\
-(-y)^b, & y<0,
\end{cases}
\qquad 
0<a,b\le1,
\end{equation}
$\mathrm{ALG}(u) = \infty$ if and only if $a < b$. The regime $b \leq a$, including mean utility at $a=b=1$, provides the main class studied in the fixed-market application below.
\end{remark}

\begin{remark}
Analogous sensitivity criteria are available for other classes of utility functionals, including classical and Bühlmann $u$-mean certainty equivalents; see \citet[\emph{Section~5.2}]{HKM2024}. However, to maintain clarity and highlight the central message of the paper, we focus on expected utility due to its foundational role in decision theory.
\end{remark}

\subsubsection{Risk functionals}
\label{subsec:risk functional examples}

We next record the corresponding sensitivity and regularity properties for several widely used risk functionals. Ordinary $\VaR^\alpha$ and $\ES^\alpha$ are cash-convex and satisfy the lower Fatou property, but neither satisfies $\SLL$ on $L$. We first consider two tail-sensitive modifications of these functionals.

\paragraph{Loss Value at Risk and adjusted Expected Shortfall.}

Let $\beta:(-\infty,0]\to[0,1)$ be increasing. The \emph{Loss Value at Risk} associated with $\beta$ is defined by
\begin{equation*}
\LVaR^\beta(Y)
:=
\sup_{x\in(-\infty,0]}
\left\{\VaR^{\beta(x)}(Y)+x\right\},
\quad Y\in L.
\end{equation*}
It satisfies cash-convexity and the lower Fatou property. Moreover, by \citet[\emph{Proposition~4.1}]{HKM2024}, $\LVaR^\beta$ satisfies $\SLL$ on $L$ if and only if $\inf_{x\in(-\infty,0]}\beta(x)=0$. Constant choices $\beta\equiv\alpha\in(0,1)$ recover $\LVaR^\beta=\VaR^\alpha$, and hence fail this condition.

Similarly, let $g:(0,1]\to[0,\infty]$ be decreasing with $g(1)=0$. The
\emph{adjusted Expected Shortfall} associated with $g$ is
\begin{equation*}
\ES^g(Y)
:=
\sup_{\alpha\in(0,1]}
\left\{\ES^\alpha(Y)-g(\alpha)\right\},
\quad Y\in L.
\end{equation*}
It also satisfies cash-convexity and the lower Fatou property. By \citet[\emph{Proposition~4.2}]{HKM2024}, $\ES^g$ satisfies $\SLL$ on $L$ if and only if $g(\alpha)<\infty$ for every $\alpha\in(0,1]$. Ordinary $\ES^\alpha$ is recovered by taking $g=\infty$ on $(0,\alpha)$ and $g=0$ on $[\alpha,1]$, and hence fails $\SLL$ on $L$.

\paragraph{Risk functionals generated by loss functions.}

A function $\ell:\mathbb R\to(-\infty,\infty]$ is called a \emph{loss function} if it is increasing, $\ell(0)=0$, and $\limsup_{y\to-\infty} \ell(y)/y < \infty$. In analogy with the utility-side definition, set
\begin{equation*}
\ALG(\ell)
:=
\limsup_{y\to\infty}
\frac{\ell(y)}{|\ell(-y)|}
\in[0,\infty].
\end{equation*} 
Thus, if $\ell(y)=-u(-y)$, then $\ALG(\ell)=\ALG(u)$.

\medskip
The \emph{expected weighted loss} associated with $\ell$ is
\begin{equation*}
\EW^\ell(Y):=\mathbb E[\ell(-Y)],
\quad Y\in L.
\end{equation*}
It satisfies the lower Fatou property when $\ell$ is lower semicontinuous and is convex, hence cash-convex, when $\ell$ is convex. If $\ell$ is positively star-shaped on $\mathbb R_-$, then
\citet[\emph{Theorem 5.1}]{HKM2024} gives $\EW^\ell$ satisfies $\SLL$ on $L$ if and only if $\ALG(\ell)=\infty$.

The \emph{shortfall risk measure} associated with $\ell$ is defined by
\begin{equation*}
\SR^\ell(Y)
:=
\inf\left\{
m\in\mathbb R:
\mathbb E[\ell(-Y-m)]\le0
\right\}.
\end{equation*}
It satisfies the lower Fatou property when $\ell$ is lower semicontinuous, and cash-convexity when $\ell$ is positively star-shaped. If, also, $\ell(y)>0$ on $\mathbb R_+$ and $\ell$ is positively star-shaped on $\mathbb R_-$, then \citet[\emph{Proposition~4.3}]{HKM2024} yields $\SR^\ell$ satisfies $\SLL$ on $L$ if and only if $\ALG(\ell)=\infty$.

Finally, let $\ell$ be a real-valued convex loss function satisfying
$\ell(y)\ge y$ for every $y\in\mathbb R$. On its natural Orlicz heart, the
associated \emph{optimized certainty equivalent risk measure} is
\begin{equation*}
\OCE^\ell(Y)
:=
\inf_{\eta\in\mathbb R}
\left\{
\mathbb E[\ell(\eta-Y)]-\eta
\right\}.
\end{equation*}
It is a cash-convex risk functional satisfying the lower Fatou property.
Moreover, by \citet[\emph{Theorem~5.7}]{HKM2024},
$\OCE^\ell$ satisfies $\SLL$ if and only if
\begin{equation}
\label{eq:OCE SLL criterion}
\liminf_{y\to\infty}\frac{\ell(y)}{y}=\infty
\qquad\text{and}\qquad
\limsup_{y\to-\infty}\frac{\ell(y)}{y}=0.
\end{equation}

\subsection{Market-independent implications}

We now combine the preceding sensitivity criteria with Theorem \ref{thm:global wp original characterization normalized}. There are two cases. If $u(0)=-\infty$, Proposition \ref{prop:expected utility minus infinity at zero} already gives market-independent well-posedness for every cash-convex risk functional satisfying the lower Fatou property. If $u(0)\in\mathbb R$, we normalize $u(0)=0$ and use the asymptotic loss-gain ratio to decide whether expected utility itself satisfies $\SLL$.

\begin{corollary}
\label{cor:expected utility market-independent characterization}
Let $u$ be an unbounded, upper semicontinuous utility function with $u(0)=0$ that is negatively star-shaped on $\mathbb R_+$. Let $\mathcal R$ be a normalized, cash-convex risk functional satisfying the lower Fatou property. Then the following statements are equivalent:
\begin{enumerate}[label=\textnormal{(\alph*)}]
    \item $\ALG(u)=\infty$, or $\mathcal R$ satisfies $\SLL$ on $L$;
    \item $(\mathcal E_u,\mathcal R)$-terminal-wealth portfolio selection is
    market-independent weakly well posed;
    \item $(\mathcal E_u,\mathcal R)$-terminal-wealth portfolio selection is
    market-independent well posed.
\end{enumerate}
\end{corollary}

This corollary gives a simple either-or criterion. If $\ALG(u)=\infty$, expected utility is sufficiently loss-sensitive, and the choice of risk functional is irrelevant for market-independent well-posedness, provided the required regularity assumptions hold. If $\ALG(u)<\infty$, expected utility does not control large leveraged losses, and market-independent well-posedness is equivalent to $\SLL$ of the risk functional. Table \ref{tab:global_well_posedness} summarizes the resulting classification for the risk functionals discussed above, under the regularity assumptions stated in Section \ref{subsec:sensitivity criteria}.

\begin{table}[H]
\centering
\small
\setlength{\tabcolsep}{4pt}
\renewcommand{\arraystretch}{1.25}
\begin{tabular}{@{}lp{7.4cm}cc@{}}
\toprule
\textbf{Risk functional}
&
\textbf{Risk-side condition}
&
$\mathrm{ALG}(u)=\infty$
&
$\mathrm{ALG}(u)<\infty$
\\
\midrule

$0$
& Does not satisfy $\SLL$
& \cmark & \xmark
\\

\midrule

$\mathrm{VaR}^{\alpha},\ \mathrm{ES}^{\alpha}$
& Do not satisfy $\SLL$
& \cmark & \xmark
\\

\midrule

$\mathrm{LVaR}^{\beta}$
& Satisfies $\SLL$ if
$\inf_{x\in(-\infty,0]}\beta(x)=0$
& \cmark & \cmark
\\
&
Does not satisfy $\SLL$ otherwise
& \cmark & \xmark
\\

\midrule

$\mathrm{ES}^{g}$
& Satisfies $\SLL$ if $g$ is finite everywhere
& \cmark & \cmark
\\
&
Does not satisfy $\SLL$ otherwise
& \cmark & \xmark
\\

\midrule

$\mathrm{EW}^{\ell}$
& Satisfies $\SLL$ if $\mathrm{ALG}(\ell)=\infty$
& \cmark & \cmark
\\
&
Does not satisfy $\SLL$ if $\mathrm{ALG}(\ell)<\infty$
& \cmark & \xmark
\\

\midrule

$\mathrm{SR}^{\ell}$
& Satisfies $\SLL$ if $\mathrm{ALG}(\ell)=\infty$
& \cmark & \cmark
\\
&
Does not satisfy $\SLL$ if $\mathrm{ALG}(\ell)<\infty$
& \cmark & \xmark
\\

\midrule

$\mathrm{OCE}^{\ell}$
& Satisfies $\SLL$ if \eqref{eq:OCE SLL criterion} holds
& \cmark & \cmark
\\
&
Does not satisfy $\SLL$ otherwise
& \cmark & \xmark
\\

\bottomrule
\end{tabular}
\caption{Market-independent well-posedness of
$(\mathcal E_u,\mathcal R)$-terminal-wealth portfolio selection. A checkmark indicates market-independent well-posedness, while a cross indicates that even market-independent weak well-posedness fails.}
\label{tab:global_well_posedness}
\end{table}


\subsection{Fixed-market implications in elliptical markets}
\label{subsec:elliptical implications}

We now focus on cases where the market-independent criterion fails. Our main example is the $S$-shaped family \eqref{eq:S-shaped utility} with $0<b\leq a \leq 1$, combined with VaR or ES. In this regime, neither the expected utility functional nor the risk functional satisfies $\SLL$ on $L$. Nevertheless, the fixed-market problem may still be well posed, because only the zero-cost payoffs generated by the market are relevant. For clarity, we first state the results for Gaussian markets. The argument uses only the location-scale structure of portfolio returns and therefore extends to full-support elliptical models; see Remark \ref{rmk:Gaussian to elliptical}. So let $X\sim N_d(\mu,\Sigma)$, where $\mu\in\mathbb R^d$ and $\Sigma\in\mathbb R^{d\times d}$ is positive definite, and let $\bar S$ be the corresponding market. For every $\pi\in\mathbb R^d$,
\begin{equation*}
    X_\pi\sim N(\mu_\pi,\sigma_\pi^2),
    \qquad
    \mu_\pi:=\pi^\top\mu,
    \qquad
    \sigma_\pi:=\sqrt{\pi^\top\Sigma\pi}.
\end{equation*}
The maximal Sharpe ratio is
$
\SR_{\max}(X)
:=
\sup_{\pi\ne\mathbf 0}\frac{\pi^\top\mu}
{\sqrt{\pi^\top\Sigma\pi}}
=
\sqrt{\mu^\top\Sigma^{-1}\mu}$.

\begin{proposition}
\label{prop:Gaussian fixed market risk SLL}
Let $\mathcal R$ be a law-invariant, cash-additive, positively homogeneous risk functional satisfying the lower Fatou property. Let $Z\sim N(0,1)$ and suppose $\kappa_{\mathcal R} := \mathcal R(Z)/|\mathcal R(1)| \in\mathbb R$. Then $\mathcal R$ satisfies $\SLL$ on $\bar S$ if and only if $\SR_{\max}(X)<\kappa_{\mathcal R}$.
\end{proposition}

Combined with Theorem \ref{thm:local wp original characterization}, Proposition \ref{prop:Gaussian fixed market risk SLL} gives a utility-independent sufficient condition for fixed-market well-posedness. Whenever $\SR_{\max}(X)<\kappa_{\mathcal R}$, the risk functional satisfies $\SLL$ on $\bar S$, and hence $(\mathcal U,\mathcal R)$-terminal-wealth portfolio selection is well posed for every utility functional satisfying the required regularity assumptions.

For $S$-shaped utilities with $0<b \leq a \leq 1$, this condition is also necessary. Thus,  the problem is well posed when the maximal Sharpe ratio lies below the standardized risk threshold, whereas it is not even weakly well posed once the threshold is reached.

\begin{theorem}
\label{thm:Gaussian fixed market characterization}
Let $\mathcal R$ be a law-invariant, cash-additive, positively homogeneous risk functional satisfying the lower Fatou property. Let $Z \sim N(0,1)$ and suppose $\kappa_\mathcal{R}:=\mathcal{R}(Z)/|\mathcal{R}(1)| \in (0,\infty)$. Let $u$ be the $S$-shaped utility with $0< b \leq a \leq 1$. Then the following statements are equivalent:
\begin{enumerate}[label=\textnormal{(\alph*)}]
    \item $\SR_{\max}(X)<\kappa_{\mathcal R}$;
    \item $(\mathcal E_u,\mathcal R)$-terminal-wealth portfolio selection is
    weakly well posed for $\bar S$;
    \item $(\mathcal E_u,\mathcal R)$-terminal-wealth portfolio selection is
    well posed for $\bar S$.
\end{enumerate}
\end{theorem}

\begin{remark}
The condition $\kappa_{\mathcal R}>0$ is mild. In particular, if $\mathcal R$ is a law-invariant coherent risk measure that does not coincide with expected loss, then \citet[\emph{Corollary~5.1}]{follmer2013convex} show that $\mathcal R(Y)>\mathbb E[-Y]$ for every nonconstant $Y$, and therefore $\kappa_{\mathcal R} = \mathcal R(Z)/1> 0$.
\end{remark}

Theorem \ref{thm:Gaussian fixed market characterization} gives a directly implementable diagnostic. Given estimates of $\mu$ and $\Sigma$, one computes $\SR_{\max}(X)$ and compares it with $\kappa_{\mathcal R}$. For these preferences, this comparison determines well-posedness without solving the optimization problem. Specializing to VaR and ES gives explicit thresholds.

\begin{corollary}
\label{cor:Gaussian VaR ES Sharpe thresholds}
Let $u$ be an $S$-shaped utility with $0<b \leq a \leq 1$, and let $\Phi$ and $\varphi$ denote the standard normal distribution function and density. Then:
\begin{enumerate}[label=\textnormal{(\alph*)}]
    \item For $\alpha \in(0,1/2)$ and $\mathcal R=\VaR^\alpha$, $(\mathcal{E}_u, \mathcal{R})$-terminal-wealth
    portfolio selection is (weakly) well posed for $\bar S$ if and only if $\SR_{\max}(X)<-\Phi^{-1}(\alpha).$

    \item For $\alpha\in(0,1)$ and $\mathcal R=\ES^\alpha$, $(\mathcal{E}_u, \mathcal{R})$-terminal-wealth portfolio selection is (weakly) well posed for $\bar S$ if and only if $\SR_{\max}(X) < \varphi(\Phi^{-1}(\alpha))/\alpha$.
\end{enumerate}
\end{corollary}

\begin{remark}
\label{rmk:Gaussian to elliptical}
The Gaussian assumption is not essential. Suppose that $X\sim\mathcal E_d(\mu,\Sigma,\psi)$ has full support and finite second moments, with $\Sigma=\operatorname{Cov}(X)$ positive definite. Then there exists a symmetric random variable $Z_\psi$ with mean zero and variance one such that, for every $\pi \in \mathbb{R}^d$,
\begin{equation*}
X_\pi\overset{d}{=}\mu_\pi+\sigma_\pi Z_\psi,
\qquad
\mu_\pi=\pi^\top\mu,
\qquad
\sigma_\pi=\sqrt{\pi^\top\Sigma\pi}.
\end{equation*}
The distribution of $Z_\psi$ is independent of $\pi$. Hence, defining
$\kappa_{\mathcal R,\psi}
:=
\frac{\mathcal R(Z_\psi)}{|\mathcal R(1)|}$,
the preceding proofs apply with $\kappa_{\mathcal R}$ replaced by $\kappa_{\mathcal R,\psi}$. Thus well-posedness is characterized by
\begin{equation*}
\SR_{\max}(X)
=
\sqrt{\mu^\top\Sigma^{-1}\mu}
<
\kappa_{\mathcal R,\psi}.
\end{equation*}
For VaR and ES, the corresponding thresholds are $\kappa_{\VaR^\alpha,\psi}=\VaR^\alpha(Z_\psi)$ and $\kappa_{\ES^\alpha,\psi}=\ES^\alpha(Z_\psi)$. This includes multivariate Student models with degrees of freedom greater than two, where $Z_\psi$ is the corresponding univariate Student random variable with mean zero and variance one.
\end{remark}

\section{Conclusion and Outlook}
\label{sec:conclusion}

This paper studies when utility-risk terminal-wealth portfolio selection is well posed in one-period markets. The starting point is that even in bounded arbitrage-free markets with strictly positive terminal prices, a utility objective may reward leveraged positions whose downside is not sufficiently controlled by the chosen risk constraint. 

Our first contribution is a fixed-market theory. For a given market, the relevant objects are the zero-cost payoffs. We show that local large-loss conditions on the utility and risk functionals are sufficient for well-posedness, and become necessary under an alignment condition; see Theorem \ref{thm:local wp original characterization}. Thus fixed-market well-posedness depends not only on the utility-risk pair, but also on the geometry and distribution of the attainable payoff space.

Our second contribution is a market-independent theory. Under mild regularity assumptions, market-independent weak well-posedness and market-independent well-posedness are equivalent, and both are characterized by a simple either-or condition: either the utility functional or the risk functional must be sensitive to large losses; see Theorem \ref{thm:global wp original characterization normalized}. In this sense, utility and risk share responsibility for controlling large leveraged losses. If one functional penalizes such losses sufficiently strongly, it can compensate for the weakness of the other; if neither does, ill-posed sequences may exist in some arbitrage-free market.

For expected utility, the criterion becomes particularly transparent. The relevant quantity is the asymptotic loss-gain ratio $\ALG(u)$. For a broad class of utility functions, expected utility is sensitive to large losses exactly when $\ALG(u)=\infty$. Hence market-independent well-posedness of $(\mathcal E_u,\mathcal R)$-terminal-wealth portfolio selection reduces to the condition that either $\ALG(u)=\infty$ or $\mathcal R$ is sensitive to large losses. This gives a direct classification of many standard examples, including concave utilities, $S$-shaped utilities, VaR, ES, entropic risk and related tail-sensitive risk functionals; see Table \ref{tab:global_well_posedness}.

The fixed-market results are essential when the market-independent criterion fails. Ordinary VaR and ES do not satisfy sensitivity to large losses on the whole space, so they cannot guarantee market-independent well-posedness when expected utility is not itself sufficiently loss-sensitive. This does not make VaR or ES unsuitable; rather, it means that their effectiveness depends on the market model. For elliptical returns, we show that this dependence can be expressed through the maximal Sharpe ratio. In Gaussian markets, the fixed-market well-posedness of VaR- and ES-constrained problems is determined by explicit Sharpe-ratio thresholds; see Corollary \ref{cor:Gaussian VaR ES Sharpe thresholds}.

Several questions remain open. One direction is computational. Much of the numerical literature on risk-constrained portfolio selection imposes no-short-selling or other compactness restrictions; see, e.g., \citet{adam2008spectral} and \citet{geissel2022portfolio}. Our results make it possible to revisit such empirical and parametric problems without imposing these restrictions. A second direction is to develop sharper fixed-market characterizations beyond elliptical distributions, including dual or martingale-measure criteria in the spirit of \citet{herdegen2020dual,herdegen2025rho}. A third is to extend the present theory to robust utility and model uncertainty, where the utility objective itself is evaluated across a family of probability measures. Finally, it would be natural to ask which parts of the one-period theory admit multi-period or continuous-time analogues.

\appendix
\small
\section{Counterexamples}
\label{app:counterexamples}

\begin{example}
\label{exa:no (U,R) arbitrage, but portfolio optimization ill posed}
In this example we construct a utility functional $\mathcal{U}$ that satisfies the upper Fatou property, a risk functional $\mathcal{R}$ that satisfies the lower Fatou property and a market $\bar{S} \in \mathbb{S}(L)$ such that $(\mathcal{U},\mathcal{R})$-terminal-wealth portfolio selection is weakly well posed for $\bar{S}$, but not well posed.

\medskip
Let $L = L^1$. Define the utility functional $\mathcal{U} : L \to \mathbb{R}$ by
\begin{equation*}
    \mathcal{U}(Y)=\begin{cases}
        \essinf Y,&\text{if } Y \geq 0 \ \mathbb{P}\text{-a.s.,} \\ 
        f(\mathbb{E}[Y]),&\text{otherwise}, 
    \end{cases} 
\end{equation*}
where $f:\mathbb{R} \to \mathbb{R}$ is an increasing upper semicontinuous function with $f(a) > 0$ for some $a > 0$, and for all $y \in \RR$, $f(y) \leq y$ and $f(y) < f(\infty)\coloneqq\lim_{z \to \infty} f(z) < \infty$. For example, one can take $f(y)=1-\exp(-y)$ for $y \in \RR$. Let $\mathcal{R} \equiv 0$. Note that $\mathcal{U}$ satisfies the upper Fatou property, $\mathcal{U}(\infty) = \infty$ and $\mathcal{R}$ satisfies the lower Fatou property. Let $\bar{S} = (S^0_t, S^1_t)_{t \in \{0,1\}}$ where $S^0_0 = S^0_1 = 1$, $S^1_0 = 1$ and $S^1_1 = 1+X^1$. Assume $X^1 \in L$ is unbounded from below and $\mathbb{E}[X^1] > 0$ so $\bar{S} \in \mathbb{S}(L)$. 

Choose $w\in(0,f(\infty))$ and $R_{\max}\in[0,\infty)$. The optimization problem \eqref{eq:maximize utility} is then equivalent to
\begin{equation*}
\sup_{\vartheta^1\in\mathbb R} U(w+\vartheta^1X^1).
\end{equation*}
For $\vartheta^1>0$, the random payoff $w+\vartheta^1X^1$ is negative with positive probability, and hence
\begin{equation*}
U(w+\vartheta^1X^1)=f(w+\vartheta^1\mathbb E[X^1]) \uparrow f(\infty)
\qquad\text{as }\vartheta^1\to\infty.
\end{equation*}
Since $f(y)<f(\infty)$ for every $y\in\mathbb R$, this value is not attained. Moreover, the riskless portfolio gives utility $U(w)=w<f(\infty)$, while positions with $\vartheta^1<0$ do not yield values above $f(\infty)$. Thus the supremum is $f(\infty)$, but it is not attained. Therefore $(\mathcal U,\mathcal R)$-terminal-wealth portfolio selection is not well posed for the market $\bar S$. It is nevertheless weakly well posed for $\bar S$, since the supremum is finite and strictly smaller than $\mathcal U(\infty)=\infty$.
\end{example}

\begin{example}
\label{exa:local weak well posedness}
The converse of Proposition \ref{prop:local SLL implies no (U,R) arbitrage} fails in general, even for law-invariant utility and risk functionals on an atomless probability space.

\medskip
Let $(\Omega,\mathcal F,\mathbb P)=([0,1],\mathcal B([0,1]),\mathrm{Leb})$ and consider the one-dimensional market $X=(X^1)$ given by
\begin{equation*}
X^1(\omega)=
\begin{cases}
2, & \omega\in[0,0.4],\\
-1, & \omega\in(0.4,1].
\end{cases}
\end{equation*}
Define
\begin{equation*}
\mathcal U(Y):=\mathbb E[Y],
\qquad
\mathcal R(Y):=\VaR^{0.5}(Y),
\qquad Y\in L^1.
\end{equation*}
Then both $\mathcal U$ and $\mathcal R$ are law-invariant. Also $\mathcal R$ is not sensitive to large losses on $X$: if $\pi<0$, then $\pi X^1<0$ only on a set of probability $0.4<0.5$, and therefore
\begin{equation*}
\VaR^{0.5}(\lambda X_\pi)\le 0,
\quad\text{for all }\lambda>0.
\end{equation*}
And $\mathcal U$ is not robustly weakly sensitive to large losses on $X$: taking $\pi=1$ and $Y\equiv 1$, we have
\begin{equation*}
\mathcal U\bigl(\lambda(X_\pi+\varepsilon Y)\bigr)
=
\lambda(0.2+\varepsilon)\to\infty
=
\mathcal U(\infty),
\quad\text{for every }\varepsilon>0.
\end{equation*}

Nevertheless, $(\mathcal U,\mathcal R)$-portfolio selection is weakly well posed for $X$. Indeed, if $\pi>0$, then $\VaR^{0.5}(\pi X^1)=\pi$, so the constraint $\VaR^{0.5}(\pi X^1)\le \widetilde R_{\max}$ implies $\pi\le \widetilde R_{\max}$. If $\pi<0$, then $\mathbb E[\pi X^1]=0.2\pi<0$. Hence
\begin{equation*}
\sup_{\pi:\,\VaR^{0.5}(\pi X^1)\le \widetilde R_{\max}} \mathbb E[\pi X^1]
=
0.2\,\widetilde R_{\max}
<
\infty
=
\mathcal U(\infty).
\end{equation*}
\end{example}

\begin{example}
\label{example:sensitivity equivalence}
In this example we show that the expected utility functional $\mathcal{E}_u$ may not be sensitivity equivalent if the utility function $u$ is bounded.

\medskip
Let $u:\mathbb{R} \to \mathbb{R}$ be defined by
\begin{equation*}
    u(y) = \begin{cases}
        1-\e^{-y},&\text{if } y \geq 0, \\
        \e^y-1,&\text{if } y < 0.
    \end{cases}
\end{equation*}
Then $\mathcal{E}_u$ is asymptotically sensitive to large losses. Indeed, if $Y \in L$ and $\mathbb{P}[Y < 0] > 0$, then 
\begin{equation*}
    \lim_{\lambda \to \infty}\mathcal{E}_u(\lambda Y) = 
    \mathbb{P}[Y > 0] - \mathbb{P}[Y < 0] < 1 = \mathcal{E}_u(\infty).
\end{equation*}
However, $\mathcal{E}_u$ is not sensitive to large losses since the limit can remain positive. For instance, if the gain region dominates the loss region, that is $\mathbb{P}[Y > 0] > \mathbb{P}[Y < 0]$, then $\lim_{\lambda \to \infty}\mathcal{E}_u(\lambda Y) > 0$.
\end{example}

\section{General Theory for Terminal-Wealth Portfolio Selection}
\label{app:terminal-wealth portfolio selection}

\subsection{Fixed-market well-posedness}
\label{app:subsec:terminal wealth Fixed-market well-posedness}

We derive here the general fixed-market terminal-wealth theory, without the simplifying assumptions imposed in Section \ref{subsec:terminal wealth Fixed-market well-posedness}. Let $\bar S\in\mathbb S(L)$ be a fixed market with risk-free rate $r\in(-1,\infty)$ and corresponding excess-return vector $X\in\mathbb X(L)$. By Lemma \ref{lem:local equivalence reparametrization}, the relevant terminal-wealth conditions are obtained by translating the fixed-market excess-return notions from Section \ref{subsec:excess return well posedness:fixed-market} back to terminal wealth. Since, in general, the utility and risk functionals need not be normalized or cash-additive, the corresponding notions must be formulated relative to a strictly positive cash baseline $c>0$.

\begin{definition}
Let $c>0$, $\mathcal R$ a risk functional and  $\mathcal U$ a utility functional. Then
\begin{itemize}
\item $\mathcal R$ satisfies \emph{cash $\SLL$ on $\bar S$ at level $c$} if, for every zero-cost portfolio $\bar\eta\in\mathbb R^{1+d}\setminus\{\mathbf 0\}$ with $\bar\eta\cdot\bar S_0=0$, there exists $\lambda_{\bar\eta,c}\in(0,\infty)$ such that
\begin{equation*}
\mathcal R\bigl(c+\lambda\,\bar\eta\cdot \bar S_1\bigr)>\mathcal R(c)
\quad\text{for all }\lambda\in(\lambda_{\bar\eta,c},\infty).
\end{equation*}
\item $\mathcal U$ satisfies \emph{robust cash $\ASLL$ on $\bar S$ at level $c$} if, for every zero-cost portfolio $\bar\eta\in\mathbb R^{1+d}\setminus\{\mathbf 0\}$ with $\bar\eta\cdot\bar S_0=0$ and every $Y\in L$, there exists $\varepsilon_{\bar\eta,Y,c}>0$ such that
\begin{equation*}
\limsup_{\lambda\to\infty}
\mathcal U\bigl(c+\lambda(\bar\eta\cdot \bar S_1+\varepsilon Y)\bigr)
<
\mathcal U(\infty)
\quad\text{for all }\varepsilon\in(0,\varepsilon_{\bar\eta,Y,c}).
\end{equation*}
\item $\mathcal U$ satisfies \emph{robust cash $\SLL$ on $\bar S$ at level $c$} if, for every zero-cost portfolio $\bar\eta\in\mathbb R^{1+d}\setminus\{\mathbf 0\}$ with $\bar\eta\cdot\bar S_0=0$ and every $Y\in L$, there exist $\varepsilon_{\bar\eta,Y,c}>0$ and $\lambda_{\bar\eta,Y,c}\in(0,\infty)$ such that
\begin{equation*}
\mathcal U\bigl(c+\lambda(\bar\eta\cdot \bar S_1+\varepsilon Y)\bigr)
<
\mathcal U(c)
\end{equation*}
for all $\lambda\in(\lambda_{\bar\eta,Y,c},\infty)$ and $\varepsilon\in(0,\varepsilon_{\bar\eta,Y,c})$.
\item $\mathcal R$ is \emph{cash aligned with $\mathcal U$ on $\bar S$ at level $c$} if, for every zero-cost portfolio $\bar\eta\in\mathbb R^{1+d}\setminus\{\mathbf 0\}$ with $\bar\eta\cdot\bar S_0=0$,
\begin{equation*}
\mathcal R\bigl(c+\lambda\,\bar\eta\cdot \bar S_1\bigr)\le \mathcal R(c)
\ \text{for all }\lambda>0
\implies
\lim_{\lambda\to\infty}\mathcal U\bigl(c+\lambda\,\bar\eta\cdot \bar S_1\bigr)
=
\mathcal U(\infty).
\end{equation*}
\item $\mathcal R$ is \emph{cash aligned with $\mathcal U$ on $\bar S$} if it is cash aligned with $\mathcal U$ on $\bar S$ at every level $c>0$.
\end{itemize}
\end{definition}

The next proposition shows that these are exactly the fixed-market terminal-wealth counterparts of the notions introduced in Section \ref{subsec:excess return well posedness:fixed-market}.

\begin{proposition}
\label{app:prop:local properties transfer back}
Let $c>0$, and set $w:=c/(1+r)$. Then the following hold:
\begin{enumerate}[label=\textnormal{(\alph*)}]
\item $\mathcal R$ satisfies cash $\SLL$ on $\bar S$ at level $c$ if and only if $\mathcal R_{w,r}$ satisfies $\SLL$ on $X$.
\item $\mathcal U$ satisfies robust cash $\ASLL$ on $\bar S$ at level $c$ if and only if $\mathcal U_{w,r}$ satisfies robust $\ASLL$ on $X$.
\item $\mathcal U$ satisfies robust cash $\SLL$ on $\bar S$ at level $c$ if and only if $\mathcal U_{w,r}$ satisfies robust $\SLL$ on $X$.
\item $\mathcal R$ is cash aligned with $\mathcal U$ on $\bar S$ at level $c$ if and only if $\mathcal R_{w,r}$ is aligned with $\mathcal U_{w,r}$ on $X$.
\end{enumerate}
\end{proposition}

\begin{proof}
Since $\bar S$ is normalized, every portfolio $\bar\eta=(\eta^0,\eta^1,\ldots,\eta^d)\in\mathbb R^{1+d}$ with $\bar\eta\cdot \bar S_0=0$ satisfies $\eta^0=-\sum_{i=1}^d\eta^i$. Hence, if we set
$\pi:=(\eta^1,\ldots,\eta^d)\in\mathbb R^d$, then
\begin{equation*}
\bar\eta\cdot \bar S_1
=
\eta^0S_1^0+\sum_{i=1}^d \eta^i S_1^i
=
\sum_{i=1}^d \eta^i(S_1^i-S_1^0)
=
X_\pi.
\end{equation*}
Conversely, every $\pi\in\mathbb R^d$ determines a zero-cost portfolio $\bar\eta^\pi:=(-\sum_{i=1}^d \pi^i,\pi^1,\ldots,\pi^d)$ such that $\bar\eta^\pi\cdot \bar S_1=X_\pi$. Since the market is non-redundant, $\bar\eta\neq 0$ if and only if $\pi\neq 0$. Finally, the relation $c=w(1+r)$ determines $w\in(0,\infty)$ uniquely. 

(a) For $\bar\eta\cdot \bar S_1=X_\pi$ and $c=w(1+r)$, $\mathcal R_{w,r}(\lambda X_\pi) = \mathcal R(c+w\lambda\,\bar\eta\cdot \bar S_1)-\mathcal R(c)$. Since $w\lambda$ ranges over $(0,\infty)$ as $\lambda$ ranges over $(0,\infty)$, the defining inequalities are equivalent.

(b) Let $\pi\in\mathbb R^d\setminus\{0\}$, $Y\in L$, and set $\bar\eta:=\bar\eta^\pi$. Then
\begin{equation*}
\mathcal U_{w,r}\bigl(\lambda(X_\pi+\varepsilon Y)\bigr)
=
\mathcal U\bigl(c+w\lambda(\bar\eta\cdot \bar S_1+\varepsilon Y)\bigr)-\mathcal U(c),
\end{equation*}
and $\mathcal U_{w,r}(\infty)=\mathcal U(\infty)-\mathcal U(c)$. Since $w\lambda\to\infty$ if and only if $\lambda\to\infty$, the defining inequalities for robust cash $\ASLL$ on $\bar S$ at level $c$ and robust $\ASLL$ on $X$ are equivalent.

(c) The proof is identical to (b), replacing $\limsup_{\lambda\to\infty} \mathcal U_{w,r}\bigl(\lambda(X_\pi+\varepsilon Y)\bigr) < \mathcal U_{w,r}(\infty)$ by $\mathcal U_{w,r}\bigl(\lambda(X_\pi+\varepsilon Y)\bigr)<0$ for all sufficiently large $\lambda$.

(d) Let $\pi\in\mathbb R^d\setminus\{0\}$ and
$\bar\eta:=\bar\eta^\pi$. Then
\begin{equation*}
\mathcal R_{w,r}(\lambda X_\pi)\le 0
\ \text{for all }\lambda>0
\iff
\mathcal R(c+w\lambda\,\bar\eta\cdot \bar S_1)\le \mathcal R(c)
\ \text{for all }\lambda>0,
\end{equation*}
and similarly
\begin{equation*}
\lim_{\lambda\to\infty}\mathcal U_{w,r}(\lambda X_\pi)=\mathcal U_{w,r}(\infty)
\iff
\lim_{\mu\to\infty}\mathcal U(c+\mu\,\bar\eta\cdot \bar S_1)=\mathcal U(\infty),
\end{equation*}
where $\mu=w\lambda$. Hence the two notions of alignment are equivalent.
\end{proof}

We can now translate the fixed-market excess-return results from Section~\ref{subsec:excess return well posedness:fixed-market} back to terminal-wealth portfolio selection. The only additional structural assumption that appears is cash-convexity of $\mathcal R$. This is precisely what ensures that, for every $w\in(0,\infty)$, the excess-return risk functional $\mathcal R_{w,r}$ is positively star-shaped, as required in Theorem \ref{thm:local weak wp reparametrized characterization}.

\begin{theorem}
\label{app:thm:local wp original characterization}
Assume that $\mathcal U$ satisfies the upper Fatou property, and that $\mathcal R$ satisfies cash-convexity and the lower Fatou property. Consider the following statements:
\begin{enumerate}[label=\textnormal{(\alph*)}]
    \item For every $c>0$, either $\mathcal U$ satisfies robust cash $\ASLL$ on $\bar S$ at level $c$, or $\mathcal R$ satisfies cash $\SLL$ on $\bar S$ at level $c$.
    \item $(\mathcal U,\mathcal R)$-terminal-wealth portfolio selection is weakly well posed for $\bar S$.
    \item For every $c>0$, either $\mathcal U$ satisfies robust cash $\SLL$ on $\bar S$ at level $c$, or $\mathcal R$ satisfies cash $\SLL$ on $\bar S$ at level $c$.
    \item $(\mathcal U,\mathcal R)$-terminal-wealth portfolio selection is well posed for $\bar S$.
\end{enumerate}
Then \textnormal{(a)} implies \textnormal{(b)}, and \textnormal{(c)} implies \textnormal{(d)}. If, in addition, $\mathcal R$ is cash aligned with $\mathcal U$ on $\bar S$, then all four statements are equivalent.
\end{theorem}

\begin{proof}
By Proposition \ref{prop:relating U with tilde U}, for every $w\in(0,\infty)$, $\mathcal U_{w,r}$ satisfies the upper Fatou property, while
$\mathcal R_{w,r}$ satisfies the lower Fatou property and positive star-shapedness. 

By Proposition \ref{app:prop:local properties transfer back}, statement (a) is equivalent to
\begin{equation*}
\forall\,w\in(0,\infty),\quad
\bigl[
\mathcal U_{w,r}\text{ satisfies robust }\ASLL\text{ on }X
\ \text{or}\
\mathcal R_{w,r}\text{ satisfies }\SLL\text{ on }X
\bigr],
\end{equation*}
and statement (c) is equivalent to
\begin{equation*}
\forall\,w\in(0,\infty),\quad
\bigl[
\mathcal U_{w,r}\text{ satisfies robust }\SLL\text{ on }X
\ \text{or}\
\mathcal R_{w,r}\text{ satisfies }\SLL\text{ on }X
\bigr].
\end{equation*}
Hence, by Theorem \ref{thm:local weak wp reparametrized characterization}, statement (a) implies that $(\mathcal U_{w,r},\mathcal R_{w,r})$-excess-return portfolio selection is weakly well posed for $X$ for every $w\in(0,\infty)$, and statement (c) implies that $(\mathcal U_{w,r},\mathcal R_{w,r})$-excess-return portfolio selection is well posed for $X$ for every $w\in(0,\infty)$. By Lemma \ref{lem:local equivalence reparametrization}, these are equivalent to
(b) and (d), respectively.

Now assume in addition that, $\mathcal R$ is cash aligned with $\mathcal U$ on $\bar S$. Thus for every $c>0$, $\mathcal R$ is cash aligned with $\mathcal U$ on $\bar S$ at level $c$. By Proposition \ref{app:prop:local properties transfer back}, this is equivalent to alignment of $\mathcal R_{w,r}$ with $\mathcal U_{w,r}$ on $X$ for every $w\in(0,\infty)$. Hence, for each $w\in(0,\infty)$, all four statements in Theorem \ref{thm:local weak wp reparametrized characterization} are equivalent. Translating back via Proposition \ref{app:prop:local properties transfer back} and Lemma \ref{lem:local equivalence reparametrization}, we conclude that (a), (b), (c) and (d) are equivalent.
\end{proof}

\subsection{Market-independent well-posedness}
\label{app:terminal wealth market-independent well-posedness}

We record here the general market-independent terminal-wealth theory, without the simplifying finite-at-zero assumptions imposed in Section \ref{subsec:terminal wealth market-independent well-posedness}. In this setting, the relevant large-loss notions must be formulated relative to a strictly positive cash baseline $c>0$.

\begin{definition}
Let $c>0$, $\mathcal R$ a risk functional and  $\mathcal U$ a utility functional. Then
\begin{itemize}
\item $\mathcal R$ satisfies \emph{cash $\SLL$ on $L$ at level $c$} if, for every $Y\in L$ with $\mathbb P[Y<0]>0$, there
exists $\lambda_{Y,c}\in(0,\infty)$ such that
\begin{equation*}
\mathcal R(c+\lambda Y)>\mathcal R(c)
\quad\text{for all }\lambda\in(\lambda_{Y,c},\infty).
\end{equation*}
\item $\mathcal U$ satisfies \emph{cash $\ASLL$ on $L$ at level $c$} if, for every $Y\in L$ with $\mathbb P[Y<0]>0$,
\begin{equation*}
\limsup_{\lambda\to\infty}\mathcal U(c+\lambda Y)
<
\mathcal U(\infty).
\end{equation*}
\item  $\mathcal U$ satisfies \emph{cash $\SLL$ on $L$ at level $c$} if, for every $Y\in L$ with $\mathbb P[Y<0]>0$, there exists $\lambda_{Y,c}\in(0,\infty)$ such that
\begin{equation*}
\mathcal U(c+\lambda Y)<\mathcal U(c)
\quad\text{for all }\lambda\in(\lambda_{Y,c},\infty).
\end{equation*}
\item $\mathcal R$ is \emph{cash aligned with $\mathcal U$ on $L$ at level $c$} if, for every $Y\in L$ with $\mathbb P[Y<0]>0$,
\begin{equation*}
\mathcal R(c+\lambda Y)\le \mathcal R(c)\ \text{for all }\lambda>0
\implies
\lim_{\lambda\to\infty}\mathcal U(c+\lambda Y)=\mathcal U(\infty).
\end{equation*}
\item $\mathcal R$ is \emph{cash aligned with $\mathcal U$ on $L$} if it is cash aligned with $\mathcal U$ on $L$ at every level $c>0$.
\end{itemize}

\end{definition}

These are exactly the market-independent terminal-wealth counterparts of the global excess-return notions from Section~\ref{subsec:excess return well posedness:market-independent}. The next
proposition makes this precise.

\begin{proposition}
\label{app:prop:global properties transfer back}
Let $c>0$, $r\in(-1,\infty)$, and set $w:=c/(1+r)$. Then the following hold:
\begin{enumerate}[label=\textnormal{(\alph*)}]
    \item $\mathcal R$ satisfies cash $\SLL$ on $L$ at level $c$ if and only if $\mathcal R_{w,r}$ satisfies $\SLL$ on $L$.
    \item $\mathcal U$ satisfies cash $\ASLL$ on $L$ at level $c$ if and only if $\mathcal U_{w,r}$ satisfies $\ASLL$ on $L$.
    \item $\mathcal U$ satisfies cash $\SLL$ on $L$ at level $c$ if and only if $\mathcal U_{w,r}$ satisfies $\SLL$ on $L$.
    \item $\mathcal R$ is cash aligned with $\mathcal U$ on $L$ at level $c$ if and only if $\mathcal R_{w,r}$ is aligned with
    $\mathcal U_{w,r}$ on $L$.
\end{enumerate}
\end{proposition}

\begin{proof}
(a) Let $Y\in L$ satisfy $\mathbb P[Y<0]>0$. Since $c=w(1+r)$, $\mathcal R_{w,r}(\lambda Y) = \mathcal R(c+w\lambda Y)-\mathcal R(c)$. Because $w\lambda$ ranges over $(0,\infty)$ as $\lambda$ ranges over $(0,\infty)$, the defining inequalities for cash $\SLL$ on $L$ at level $c$ and $\SLL$ of $\mathcal R_{w,r}$ on $L$ are equivalent.

(b) Let $Y\in L$ satisfy $\mathbb P[Y<0]>0$. Then
\begin{equation*}
\mathcal U_{w,r}(\lambda Y)
=
\mathcal U(c+w\lambda Y)-\mathcal U(c),
\qquad
\mathcal U_{w,r}(\infty)=\mathcal U(\infty)-\mathcal U(c).
\end{equation*}
Since $w\lambda\to\infty$ if and only if $\lambda\to\infty$, the defining inequalities for cash $\ASLL$ on $L$ at level $c$ and $\ASLL$ of $\mathcal U_{w,r}$ on $L$ are equivalent.

(c) The proof is identical to (a), replacing $\mathcal R$ by $\mathcal U$ and the inequality $\mathcal R(c+\lambda Y)>\mathcal R(c)$ by $\mathcal U(c+\lambda Y)<\mathcal U(c)$.

(d) Let $Y\in L$ satisfy $\mathbb P[Y<0]>0$. Then
\begin{equation*}
\mathcal R_{w,r}(\lambda Y)\le 0
\ \text{for all }\lambda>0
\iff
\mathcal R(c+w\lambda Y)\le \mathcal R(c)
\ \text{for all }\lambda>0,
\end{equation*}
and similarly
\begin{equation*}
\lim_{\lambda\to\infty}\mathcal U_{w,r}(\lambda Y)=\mathcal U_{w,r}(\infty)
\iff
\lim_{\mu\to\infty}\mathcal U(c+\mu Y)=\mathcal U(\infty),
\end{equation*}
where $\mu=w\lambda$. Hence the two notions of alignment are equivalent.
\end{proof}

In view of Proposition \ref{app:prop:global properties transfer back}, the aligned market-independent excess-return result from Section \ref{subsec:excess return well posedness:market-independent} admits the following terminal-wealth counterpart.

\begin{proposition}
\label{app:prop:global wp original characterization aligned}
Assume that $\mathcal U$ satisfies the upper Fatou property, and that $\mathcal R$ satisfies cash-convexity and the lower Fatou property. Consider the following statements:
\begin{enumerate}[label=\textnormal{(\alph*)}]
    \item For each $c>0$, either $\mathcal U$ satisfies cash $\ASLL$ on $L$ at level $c$ or $\mathcal R$ satisfies cash $\SLL$ on $L$ at level~$c$.
    \item $(\mathcal U,\mathcal R)$-terminal-wealth portfolio selection is market-independent weakly well posed.
    \item For each $c>0$, either $\mathcal U$ satisfies cash $\SLL$ on $L$ at level $c$ or $\mathcal R$ satisfies cash $\SLL$ on $L$ at level~$c$.
    \item $(\mathcal U,\mathcal R)$-terminal-wealth portfolio selection is market-independent well posed.
\end{enumerate}
Then \textnormal{(a)} implies \textnormal{(b)}, and \textnormal{(c)} implies \textnormal{(d)}. If, in addition, $\mathcal R$ is cash aligned with $\mathcal U$ on $L$, then all four statements are equivalent.
\end{proposition}

\begin{proof}
By Proposition \ref{prop:relating U with tilde U}, for every $w\in(0,\infty)$ and $r\in(-1,\infty)$, the functional $\mathcal U_{w,r}$ satisfies the upper Fatou property, while $\mathcal R_{w,r}$ satisfies
positive star-shapedness and the lower Fatou property.

By Proposition \ref{app:prop:global properties transfer back}, statement (a) is equivalent to
\begin{equation*}
\forall\,w\in(0,\infty),\ r\in(-1,\infty),\quad
\bigl[
\mathcal U_{w,r}\text{ satisfies }\ASLL\text{ on }L
\ \text{or}\
\mathcal R_{w,r}\text{ satisfies }\SLL\text{ on }L
\bigr],
\end{equation*}
and statement (c) is equivalent to
\begin{equation*}
\forall\,w\in(0,\infty),\ r\in(-1,\infty),\quad
\bigl[
\mathcal U_{w,r}\text{ satisfies }\SLL\text{ on }L
\ \text{or}\
\mathcal R_{w,r}\text{ satisfies }\SLL\text{ on }L
\bigr].
\end{equation*}
Hence, by Proposition \ref{prop:global wp reparametrized characterization aligned}, statement (a) implies that $(\mathcal U_{w,r},\mathcal R_{w,r})$-excess-return portfolio selection is market-independent weakly well posed for every  $w\in(0,\infty)$ and $r\in(-1,\infty)$, and statement (c) implies that $(\mathcal U_{w,r},\mathcal R_{w,r})$-excess-return portfolio selection is market-independent well posed for every $w\in(0,\infty)$ and $r\in(-1,\infty)$. By Lemma \ref{lemma:global well-posedness}, these are equivalent to (b) and (d), respectively.

If, in addition, for every $c>0$, $\mathcal R$ is cash aligned with $\mathcal U$ on $L$ at level $c$, then Proposition \ref{app:prop:global properties transfer back} implies that $\mathcal R_{w,r}$ is aligned with $\mathcal U_{w,r}$ on $L$ for every $w\in(0,\infty)$ and $r\in(-1,\infty)$. Hence, for each $(w,r)$, all four statements in Proposition \ref{prop:global wp reparametrized characterization aligned} are equivalent. Translating back via Proposition \ref{app:prop:global properties transfer back} and Lemma \ref{lemma:global well-posedness}, we conclude that (a), (b), (c) and (d) are equivalent.
\end{proof}

Proposition \ref{app:prop:global wp original characterization aligned} is the market-independent terminal-wealth counterpart of Proposition \ref{prop:global wp reparametrized characterization aligned}. As in the excess-return setting, it still relies on an alignment assumption. The final step is to show that, in the market-independent setting, this extra assumption can again be removed.

The appropriate terminal-wealth analogue of $\SE$ is the following cash-based version.

\begin{definition}
\label{def:cash sensitivity equivalent}
Let $c > 0$. We say that $\mathcal U$ satisfies \emph{cash $\SE$ at level $c$} if, $\mathcal U$ satisfies cash $\ASLL$ on $L$ at level $c$ if and only if it satisfies cash $\SLL$ on $L$ at level $c$. We say that $\mathcal U$ satisfies \emph{cash $\SE$} if it satisfies cash $\SE$ at level $c$ for every $c > 0$.
\end{definition}

By parts (b) and (c) of Proposition \ref{app:prop:global properties transfer back}, for every $w\in(0,\infty)$ and $r\in(-1,\infty)$, $\mathcal U$ satisfies cash $\SE$ at level $w(1+r)$ if and only if $\mathcal U_{w,r}$ is sensitivity equivalent. Therefore, $\mathcal U$ satisfies cash $\SE$ if and only if $\mathcal U_{w,r}$ is $\SE$ for every $w\in(0,\infty)$ and $r\in(-1,\infty)$. Thus the next theorem is the direct terminal-wealth counterpart of Theorem \ref{thm:global wp reparametrized characterization}.

\begin{theorem}
\label{app:thm:global wp original characterization}
Assume that $\mathcal U$ satisfies the upper Fatou property and cash $\SE$, and that $\mathcal R$ satisfies cash-convexity and the lower Fatou property. Suppose in addition that $\mathcal U$ or $\mathcal R$ is law-invariant. Then the following statements are equivalent:
\begin{enumerate}[label=\textnormal{(\alph*)}]
    \item For each $c>0$, either $\mathcal U$ satisfies cash $\SLL$ on $L$ at level $c$, or $\mathcal R$ satisfies cash $\SLL$ on $L$ at level~$c$.
    \item $(\mathcal U,\mathcal R)$-terminal-wealth portfolio selection is market-independent weakly well posed.
    \item $(\mathcal U,\mathcal R)$-terminal-wealth portfolio selection is market-independent well posed.
\end{enumerate}
\end{theorem}

\begin{proof}
By Proposition \ref{prop:relating U with tilde U}, for every $w\in(0,\infty)$ and $r\in(-1,\infty)$, the functional $\mathcal U_{w,r}$ satisfies the upper Fatou property, while $\mathcal R_{w,r}$ satisfies
positive star-shapedness and the lower Fatou property. Moreover, if $\mathcal U$ or $\mathcal R$ is law-invariant, then so is $\mathcal U_{w,r}$
or $\mathcal R_{w,r}$, respectively.

By the paragraph preceding the theorem, $\mathcal U_{w,r}$ is sensitivity equivalent for every $w\in(0,\infty)$ and $r\in(-1,\infty)$. By
Proposition \ref{app:prop:global properties transfer back}, statement (a) is equivalent to
\begin{equation*}
\forall\,w\in(0,\infty),\ r\in(-1,\infty),\quad
\bigl[
\mathcal U_{w,r}\text{ satisfies }\SLL\text{ on }L
\ \text{or}\
\mathcal R_{w,r}\text{ satisfies }\SLL\text{ on }L
\bigr].
\end{equation*}
Hence, for every $w\in(0,\infty)$ and $r\in(-1,\infty)$, all the assumptions of Theorem \ref{thm:global wp reparametrized characterization} are satisfied. Applying that theorem for every $w$ and $r$, we conclude that statement (a) is equivalent to $(\mathcal U_{w,r},\mathcal R_{w,r})$-excess-return portfolio selection is market-independent weakly well posed for every $w,r$; and also equivalent to $(\mathcal U_{w,r},\mathcal R_{w,r})$-excess-return portfolio selection is market-independent well posed for every $w,r$. By Lemma \ref{lemma:global well-posedness}, these are precisely statements (b) and (c). This proves the equivalence of (a), (b) and (c).
\end{proof}

Theorem \ref{app:thm:global wp original characterization} is the market-independent characterization of terminal-wealth portfolio selection. In contrast to Proposition \ref{app:prop:global wp original characterization aligned}, it no longer relies on any alignment assumption; instead, the utility-side assumption is formulated in terms of cash $\SE$.

\medskip

Under additional \emph{finite-at-zero} assumptions, the cash large-loss conditions introduced above can be rewritten in terms of the \emph{normalized terminal-wealth functionals}
\begin{equation*}
\widehat{\mathcal U}(Y):=\mathcal U(Y)-\mathcal U(0),
\qquad
\widehat{\mathcal R}(Y):=\mathcal R(Y)-\mathcal R(0).
\end{equation*}
The next proposition records the precise relation.

\begin{proposition}
\label{app:prop:cash versus ordinary normalized notions}
Let $c > 0$ and assume that $\mathcal{U}(0), \mathcal{R}(0) \in \mathbb{R}$. Then the following hold:
\begin{enumerate}[label=\textnormal{(\alph*)}]
    \item $\mathcal U$ satisfies cash $\ASLL$ on $L$ at level $c$ if and only if $\widehat{\mathcal U}$ satisfies $\ASLL$ on $L$.
    \item If $\widehat{\mathcal U}$ satisfies $\SE$, then $\mathcal U$ satisfies cash $\SLL$ on $L$ at level $c$ if and only if $\widehat{\mathcal U}$ satisfies $\SLL$ on $L$.
    \item If $\mathcal R$ is cash-convex, then $\mathcal R$ satisfies cash $\SLL$ on $L$ at level $c$ if and only if $\widehat{\mathcal R}$ satisfies $\SLL$ on $L$.
\end{enumerate}
\end{proposition}

\begin{proof}
(a) Fix $c>0$. Suppose first that $\widehat{\mathcal U}$ satisfies $\ASLL$ on $L$. Let $Y\in L$ satisfy $\mathbb P[Y<0]>0$. Choose $\gamma>0$ large enough so that $\mathbb P[Y+c/\gamma<0]>0$,
and set $Z:=Y+c/\gamma$. Then, for every $\lambda\ge \gamma$, $c+\lambda Y\le \lambda Z$. Hence, by monotonicity and $\ASLL$ of $\widehat{\mathcal U}$, $\limsup_{\lambda\to\infty}\widehat{\mathcal U}(c+\lambda Y) \le \limsup_{\lambda\to\infty}\widehat{\mathcal U}(\lambda Z) < \widehat{\mathcal U}(\infty)$.
Since $\widehat{\mathcal U}=\mathcal U-\mathcal U(0)$, this is equivalent to $\limsup_{\lambda\to\infty}\mathcal U(c+\lambda Y)<\mathcal U(\infty)$. Thus $\mathcal U$ satisfies cash $\ASLL$ on $L$ at level $c$.

Conversely, suppose that $\mathcal U$ satisfies cash $\ASLL$ on $L$ at level $c$. Let $Y\in L$ satisfy $\mathbb P[Y<0]>0$. Since $\lambda Y\le c+\lambda Y$ for all $\lambda > 0$, monotonicity yields $\mathcal U(\lambda Y)\le \mathcal U(c+\lambda Y)$ for all $\lambda > 0$. Taking $\limsup$ and using cash $\ASLL$ at level $c$, we obtain $\limsup_{\lambda\to\infty}\mathcal U(\lambda Y)<\mathcal U(\infty)$, which is equivalent to $\widehat{\mathcal U}$ satisfying $\ASLL$ on $L$.

(b) Fix $c>0$, and assume that $\widehat{\mathcal U}$
satisfies $\SE$. Suppose first that $\widehat{\mathcal U}$ satisfies $\SLL$ on $L$. Let $Y\in L$ satisfy $\mathbb P[Y<0]>0$. Choose
$\gamma>0$ large enough so that $\mathbb P[Y+c/\gamma<0]>0,$ and set $Z:=Y+c/\gamma$. Then, for every $\lambda\ge \gamma$, $c+\lambda Y\le \lambda Z$. By monotonicity and $\SLL$ of $\widehat{\mathcal U}$, we have $\widehat{\mathcal U}(c+\lambda Y)\le \widehat{\mathcal U}(\lambda Z)<0$ for all sufficiently large \(\lambda\). Equivalently,
$\mathcal U(c+\lambda Y)<\mathcal U(0)\le \mathcal U(c)$, where the last inequality follows from monotonicity of $\mathcal U$. Thus $\mathcal U$ satisfies cash $\SLL$ on $L$ at level $c$.

Conversely, suppose that $\mathcal U$ satisfies cash $\SLL$ on $L$ at level $c$. Then $\mathcal U$ satisfies cash $\ASLL$ on $L$ at level $c$, and so by part (a), $\widehat{\mathcal U}$ satisfies $\ASLL$ on $L$. Since $\widehat{\mathcal U}$ is $\SE$, it follows that $\widehat{\mathcal U}$ satisfies $\SLL$ on $L$.

(c) Assume that $\mathcal R$ is cash-convex. Define $\mathcal V:= -\widehat{\mathcal R} = \mathcal R(0)-\mathcal R.$ Then $\mathcal V$ is a normalized utility functional, and cash-convexity of $\mathcal R$ is equivalent to cash-concavity of $\mathcal V$. Hence, by Proposition \ref{prop:WSLL equivalent to SLL}, $\mathcal V$ satisfies $\SE$. Fix $c>0$. By part (b) applied to  $\mathcal V$, $\mathcal V$ satisfies cash $\SLL$ on $L$ at level $c$ if and only if $\mathcal V$ satisfies $\SLL$ on $L$. Since $\mathcal R(c+\lambda Y)>\mathcal R(c)$ if and only if $\mathcal V(c+\lambda Y)<\mathcal V(c)$, and $\widehat{\mathcal R}(\lambda Y)>0$ if and only if $\mathcal V(\lambda Y)<0$, this is equivalent to saying that $\mathcal R$ satisfies cash $\SLL$ on $L$ at level $c$ if and only if $\widehat{\mathcal R}$ satisfies $\SLL$ on $L$.
\end{proof}

\section{Additional Results and Proofs}
\label{app:Additional Results and Proofs}

\subsection*{Supplementary Results}

\begin{proposition}
\label{prop:cash-additivity}
If $\mathcal R:L\to(-\infty,\infty]$ is a cash-additive risk functional, then $\mathcal R(0)=0$, and for all $Y\in L$ and $c\in\mathbb R$, $\mathcal R(Y+c)=\mathcal R(Y)+\alpha c$,
where $\alpha:=\mathcal R(1)\in(-\infty,0]$. An analogous result holds for cash-additive utility functionals, but with $\alpha:=\mathcal{U}(1) \in (0,\infty)$.
\end{proposition}

\begin{proof}
Let $c>0$. Since $\mathcal R$ is positive finite, $\mathcal R(c)\in\mathbb R$. Applying cash-additivity with $Y = 0$ gives $\mathcal R(c)=\mathcal R(0)+\mathcal R(c)$. Because $\mathcal R(c)$ is finite, it follows that $\mathcal R(0)=0$.

Next, for every $c>0$, $0=\mathcal R(0)=\mathcal R(c+(-c))=\mathcal R(c)+\mathcal R(-c)$ so $\mathcal R(-c)=-\mathcal R(c)\in\mathbb R$. Hence $\mathbb R\ni c\mapsto \mathcal R(c)$ is real-valued on all of $\mathbb{R}$. Since $\mathcal R$ is cash-additive, the map $c\mapsto \mathcal R(c)$ is additive on $\mathbb R$. Since $\mathcal R$ is decreasing, it is monotone. Therefore, by the standard monotone form of Cauchy's functional equation, there exists $\alpha\in\mathbb R$ such that $\mathcal R(c)=\alpha c$, for all $c\in\mathbb R$. Evaluating at $c=1$ gives $\alpha=\mathcal R(1)$. Finally, since $\mathcal{R}$ is decreasing, $\alpha=\mathcal R(1)\le \mathcal R(0)=0$. Thus $\alpha\in(-\infty,0]$, and $\mathcal R(Y+c)=\mathcal R(Y)+\mathcal R(c)=\mathcal R(Y)+\alpha c$ for all $Y\in L$ and $c\in\mathbb R$.

Finally, if $\mathcal{U}$ is a cash-additive utility functional, then $-\mathcal{U}$ is a cash-additive risk functional. By the above argument, $\mathcal{U}(0) = 0$, and for all $Y \in L$ and $c \in \mathbb{R}$, $\mathcal{U}(Y+c) = \mathcal{U}(Y)+\alpha c$, where $\alpha:=\mathcal{U}(1) \in [0,\infty)$. Since $\mathcal{U}(Y) < \mathcal{U}(\infty)$ for all $Y \in L$, it must be the case that $\mathcal{U}(1) \neq 0$.
\end{proof}

\begin{proposition}
\label{prop:cash-convexity criteria}
If $\mathcal{H}:L\to[-\infty,\infty]$ is cash-additive and negatively/positively star-shaped, then it is cash-concave/convex.
\end{proposition}

\begin{proof}
Let $Y \in L$, $c \in \RR$ and $\lambda \in (0,1)$. If $\mathcal{H}$ is cash-additive and negatively star-shaped, then
\begin{equation*}
    \mathcal{H}(\lambda Y + (1-\lambda)c) = \mathcal{H}(\lambda Y) + \mathcal{H}((1-\lambda) c) \leq \lambda \mathcal{H}(Y) + (1-\lambda) \mathcal{H}(c),
\end{equation*}
where the equality is by cash-additivity and the inequality follows from negative star-shapedness. Thus, $\mathcal{H}$ is cash-concave. A similar argument shows that if $\mathcal{H}$ is cash-additive and positively star-shaped, then it satisfies cash-convexity.
\end{proof}

\begin{proposition}
\label{prop:relating U with tilde U}
Let $\mathcal{U}$ be a utility functional, $w \in (0,\infty)$ and $r \in (-1,\infty)$. Then, $\mathcal{U}_{w,r}$ is a normalized utility functional. Moreover:
\begin{enumerate}[label=\textnormal{(\arabic*)}]
    \item If $\mathcal{U}$ has the upper Fatou property, then so does $\mathcal{U}_{w,r}$.
    \item If $\mathcal{U}$ is law-invariant, then so is $\mathcal{U}_{w,r}$.
    \item If $\mathcal{U}$ is cash-additive, then so is $\mathcal{U}_{w,r}$.
    \item If $\mathcal{U}$ is concave, then so is $\mathcal{U}_{w,r}$.
    \item If $\mathcal{U}$ is cash-concave, then $\mathcal{U}_{w,r}$ is negatively star-shaped.
\end{enumerate}
An analogous inheritance result holds for risk functionals and their $(w,r)$-excess-return functionals, with the lower Fatou property replacing the upper Fatou property, convexity replacing concavity, and positive star-shapedness replacing negative star-shapedness.
\end{proposition}

\begin{proof}
It is not difficult to verify $\mathcal{U}_{w,r}$ is a normalized utility functional. Moreover, we omit the proofs of statements (1)--(4) because they are elementary. That leaves statement (5). So, assume $\mathcal{U}$ is cash-concave. Let $\lambda \in (0,1)$, $Y \in L$ and $Z \coloneqq wY+w(1+r)$. Then, by cash-concavity,
\begin{equation*}
    \mathcal{U}(\lambda Z + (1-\lambda)w(1+r)) \geq \lambda \mathcal{U}(Z) + (1-\lambda)\mathcal{U}(w(1+r)).
\end{equation*}
It follows that
\begin{equation*}
    \mathcal{U}(w(1+r) + w\lambda Y) \geq \lambda \mathcal{U}(wY+w(1+r)) + (1-\lambda) \mathcal{U}(w(1+r)).
\end{equation*}
Rearranging this inequality reveals that $\mathcal{U}_{w,r}(\lambda Y) \geq \lambda \mathcal{U}_{w,r}(Y)$. Whence, $\mathcal{U}_{w,r}$ is negatively star-shaped. The corresponding proof for a risk functional $\mathcal{R}$ follows by applying the above argument to the utility functional $-\mathcal R$ and its $(w,r)$-excess-return utility functional $-\mathcal R_{w,r}$.
\end{proof}

\begin{proposition}
\label{prop:reparameterization equivalent}
Let $\bar S\in\mathbb S(L)$ be a normalized market model with risk-free rate $r\in(-1,\infty)$ and corresponding excess-return vector $X\in\mathbb X(L)$. Fix $w\in(0,\infty)$ and $R_{\max}\in[\mathcal R(w(1+r)),\infty)$, and set $\widetilde R_{\max}:=R_{\max}-\mathcal R(w(1+r))\in[0,\infty)$. For every $\bar\theta=(\theta^0,\ldots,\theta^d)\in\mathbb R^{1+d}$
satisfying $\bar\theta\cdot \bar S_0=w$, define $\pi=(\pi^1,\ldots,\pi^d)\in\mathbb R^d$ by
\begin{equation*}
\pi^i=\frac{\theta^i}{w}, \qquad i=1,\ldots,d.
\end{equation*}
Then $\bar\theta\cdot \bar S_1=w(1+r)+wX_\pi$. Conversely, each $\pi\in\mathbb R^d$ corresponds to $\bar\theta\in\mathbb R^{1+d}$ with $\bar \theta \cdot \bar S_0 = w$, given by
\begin{equation*}
\theta^i=w\pi^i,\quad i=1,\ldots,d,
\qquad
\theta^0=w-\sum_{i=1}^d \theta^i
      =w\Bigl(1-\sum_{i=1}^d \pi^i\Bigr).
\end{equation*}
Under this bijection, the feasible sets of \eqref{eq:maximize utility} and \eqref{eq:maximize utility reparameterized} coincide. Hence the two problems have the same set of optimizers (modulo this correspondence), and their optimal values differ
only by the constant shift
\begin{equation*}
\sup_{\pi \in \mathbb R^d:\,\mathcal{R}_{w,r}(X_\pi) \leq \tilde{R}_{\max}}
\mathcal{U}_{w,r}(X_\pi)
=
\sup_{\bar{\theta} \in \mathbb R^{1+d}:\,\bar{\theta}\cdot \bar{S}_0 = w,\,
\mathcal{R}(\bar{\theta}\cdot \bar{S}_1)\leq R_{\max}}
\mathcal{U}(\bar{\theta}\cdot \bar{S}_1)
-\mathcal{U}(w(1+r)).
\end{equation*}
\end{proposition}

\begin{proof}
Recall that the market is normalized, so $S^i_0=1$ for all $i=0,\ldots,d$, and that
\begin{equation*}
S^0_1=1+r,
\qquad
S^i_1=S^0_1+X^i=(1+r)+X^i,\quad i=1,\ldots,d.
\end{equation*}
Let $\bar\theta=(\theta^0,\ldots,\theta^d)\in\mathbb R^{1+d}$ satisfy
$\bar\theta\cdot \bar S_0=w$, and define $\pi\in\mathbb R^d$ by $\pi^i=\theta^i/w$, $i=1,\ldots,d$. Since $\bar\theta\cdot \bar S_0=w$ and $S^i_0=1$ for all $i$, we have $\sum_{i=0}^d \theta^i=w$. Therefore,
\begin{align*}
\bar\theta\cdot \bar S_1
& =\theta^0(1+r)+\sum_{i=1}^d \theta^i\bigl((1+r)+X^i\bigr)
=(1+r)\sum_{i=0}^d \theta^i+\sum_{i=1}^d \theta^i X^i \\ &
=w(1+r)+w\sum_{i=1}^d \pi^i X^i
=w(1+r)+wX_\pi.
\end{align*}
Substituting this into \eqref{eq:maximize utility}, the problem becomes
\begin{equation*}
\sup_{\pi \in \mathbb R^d} \mathcal U(w(1+r)+wX_\pi)
\quad\text{subject to}\quad
\mathcal R(w(1+r)+wX_\pi)\le R_{\max}.
\end{equation*}
By the definitions of $\mathcal U_{w,r}$ and $\mathcal R_{w,r}$, this is equivalent to
\begin{equation*}
\sup_{\pi \in \mathbb R^d}
\bigl(\mathcal U_{w,r}(X_\pi)+\mathcal U(w(1+r))\bigr)
\quad\text{subject to}\quad
\mathcal R_{w,r}(X_\pi)+\mathcal R(w(1+r))\le R_{\max}.
\end{equation*}
The constraint is exactly $\mathcal R_{w,r}(X_\pi)\le \widetilde R_{\max}$, where $\widetilde R_{\max}=R_{\max}-\mathcal R(w(1+r))$. This proves the claimed identity for the optimal values. The correspondence of feasible sets and optimizers follows from the bijection
\begin{equation*}
\pi^i=\frac{\theta^i}{w},
\qquad
\theta^i=w\pi^i,\ i=1,\ldots,d,
\qquad
\theta^0=w\Bigl(1-\sum_{i=1}^d \pi^i\Bigr). \qedhere
\end{equation*}
\end{proof}

\begin{proposition}
\label{prop:u rho arbitrage implies unbounded sequence}
Let $w \in (0,\infty)$, $r \in (-1,\infty)$ and $X \in \mathbb{X}(L)$. If the sequence $(\pi_n)_{n \geq 1} \subset \mathbb{R}^d$ satisfies $\mathcal{U}_{w,r}(X_{\pi_{n}}) \to \mathcal{U}_{w,r}(\infty)$, then $\lVert \pi_n \rVert \to \infty$.
\end{proposition}

\begin{proof}
Assume $\mathcal{U}_{w,r}(X_{\pi_{n}}) \to \mathcal{U}_{w,r}(\infty)$. Seeking a contradiction, suppose $(\pi_n)_{n \geq 1}$ is bounded. By passing to a subsequence, we may assume that it has a limit $\pi = (\pi^1, \ldots, \pi^d) \in \mathbb{R}^d$. Thus, there exists $N \in \NN$ such that $\lvert \pi^i - \pi^i_n \rvert \leq 1$ for $i \in \{1, \ldots, d \}$ and $n \geq N$. Set
\begin{equation*}
Y^i \coloneqq \lvert X^i \rvert(1 + \lvert \pi^i \rvert), \qquad i \in \{1, \ldots, d \}.
\end{equation*}
Note that $Y^i \in L$ and $Y^i \geq \pi^i_n X^i \ \mathbb{P}$-a.s., for all $n \geq N$. Hence, $Y \coloneqq Y^1 + \ldots +Y^d \in L$ and $Y \geq X_{\pi_n} \ \mathbb{P}$-a.s., for all $n \geq N$. It follows by monotonicity that $\mathcal{U}_{w,r}(Y) \geq \mathcal{U}_{w,r}(X_{\pi_n})$ for all $n \geq N$. But since $\mathcal{U}_{w,r}(X_{\pi_{n}}) \to \mathcal{U}_{w,r}(\infty)$, this implies $\mathcal{U}_{w,r}(Y) \geq \mathcal{U}_{w,r}(\infty)$, and we arrive at a contradiction.
\end{proof}

\begin{proposition} 
\label{prop: (U,R)-arbitrage can be globally centered at 0}
Let $w \in (0,\infty)$, $r \in (-1,\infty)$ and $X \in \mathbb{X}(L)$. Assume $\mathcal{R}_{w,r}$ satisfies the lower Fatou property and positive star-shapedness. If there exists $\tilde{R} \in [0,\infty)$ and an unbounded sequence $( \pi_{n} )_{n \geq 1}$ with $\mathcal{R}_{w,r}( X_{\pi_{n}}) \leq \tilde{R}$ for all $n \geq 1$, then there exists $\pi \in \mathbb{R}^{d} \setminus \{ \mathbf{0}\}$ with 
\begin{equation*}
    \mathcal{R}_{w,r}( \lambda X_{\pi}) \leq 0 \quad \text{for all } \lambda \in (0,\infty).
\end{equation*}
\end{proposition}

\begin{proof}
Set $\alpha_n:=\|\pi_n\|_\infty$ for $n \geq 1$. Since $(\pi_n)_{n\ge1}$ is unbounded and all norms on $\mathbb R^d$ are equivalent, we have $\alpha_n\to\infty$. Define
\begin{equation*}
\hat\pi_n:=\frac{\pi_n}{\alpha_n}\in[-1,1]^d.
\end{equation*}
By compactness of $[-1,1]^d$, after passing to a subsequence we may assume that $\hat\pi_n\to\pi\in[-1,1]^d$. Since $\|\hat\pi_n\|_\infty=1$ for all $n$, continuity of the sup norm yields $\|\pi\|_\infty=1$, and hence $\pi\neq \mathbf 0$. It follows that
\begin{equation*}
X_{\pi_n}/\alpha_n = X_{\hat\pi_n}\to X_\pi
\qquad \mathbb P\text{-a.s.}
\end{equation*}
Fix $\lambda>0$. Since $\alpha_n\to\infty$, there exists $N\in\mathbb N$ such that $\lambda/\alpha_n \in(0,1)$, for all $n\ge N$. By positive star-shapedness of $\mathcal R_{w,r}$,
\begin{equation*}
\mathcal R_{w,r}\!\left(\lambda \frac{X_{\pi_n}}{\alpha_n}\right)
\le
\frac{\lambda}{\alpha_n}\mathcal R_{w,r}(X_{\pi_n})
\le
\frac{\lambda}{\alpha_n}\tilde R,
\qquad\text{for all } n\ge N.
\end{equation*}
Applying the lower Fatou property gives
\begin{equation*}
\mathcal R_{w,r}(\lambda X_\pi)
\le
\liminf_{n\to\infty}
\mathcal R_{w,r}\!\left(\lambda \frac{X_{\pi_n}}{\alpha_n}\right)
\le
\liminf_{n\to\infty}\frac{\lambda}{\alpha_n}\tilde R
=0.
\end{equation*}
Since $\lambda>0$ was arbitrary, the conclusion follows.
\end{proof}

\begin{proposition}
\label{prop:robust and nonrobust global coincide}
The following hold:
\begin{enumerate}[label=\textnormal{(\alph*)}]
    \item $\mathcal U_{w,r}$ satisfies $\ASLL$ on $L$ if and only if, $\mathcal U_{w,r}$ satisfies robust $\ASLL$ on $L$, i.e., for every $Y \in L$ with $\mathbb{P}[Y<0] > 0$ and every $Z \in L$, there exists $\varepsilon_{Y,Z} > 0$ such that 
    \begin{equation*} 
    \limsup_{\lambda \to \infty} \mathcal U_{w,r}\bigl(\lambda(Y+\varepsilon Z)\bigr) < \mathcal U_{w,r}(\infty) \quad \text{for all } \varepsilon \in (0,\varepsilon_{Y,Z}). 
    \end{equation*}
    \item $\mathcal U_{w,r}$ satisfies $\SLL$ on $L$ if and only if, $\mathcal U_{w,r}$ satisfies robust $\SLL$ on $L$, i.e., for every $Y \in L$ with $\mathbb{P}[Y < 0] > 0$ and every $Z \in L$, there exist $\varepsilon_{Y,Z} > 0$ and $\lambda_{Y,Z} \in (0,\infty)$ such that 
    \begin{equation*} 
    \mathcal U_{w,r}\bigl(\lambda (Y+\varepsilon Z)\bigr) < 0, \qquad \text{for all } \lambda \in (\lambda_{Y,Z},\infty) \ \text{and } \varepsilon \in (0,\varepsilon_{Y,Z}). 
    \end{equation*}
\end{enumerate}
\end{proposition}

\begin{proof}
The implications from right to left are immediate by taking $Z=0$. For the converses, let $Y,Z\in L$ with $\mathbb P[Y<0]>0$. Choose $\delta>0$ such that $\mathbb P[Y\le -\delta]>0$.
Since
\begin{equation*}
\{Y\le -\delta\}
=
\bigcup_{m=1}^\infty \bigl(\{Y\le -\delta\}\cap\{|Z|\le m\}\bigr),
\end{equation*}
there exists $m\in\mathbb N$ such that $A:=\{Y\le -\delta\}\cap\{|Z|\le m\}$ has strictly positive probability. Set $\varepsilon_{Y,Z}:=\delta/2m$ and $\widetilde Y:=Y+\varepsilon_{Y,Z}|Z|$. Then on $A$ we have $\widetilde Y\le -\delta+\varepsilon_{Y,Z}m=-\frac{\delta}{2}<0$, and therefore $\mathbb P[\widetilde Y<0]>0$. Moreover, for every $\varepsilon\in(0,\varepsilon_{Y,Z})$,
\begin{equation*}
Y+\varepsilon Z \le Y+\varepsilon |Z|
\le Y+\varepsilon_{Y,Z}|Z|
=\widetilde Y.
\end{equation*}
Since $\mathcal U_{w,r}$ is increasing, it follows that $\mathcal U_{w,r}\bigl(\lambda(Y+\varepsilon Z)\bigr)
\le
\mathcal U_{w,r}(\lambda\widetilde Y)$ for all $\lambda>0$. We now prove statements (a) and (b).

(a) Assume that $\mathcal U_{w,r}$ satisfies $\ASLL$ on $L$. Since $\mathbb P[\widetilde Y<0]>0$, $\limsup_{\lambda\to\infty}\mathcal U_{w,r}(\lambda\widetilde Y) < \mathcal U_{w,r}(\infty)$. Hence, $\limsup_{\lambda\to\infty}\mathcal U_{w,r}\bigl(\lambda(Y+\varepsilon Z)\bigr)
<
\mathcal U_{w,r}(\infty)$ for all $\varepsilon\in(0,\varepsilon_{Y,Z})$, showing that $\mathcal U_{w,r}$ satisfies robust $\ASLL$ on $L$.

(b) Similarly, assume that $\mathcal U_{w,r}$ satisfies $\SLL$ on $L$. Since $\mathbb P[\widetilde Y<0]>0$, there exists $\lambda_{\widetilde Y}\in(0,\infty)$ such that $\mathcal U_{w,r}(\lambda\widetilde Y)<0$ for all $\lambda\in(\lambda_{\widetilde Y},\infty)$. Therefore, $\mathcal U_{w,r}\bigl(\lambda(Y+\varepsilon Z)\bigr)<0$ for all $\lambda\in(\lambda_{\widetilde Y},\infty)$ and $\varepsilon\in(0,\varepsilon_{Y,Z})$, showing that $\mathcal U_{w,r}$ satisfies robust $\SLL$ on $L$.
\end{proof}

\begin{proposition}
\label{prop:WSLL equivalent to SLL}
Let $\mathcal{U}$ be a utility functional. If $\mathcal{U}$ is cash-additive or cash-concave, it satisfies $\SE$. 
\end{proposition}

\begin{proof}



$\SE$ for $\mathcal U$ is equivalent to showing that if $\mathcal U$ does not satisfy $\SLL$, then it does not satisfy $\ASLL$.

Suppose $\mathcal{U}$ is cash-additive. Then, $\mathcal{U}(\infty) = \infty$ and $\mathcal{U}(1) \in (0,\infty)$ by Proposition \ref{prop:cash-additivity}. Assume there exists $Y \in L$ with $\mathbb{P}[Y < 0] > 0$ and a sequence $(\lambda_n) \subset \RR$ with $\lambda_n \to \infty$ such that $\mathcal{U}(\lambda_n Y) \geq 0$ for all $n$. Fix $\epsilon > 0$ so that $\mathbb{P}[Y + \epsilon < 0 ] > 0$ and let $Z  \coloneqq Y + \epsilon$. Then, $\mathbb{P}[Z < 0] > 0$ and  by Proposition \ref{prop:cash-additivity},
\begin{equation*}
    \mathcal{U}(\lambda_n Z) = \mathcal{U}(\lambda_n Y + \lambda_n \epsilon) = \mathcal{U}(\lambda_nY) + \mathcal{U}(1)\lambda_n\epsilon \to \infty.
\end{equation*}
Whence, $\mathcal{U}$ satisfies $\SE$.

Now suppose $\mathcal{U}$ is cash-concave, and hence, negatively star-shaped. Assume there exists $Y \in L$ with $\mathbb{P}[Y < 0] > 0$ and a sequence $(\lambda_n) \subset \RR$ with $\lambda_n \to \infty$ such that $\mathcal{U}(\lambda_n Y) \geq 0$ for all $n$. By negative star-shapedness of $\mathcal{U}$, $\mathcal{U}(\lambda Y) \geq 0$ for all $\lambda  \in (0,\infty)$. Let $\epsilon > 0$ be such that $\P[Y + \epsilon < 0] > 0$ and set $Z \coloneqq Y+\epsilon$. Then, by the cash-concavity and monotonicity of $\mathcal{U}$, for any $\eta \in (0,1)$,
\begin{equation*}
\mathcal{U}(\lambda Z) = \mathcal{U}(\eta \lambda \epsilon + (1-\eta)(\tfrac{\lambda}{1-\eta} Y + \lambda\epsilon)) \geq \eta \mathcal{U} (\lambda \epsilon) + (1-\eta) \mathcal{U}(\tfrac{\lambda}{1-\eta} Y) \geq \eta \mathcal{U}(\lambda \epsilon).
\end{equation*}
Therefore, $\lim_{\lambda \to \infty} \mathcal{U}(\lambda Z) \geq \lim_{\eta \to 1} \eta \mathcal{U}(\infty) = \mathcal{U}(\infty)$. Whence, $\mathcal{U}$ satisfies $\SE$.  
\end{proof}

\subsection*{Proofs in Section \ref{sec:problem formulation}}

\begin{proof}[\textbf{Proof of Proposition \ref{prop:terminal-wealth uniqueness}}]
The feasible set $
\{\bar\theta\in\mathbb R^{1+d}:\ \bar\theta\cdot\bar S_0=w,\ 
\mathcal R(\bar\theta\cdot\bar S_1)\le R_{\max}\}$
is convex, since the budget constraint is affine and the risk constraint is a sublevel set of the quasi-convex map $\bar\theta\mapsto \mathcal R(\bar\theta\cdot\bar S_1)$. Moreover, by non-redundancy of $\bar S$, the map $\bar\theta\mapsto \bar\theta\cdot\bar S_1$ is injective, and therefore $\bar\theta\mapsto \mathcal U(\bar\theta\cdot\bar S_1)$ is strictly quasi-concave on $\mathbb R^{1+d}$. Hence there cannot exist two distinct maximizers on the feasible set. Therefore the maximizer is unique.
\end{proof}

\begin{proof}[\textbf{Proof of Proposition \ref{prop:excess-return uniqueness}}]
The proof is almost identical to that of Proposition \ref{prop:terminal-wealth uniqueness}.
\end{proof}

\begin{proof}[\textbf{Proof of Lemma \ref{lem:local equivalence reparametrization}}]
This follows directly from Proposition~\ref{prop:reparameterization equivalent}, applied for each fixed $w \in (0,\infty)$, noting that for a fixed market $\bar S$ the corresponding risk-free rate $r$ and excess-return vector $X$ are fixed.
\end{proof}

\begin{proof}[\textbf{Proof of Lemma \ref{lemma:global well-posedness}}]
This follows immediately from Lemma~\ref{lem:local equivalence reparametrization} and the bijection between normalized markets $\bar S \in \mathbb S(L)$ and pairs $(X,r)\in \mathbb X(L)\times(-1,\infty)$.
\end{proof}

\subsection*{Proofs in Section \ref{sec:excess return well posedness}}

\begin{proof}[\textbf{Proof of Proposition \ref{prop:unbounded sequence of acceptable portfolios}}]
By Proposition \ref{prop:u rho arbitrage implies unbounded sequence}, the sequence $(\pi_n)_{n \geq 1}$ is unbounded. Set $\lambda_n:=\|\pi_n\|$ and $\tilde\pi_n:=\pi_n/\lambda_n$, and, by passing to a subsequence if necessary, assume that $\lambda_n\to\infty$ and $\tilde\pi_n\to\pi$ for some $\pi\in\mathbb R^d$. Since $\|\tilde\pi_n\|=1$ for all $n$, we have $\pi\neq \mathbf 0$. Set $
Y:=\sum_{i=1}^d |X^i| \in L$.
Fix $\varepsilon>0$. By passing to a further subsequence if necessary, we may assume that $|\tilde\pi_n^i-\pi^i|\le \varepsilon$ for all $i\in\{1,\dots,d\}$ and all $n$. Then $\lambda_n(X_\pi+\varepsilon Y)\ge X_{\pi_n}$ $\mathbb{P}$-a.s., for all $n$. By monotonicity of $\mathcal U_{w,r}$, it follows that
\[
\mathcal U_{w,r}\bigl(\lambda_n(X_\pi+\varepsilon Y)\bigr)\ge \mathcal U_{w,r}(X_{\pi_n})
\to \mathcal U_{w,r}(\infty).
\]
Since $\mathcal U_{w,r}\bigl(\lambda_n(X_\pi+\varepsilon Y)\bigr)\le \mathcal U_{w,r}(\infty)$, the convergence follows.
\end{proof}

\begin{proof}[\textbf{Proof of Proposition \ref{prop:local SLL implies no (U,R) arbitrage}}]
We argue by contraposition.

(a) Suppose that $(\mathcal U_{w,r},\mathcal R_{w,r})$-portfolio selection is not weakly well posed for $X$. Then there exists a sequence $(\pi_n)_{n\ge1}\subset\mathbb R^d$ such that $\mathcal U_{w,r}(X_{\pi_n})\to\mathcal U_{w,r}(\infty)$. By Proposition~\ref{prop:unbounded sequence of acceptable portfolios}, there exist $\pi\in\mathbb R^d\setminus\{\mathbf 0\}$ and $Y\in L$ such that, for every $\varepsilon>0$, there exists a sequence $(\lambda_n)_{n\ge1}$ with $\lambda_n\to\infty$ and $\mathcal U_{w,r}\bigl(\lambda_n(X_\pi+\varepsilon Y)\bigr)\to \mathcal U_{w,r}(\infty)$. Hence $\limsup_{\lambda\to\infty}\mathcal U_{w,r}\bigl(\lambda(X_\pi+\varepsilon Y)\bigr)
=
\mathcal U_{w,r}(\infty)$ for every $\varepsilon>0$, showing that $\mathcal U_{w,r}$ does not satisfy robust $\ASLL$ on $X$.

(b) Suppose in addition that $\mathcal R_{w,r}$ satisfies the lower Fatou property and positive star-shapedness, and that $(\mathcal U_{w,r},\mathcal R_{w,r})$-portfolio selection is not weakly well posed for $X$. Then Proposition~\ref{prop:U R arbitrage implies not SLL for R} yields $\pi\in\mathbb R^d\setminus\{\mathbf 0\}$ such that $\mathcal R_{w,r}(\lambda X_\pi)\le 0$, for all $\lambda>0$. Hence $\mathcal R_{w,r}$ does not satisfy $\SLL$ on $X$.
\end{proof}

\begin{proof}[\textbf{Proof of Proposition \ref{prop:sufficient local well posedness}}]
Let $X \in \mathbb{X}(L)$ and fix $\tilde{R}_{\max} \in [0,\infty)$. We first show via contradiction that, under (a) or (b), any maximizing sequence $(\pi_n)_{n \geq 1} \subset \mathbb{R}^d$ of \eqref{eq:maximize utility reparameterized} must be bounded.

(a) If $\mathcal{U}_{w,r}$ satisfies robust $\SLL$ on $X$ and $(\pi_n)_{n \geq 1}$ is unbounded, then the argument in Proposition \ref{prop:unbounded sequence of acceptable portfolios} (which we do not repeat) implies the existence of $\pi \in \RR^d \setminus \{ \mathbf{0} \}$ and $Y \in L$ such that, for every $\epsilon > 0$, there is a sequence $(\lambda_n)_{n \geq 1} \subset (0,\infty)$ with $\lambda_n \to \infty$ and $\mathcal{U}_{w,r}(\lambda_n (X_\pi + \epsilon Y)) \geq \mathcal{U}_{w,r}(X_{\pi_n})$ for all $n$. Since $(\pi_n)_{n \geq 1}$ is a maximizing sequence, $\mathcal{U}_{w,r}(X_{\pi_n}) \geq \mathcal{U}_{w,r}(X_{\mathbf{0}}) = \mathcal{U}_{w,r}(0) = 0$. Whence, $\mathcal{U}_{w,r}(\lambda_n (X_\pi + \epsilon Y)) \geq0$ for all $n$, in contradiction to our assumption that $\mathcal{U}_{w,r}$ satisfies robust $\SLL$ on $X$.

(b) If $\mathcal{R}_{w,r}$ is positively star-shaped and satisfies $\SLL$ on $X$, and $(\pi_n)_{n \geq 1}$ is unbounded, then it follows from Proposition \ref{prop: (U,R)-arbitrage can be globally centered at 0} that there exists a portfolio $\pi \in \RR^d \setminus \{ \mathbf{0} \}$ with $\mathcal{R}_{w,r}(\lambda X_\pi) \leq 0$ for all $\lambda \in (0,\infty)$ because $\mathcal{R}_{w,r}(X_{\pi_n}) \leq \tilde{R}_{\max}$ for all $n \geq 1$. But this contradicts that $\mathcal{R}_{w,r}$ satisfies $\SLL$ on $X$.

Since $(\pi_n)_{n \geq 1}$ is bounded in either case, by restricting to a subsequence and relabelling, we may assume that this sequence converges to a limit $\pi^* \in \mathbb{R}^d$. It follows from the upper Fatou property of $\mathcal{U}_{w,r}$ and the lower Fatou property of $\mathcal{R}_{w,r}$ that $\pi^*$ is a maximizer to \eqref{eq:maximize utility reparameterized}. Therefore, $(\mathcal{U}_{w,r},\mathcal{R}_{w,r})$-portfolio selection is well posed for $X$.
\end{proof}

\begin{proof}[\textbf{Proof of Theorem \ref{thm:local weak wp reparametrized characterization}}]
That (a) implies (b) is exactly Proposition~\ref{prop:local SLL implies no (U,R) arbitrage}, while (c) implies (d) is exactly Proposition \ref{prop:sufficient local well posedness}. 

Now assume in addition that $\mathcal R_{w,r}$ is aligned with $\mathcal U_{w,r}$ on $X$. By the above implications, and since (d) implies (b) and (c) implies (a) trivially, in order to prove all four statements are equivalent, it suffices to show that (b) implies (c). 

We argue by contraposition. Suppose that $\mathcal U_{w,r}$ does not satisfy robust $\SLL$ on $X$ and that $\mathcal R_{w,r}$ does not satisfy $\SLL$ on $X$. Then, by Definition \ref{def: ssl local} and positive star-shapedness of $\mathcal R_{w,r}$, there exists $\pi \in \mathbb R^d \setminus \{\mathbf 0\}$ such that $\mathcal R_{w,r}(\lambda X_\pi)\le 0$ for all $\lambda > 0$. Since $\mathcal R_{w,r}$ is aligned with $\mathcal U_{w,r}$ on $X$, it follows that $\lim_{\lambda\to\infty}\mathcal U_{w,r}(\lambda X_\pi)=\mathcal U_{w,r}(\infty)$. Hence the sequence of portfolios $(\lambda\pi)_{\lambda>0}$ satisfies $\mathcal R_{w,r}(X_{\lambda\pi})=\mathcal R_{w,r}(\lambda X_\pi)\le 0$ for all $\lambda > 0$, while 
$\mathcal U_{w,r}(X_{\lambda\pi})=\mathcal U_{w,r}(\lambda X_\pi)\to \mathcal U_{w,r}(\infty)$. Therefore $(\mathcal U_{w,r},\mathcal R_{w,r})$-portfolio selection is not weakly well posed for $X$, so (b) fails.
\end{proof}

\begin{proof}[\textbf{Proof of Proposition \ref{prop:global wp reparametrized characterization aligned}}]
We first prove the forward implications. 

Assume (a) holds. If $\mathcal U_{w,r}$ satisfies $\ASLL$ on $L$, then Proposition \ref{prop:robust and nonrobust global coincide} shows that $\mathcal U_{w,r}$ automatically satisfies robust $\ASLL$ on $L$. Hence, for every market $X\in\mathbb X(L)$, $\mathcal U_{w,r}$ satisfies robust $\ASLL$ on $X$. If instead $\mathcal R_{w,r}$ satisfies $\SLL$ on $L$, then in particular it satisfies $\SLL$ on every market $X\in\mathbb X(L)$. Thus, in either case, Theorem \ref{thm:local weak wp reparametrized characterization} implies that $(\mathcal U_{w,r},\mathcal R_{w,r})$-portfolio selection is weakly well posed for every market $X\in\mathbb X(L)$, i.e., (b) holds. The proof that (c) implies (d) is analogous.

Now assume in addition that $\mathcal R_{w,r}$ is aligned with $\mathcal U_{w,r}$ on $L$. By the above implications, and since (d) implies (b) and (c) implies (a) trivially, in order to prove all four statements are equivalent, it suffices to show that (b) implies (c). 

We argue by contraposition. Suppose that $\mathcal U_{w,r}$ does not satisfy $\SLL$ on $L$ and that $\mathcal R_{w,r}$ does not satisfy $\SLL$ on $L$. Since $\mathcal R_{w,r}$ is positively star-shaped, the failure of $\SLL$ yields some $Y\in L$ with $\mathbb P[Y<0]>0$ such that $\mathcal R_{w,r}(\lambda Y)\le 0$ for all $\lambda>0$. By alignment on $L$, it follows that $\lim_{\lambda\to\infty}\mathcal U_{w,r}(\lambda Y)=\mathcal U_{w,r}(\infty)$. In particular, $Y$ cannot satisfy $Y\le 0$ $\mathbb P$-a.s., because otherwise monotonicity would give $\mathcal U_{w,r}(\lambda Y)\le \mathcal U_{w,r}(0)=0<\mathcal U_{w,r}(\infty)$ for all $\lambda>0$. Hence $\mathbb P[Y>0]>0$, and so the one-dimensional market $X:=(Y)$ belongs to $\mathbb X(L)$. Moreover, the sequence of portfolios $(\lambda_n)_{n\in\mathbb N}$ with $\lambda_n\uparrow\infty$ yields
\[
\mathcal U_{w,r}(\lambda_n Y)\to \mathcal U_{w,r}(\infty)
\qquad\text{and}\qquad
\mathcal R_{w,r}(\lambda_n Y)\le 0
\quad \text{for all } n.
\]
Thus $X$ admits $(\mathcal U_{w,r},\mathcal R_{w,r})$-arbitrage, so $(\mathcal U_{w,r},\mathcal R_{w,r})$-portfolio selection is not weakly well posed for $X$. Therefore (b) fails.
\end{proof}

\begin{proof}[\textbf{Proof of Theorem \ref{thm:global wp reparametrized characterization}}]
The implications (a) $\Rightarrow$ (b) and (c) $\Rightarrow$ (d) follow from Proposition~\ref{prop:global wp reparametrized characterization aligned}. Since $\mathcal U_{w,r}$ satisfies $\SE$, statements (a) and (c) are equivalent. Moreover, (d) $\Rightarrow$ (b) is immediate. It therefore suffices to prove (b) $\Rightarrow$ (a).  This is intricate carried out below.
\end{proof}


\subsection*{Proof of Theorem \ref{thm:global wp reparametrized characterization} [(b) $\Rightarrow$ (a)]}


Let $w \in (0,\infty)$ and $r \in (-1,\infty)$. We begin by introducing the sets that capture the failure of $\ASLL$ and $\SLL$. For $c \in [0,\mathcal U_{w,r}(\infty))$, define
\begin{equation*}
\mathcal S^c_{\mathcal U_{w,r}}
:=
\Bigl\{
Y \in L \ \Big| \ \mathbb P[Y<0]>0 \ \textnormal{and } \exists \  (\lambda_n) \subset (0,\infty)
\textnormal{ with } \lambda_n \uparrow \infty \textnormal{ and } \mathcal U_{w,r}(\lambda_n Y) \ge c \ \forall \ n
\Bigr\},
\end{equation*}
and set $\overline{\mathcal S}_{\mathcal U_{w,r}}
:=
\bigcap_{c\in[0,\mathcal U_{w,r}(\infty))}\mathcal S^c_{\mathcal U_{w,r}}$. 

Analogously, for $c \in (\mathcal R_{w,r}(\infty),0]$ (with the convention $(0,0]:=\{0\}$), define
\begin{equation*}
\mathcal S^c_{\mathcal R_{w,r}}
:=
\Bigl\{
Y \in L \ \Big| \ \mathbb P[Y<0]>0 \ \textnormal{and } \exists \ (\lambda_n) \subset (0,\infty)
\textnormal{ with } \lambda_n \uparrow \infty \textnormal{ and } \mathcal R_{w,r}(\lambda_n Y) \le c \ \forall \ n
\Bigr\},
\end{equation*}
and set $
\overline{\mathcal S}_{\mathcal R_{w,r}}
:=
\bigcap_{c\in(\mathcal R_{w,r}(\infty),0]}\mathcal S^c_{\mathcal R_{w,r}}$.

This yields the following characterizations of $\ASLL$ and $\SLL$ for $\mathcal{U}_{w,r}$ on $L$. A corresponding result holds for $\mathcal{R}_{w,r}$.

\begin{proposition}
\label{prop: sensitivity U}
We have the following:
\begin{enumerate}[label=\textnormal{(\alph*)}]
    \item $\mathcal U_{w,r}$ satisfies $\SLL$ on $L$ if and only if $\mathcal S^0_{\mathcal U_{w,r}}=\emptyset$.
    \item $\mathcal U_{w,r}$ satisfies $\ASLL$ on $L$ if and only if $\overline{\mathcal S}_{\mathcal U_{w,r}}=\emptyset$.
\end{enumerate}
\end{proposition}

The next two propositions provide the core of the proof of (b) $\Rightarrow$ (a) in Theorem~\ref{thm:global wp reparametrized characterization}. The first shows that whenever one can find a single payoff on which risk remains ineffective and utility approaches its bliss point, one can construct a market admitting $(\mathcal U_{w,r},\mathcal R_{w,r})$-arbitrage. The second uses law-invariance to combine such payoffs.

\begin{proposition}
\label{prop:existence of UR arbitrage}
Assume $\mathcal R_{w,r}$ is positively star-shaped. If $\mathcal S^0_{\mathcal R_{w,r}} \cap \overline{\mathcal S}_{\mathcal U_{w,r}} \neq \emptyset$, then there exists a market $X \in \mathbb X(L)$ that admits $(\mathcal U_{w,r},\mathcal R_{w,r})$-arbitrage.
\end{proposition}

\begin{proof}
First note that, since $\mathcal R_{w,r}$ is positively star-shaped, if there exist $Y \in L$ and $\lambda \in (0,\infty)$ such that $\mathcal R_{w,r}(\lambda Y)\le 0$, then $\mathcal R_{w,r}(\lambda'Y)\le 0$ for all $\lambda'\in(0,\lambda)$. Therefore,
\begin{equation*}
\mathcal S^0_{\mathcal R_{w,r}}
=
\Bigl\{
Y\in L \ \Big| \ \mathbb P[Y<0]>0
\ \textnormal{and}\
\mathcal R_{w,r}(\lambda Y)\le 0 \ \textnormal{for all }\lambda>0
\Bigr\}.
\end{equation*}

Now let $X^1 \in \mathcal S^0_{\mathcal R_{w,r}} \cap \overline{\mathcal S}_{\mathcal U_{w,r}}$. Choose a sequence $(c_n)_{n\in\mathbb N}$ such that $c_n\uparrow \mathcal U_{w,r}(\infty)$ and $c_n<\mathcal U_{w,r}(\infty)$ for every $n$. Since $X^1 \in \mathcal S^{c_n}_{\mathcal U_{w,r}}$ for every $n$, we can recursively choose $\lambda_n\in(0,\infty)$ so that $\lambda_n>\max\{\lambda_{n-1},n\}$ and $\mathcal U_{w,r}(\lambda_n X^1)\ge c_n$ for all $n\in\mathbb N$. Then $\lambda_n\uparrow\infty$ and
$\mathcal U_{w,r}(\lambda_n X^1)\to \mathcal U_{w,r}(\infty)$.
Moreover, since $X^1\in \mathcal S^0_{\mathcal R_{w,r}}$,
$
\mathcal R_{w,r}(\lambda_n X^1)\le 0
$ for all $n\in\mathbb N$.
By definition, $\mathbb P[X^1<0]>0$. We also have $\mathbb P[X^1>0]>0$, for otherwise $X^1\le 0$ $\mathbb P$-a.s., and monotonicity would imply $\mathcal U_{w,r}(\lambda_nX^1)\le0$ for all $n$, contradicting $\mathcal U_{w,r}(\lambda_nX^1)\to \mathcal U_{w,r}(\infty)$. Hence the market $(X^1)=:X \in \mathbb{X}(L)$ admits $(\mathcal U_{w,r},\mathcal R_{w,r})$-arbitrage.
\end{proof}

\begin{proposition}
\label{prop:domination by survival functions}
Assume at least one of $\mathcal U_{w,r}$ and $\mathcal R_{w,r}$ is law-invariant. If $\mathcal S^0_{\mathcal R_{w,r}} \neq \emptyset$ and $\overline{\mathcal S}_{\mathcal U_{w,r}} \neq \emptyset$, then $\mathcal S^0_{\mathcal R_{w,r}} \cap \overline{\mathcal S}_{\mathcal U_{w,r}} \neq \emptyset$.
\end{proposition}

\begin{proof}
We treat the case where $\mathcal U_{w,r}$ is law-invariant; the case where $\mathcal R_{w,r}$ is law-invariant is analogous.

Let $X \in \mathcal S^0_{\mathcal R_{w,r}}$ and $Y \in \overline{\mathcal S}_{\mathcal U_{w,r}}$. Since $(\Omega,\mathcal F,\mathbb P)$ is atomless, \cite[Lemma A.32]{follmerschied:2016} yields a uniformly distributed random variable $U$ on $(0,1)$ such that $X = F_X^{\leftarrow}(U)$ $\mathbb{P}$-a.s., where $F_X^{\leftarrow}$ denotes the left-quantile function of $X$. Define $\widetilde Y := F_Y^{\leftarrow}(U)$. Then $\widetilde Y \stackrel{d}{=} Y$, so $\widetilde Y \in L$ since $L$ is law-invariant. Moreover, $X$ and $\widetilde Y$ are comonotone, as both are increasing functions of the same random variable $U$. Set $Z := X \vee \widetilde Y$. We will show $Z \in \mathcal S^0_{\mathcal R_{w,r}} \cap \overline{\mathcal S}_{\mathcal U_{w,r}}$.

Since $L$ is a Riesz space, we have $Z \in L$. As $Z \ge X$ $\mathbb P$-a.s.\ and $\mathcal R_{w,r}$ is decreasing, it follows that $\mathcal R_{w,r}(\lambda Z) \le \mathcal R_{w,r}(\lambda X) \le 0$ for all $\lambda>0$. Also, since $Z \ge \widetilde Y$ $\mathbb P$-a.s., monotonicity and law-invariance of $\mathcal U_{w,r}$ give $\mathcal U_{w,r}(\lambda Z) \ge \mathcal U_{w,r}(\lambda \widetilde Y)
= \mathcal U_{w,r}(\lambda Y)$ for all $\lambda > 0$. Finally, since $\mathbb P[X<0]>0$ and $\mathbb P[Y<0]>0$, there exists $\delta>0$ such that $F_X^{\leftarrow}(u)<0$ and $F_Y^{\leftarrow}(u)<0$, for all $u\in(0,\delta)$. Therefore, $Z = F_X^{\leftarrow}(U)\vee F_Y^{\leftarrow}(U) <0$ on $\{U\in(0,\delta)\}$, so $\mathbb P[Z<0]>0$. We conclude that $Z \in \mathcal S^0_{\mathcal R_{w,r}} \cap \overline{\mathcal S}_{\mathcal U_{w,r}}$.
\end{proof}

We now prove (b) $\Rightarrow$ (a) in Theorem \ref{thm:global wp reparametrized characterization}.  We argue by contraposition. Suppose that $\mathcal U_{w,r}$ does not satisfy $\ASLL$ on $L$ and that $\mathcal R_{w,r}$ does not satisfy $\SLL$ on $L$. Then Proposition \ref{prop: sensitivity U}(b) yields $\overline{\mathcal S}_{\mathcal U_{w,r}} \neq \emptyset$. Likewise, by the corresponding characterization for $\mathcal R_{w,r}$, it follows that $\mathcal S^0_{\mathcal R_{w,r}} \neq \emptyset$. Since at least one of $\mathcal U_{w,r}$ and $\mathcal R_{w,r}$ is law-invariant, Proposition~\ref{prop:domination by survival functions} implies
$
\mathcal S^0_{\mathcal R_{w,r}} \cap \overline{\mathcal S}_{\mathcal U_{w,r}} \neq \emptyset.
$
Hence, by Proposition~\ref{prop:existence of UR arbitrage}, there exists a market $X \in \mathbb X(L)$ that admits $(\mathcal U_{w,r},\mathcal R_{w,r})$-arbitrage. Therefore $(\mathcal U_{w,r},\mathcal R_{w,r})$-portfolio selection is not weakly well posed for this market $X$, and so it is not market-independent weakly well posed. This proves the contrapositive of \textnormal{(b)} $\Rightarrow$ \textnormal{(a)}.

\subsection*{Proofs in Section \ref{sec:terminal wealth well posedness}}

\begin{proof}[\textbf{Proof of Theorem \ref{thm:local wp original characterization}}]
By cash-additivity and normalization of $\mathcal R$, for every $c>0$ and every zero-cost payoff $Z=\bar\eta\cdot\bar S_1$,
\begin{equation*}
\mathcal R(c+\lambda Z)>\mathcal R(c)
\iff
\mathcal R(\lambda Z)>0.
\end{equation*}
Hence $\mathcal R$ satisfies $\SLL$ on $\bar S$ if and only if it satisfies cash $\SLL$ on $\bar S$ at every level $c>0$.

Next, assume that $\mathcal U$ satisfies robust $\ASLL$ on $\bar S$, and fix $c>0$. Let $\bar\eta$ be a non-zero zero-cost portfolio and $Y\in L$. Set $\widetilde Y:=Y+1$. Since $\lambda\varepsilon>c$ for all sufficiently large $\lambda$,
\begin{equation*}
c+\lambda(\bar\eta\cdot\bar S_1+\varepsilon Y)
\le
\lambda(\bar\eta\cdot\bar S_1+\varepsilon\widetilde Y),
\end{equation*}
and therefore, by monotonicity,
\begin{equation*}
\limsup_{\lambda\to\infty}
\mathcal U\bigl(c+\lambda(\bar\eta\cdot\bar S_1+\varepsilon Y)\bigr)
\le
\limsup_{\lambda\to\infty}
\mathcal U\bigl(\lambda(\bar\eta\cdot\bar S_1+\varepsilon\widetilde Y)\bigr)
<
\mathcal U(\infty)
\end{equation*}
for all sufficiently small $\varepsilon>0$. Thus robust $\ASLL$ on $\bar S$ implies robust cash $\ASLL$ on $\bar S$ at every level $c>0$. Similarly robust $\SLL$ on $\bar S$ implies robust cash $\SLL$ on $\bar S$ at every level $c>0$.

Finally, if $\mathcal R$ is aligned with $\mathcal U$ on $\bar S$, then cash-additivity of $\mathcal R$ gives
\begin{equation*}
\mathcal R(c+\lambda Z)\le \mathcal R(c)
\iff
\mathcal R(\lambda Z)\le 0,
\end{equation*}
and therefore alignment on $\bar S$ implies cash alignment on $\bar S$ at every level $c>0$.

The implications (a) $\Rightarrow$ (b) and (c) $\Rightarrow$ (d) now follow from Theorem \ref{app:thm:local wp original characterization} in Appendix \ref{app:terminal-wealth portfolio selection}.

Now assume in addition that $\mathcal R$ is aligned with $\mathcal U$ on $\bar S$. If $\mathcal R$ satisfies $\SLL$ on $\bar S$, then (a) and (c) hold trivially, and therefore so do (b) and (d). Suppose instead that $\mathcal R$ does not satisfy $\SLL$ on $\bar S$. Then by positive star-shapedness, there exists a non-zero zero-cost portfolio $\bar\eta$ such that $
\mathcal R(\lambda\,\bar\eta\cdot\bar S_1)\le 0$
for all $\lambda>0$. By alignment,
$
\lim_{\lambda\to\infty}
\mathcal U(\lambda\,\bar\eta\cdot\bar S_1)
=
\mathcal U(\infty).$ Taking $Y=0$, this shows that $\mathcal U$ satisfies neither robust $\ASLL$ on $\bar S$ nor robust $\SLL$ on $\bar S$. Hence (a) and (c) both fail. Since alignment implies cash alignment at every level $c>0$, the converse part of Theorem \ref{app:thm:local wp original characterization} in Appendix \ref{app:terminal-wealth portfolio selection} yields that (b) and (d) also fail. Therefore all four statements are equivalent.
\end{proof}

\begin{proof}[\textbf{Proof of Theorem \ref{thm:global wp original characterization normalized}}]
By Appendix \ref{app:terminal wealth market-independent well-posedness}, the general market-independent terminal-wealth theory is formulated in terms of cash $\SLL$ and cash $\SE$ at every strictly positive level $c>0$. Under normalization, however, Proposition \ref{app:prop:cash versus ordinary normalized notions} shows that, for every $c>0$, the cash $\SLL$ conditions for $\mathcal U$ and $\mathcal R$ are equivalent to the ordinary $\SLL$ conditions, and that $\mathcal U$ satisfying $\SE$ implies cash $\SE$. Hence statement (a) is exactly the simplified version of statement (a) in Theorem \ref{app:thm:global wp original characterization}. The claimed equivalence of (a), (b) and (c) therefore follows from that theorem.
\end{proof}

\subsection*{Proofs in Section \ref{sec:Examples}} 

\begin{proof}[\textbf{Proof of Proposition \ref{prop:expected utility usc}}]
The growth condition on $u$ and its monotonicity imply that there exist
$a,b\in[0,\infty)$ with
$u(y)\le a y^+ +b$ for all $y\in\mathbb R$.
Let $(Y_n)_{n\ge1}\subset L$ and $Y\in L$ satisfy $Y_n\to Y$
$\mathbb P$-a.s.\ and $|Y_n|\le Z$ $\mathbb P$-a.s.\ for some
$Z\in L$. Upper semicontinuity gives $u(Y)\ge\limsup_{n\to\infty}u(Y_n)$ $\mathbb{P}$-a.s. Moreover, $u(Y_n)\le aZ+b$ $\mathbb{P}$-a.s., where $aZ+b$ is integrable. The reverse Fatou lemma therefore yields
\begin{equation*}
\mathbb E[u(Y)] \geq \mathbb E\left[\limsup_{n\to\infty}u(Y_n)\right] \geq \limsup_{n\to\infty}\mathbb E[u(Y_n)].
\end{equation*}
Hence $\mathcal E_u$ satisfies the upper Fatou property.
\end{proof}

\begin{proof}[\textbf{Proof of Proposition \ref{prop:expected utility minus infinity at zero}}]
Fix $c>0$ and $Y\in L$ with $\mathbb P[Y<0]>0$. There exists $\delta>0$ such that $\mathbb P[Y\le-\delta]>0$. Since $u$ is increasing and $u(0)=-\infty$, we have $u(y)=-\infty$ for every $y\le0$. Consequently, whenever $\lambda>c/\delta$, $u(c+\lambda Y)=-\infty$ on a set of positive probability, and therefore
\begin{equation*}
\mathcal E_u(c+\lambda Y) = -\infty < u(c) = \mathcal E_u(c).
\end{equation*}
Thus $\mathcal E_u$ satisfies cash $\SLL$ on $L$ at every level $c>0$. The same argument shows that it satisfies cash $\ASLL$, and hence cash $\SE$. Moreover, $\mathcal E_u$ is law-invariant and, by Proposition \ref{prop:expected utility usc}, satisfies the upper Fatou property. The result now follows from Theorem \ref{app:thm:global wp original characterization}.
\end{proof}

\begin{proof}[Proof of Proposition \ref{prop:expected utility sensitivity equivalent}]
Suppose there exists $Y \in L$ with $\mathbb{P}[Y < 0] > 0$ and a sequence $(\gamma_n)_{n \geq 1} \subset \RR$ with $\gamma_n \to \infty$ such that $\cE_u(\gamma_n Y) \geq 0$ for all $n$. Then, $\ell \coloneqq \limsup_{x\to\infty}u(-x)/u(x) > -\infty$ by \citet[\emph{Theorem 5.1}]{HKM2024}, and by normalization, $\ell \leq 0$. This implies $\lim_{y \to \infty} u(y) = \infty$, for otherwise $\lim_{y \to \infty} u(y) \in (0,\infty)$, and in this case, the only way $\ell \in (-\infty,0]$ is if $\lim_{y \to -\infty} u(-y) \in (-\infty,0]$, but this would contradict the unboundedness of $u$.  Now by the definition of $\ell$, there exists a sequence $(\lambda_n)_{n \geq 1}$ with $\lambda_n  \to \infty$ such that $u(-\lambda_n)\ge(\ell-1)u(\lambda_n)$ for all $n \geq 1$. Set $Z \coloneqq \mathds{1}_{A^c} - \mathds{1}_A \in L^\infty \subset L$, where $A \in \cF$ satisfies $\mathbb{P}[A] \in (0, \epsilon/(2-l))$ for some $\epsilon \in (0,1)$. This is possible by non-atomicity of $(\Omega,\cF,\mathbb{P})$. Then, $\mathbb{P}[Z < 0] > 0$ and
\begin{align*}
\cE_u(\lambda_n Z) &= \mathbb{P}[A^c]u(\lambda_n)+\mathbb{P}[A]u(-\lambda_n) \ge (1 - \mathbb{P}[A])u(\lambda_n)+ \mathbb{P}[A](\ell-1)u(\lambda_n) \\
& =  u(\lambda_n)(1 - \mathbb{P}[A](2 - \ell)) \geq u(\lambda_n)(1-\epsilon) \to (1-\epsilon)\cE_u(\infty).
\end{align*}
Since $\lim_{y \to \infty} u(y) = \infty$, it follows that $\infty = \cE_u(\infty) = (1-\epsilon)\cE_u(\infty)$. Thus $\cE_u$ satisfies SE.
\end{proof}

\begin{proof}[\textbf{Proof of Proposition \ref{prop:Gaussian fixed market risk SLL}}]
For $\pi\ne\mathbf0$, law-invariance, cash-additivity and positive homogeneity give
\begin{equation*}
\mathcal R(X_\pi)
=
\sigma_\pi\mathcal R(Z)+\mu_\pi\mathcal R(1)
=
|\mathcal R(1)|
\bigl(\kappa_{\mathcal R}\sigma_\pi-\mu_\pi\bigr).
\end{equation*}
By positive homogeneity, $\mathcal R$ satisfies $\SLL$ on $\bar S$ if and only
if $\mathcal R(X_\pi)>0$ for every $\pi\ne\mathbf0$. Applying this condition
to both $\pi$ and $-\pi$ shows that it is equivalent to
$
|\mu_\pi|<\kappa_{\mathcal R}\sigma_\pi$
for every $\pi\ne\mathbf0$. The result follows, since
$\sup_{\pi\ne\mathbf0}
\frac{|\mu_\pi|}{\sigma_\pi}
=
\sqrt{\mu^\top\Sigma^{-1}\mu}
=
\SR_{\max}(X).$ \qedhere
\end{proof}

\begin{proof}[\textbf{Proof of Theorem \ref{thm:Gaussian fixed market characterization}}]
The functional $\mathcal E_u$ satisfies the upper Fatou property, while
$\mathcal R$ satisfies the assumptions of Theorem
\ref{thm:local wp original characterization}.

We first observe that $\mathcal E_u$ satisfies neither robust $\ASLL$ nor robust $\SLL$ on $\bar S$. If $a > b$, then every non-zero Gaussian payoff satisfies
\begin{equation*}
\mathcal E_u(\lambda X_\pi)
=
\lambda^a \mathbb E[(X_\pi^+)^a]
-
\lambda^b \mathbb E[(X_\pi^-)^b]
\to\infty.
\end{equation*}
If $a = b =:q$, choose $\pi\ne\mathbf0$ with $\mu_\pi\ge0$. When $\mu_\pi>0$, symmetry and strict monotonicity of the odd function $u$ give $\mathcal E_u(X_\pi)>0$; when $\mu_\pi=0$, the same holds after any positive constant perturbation. Since $u(\lambda y)=\lambda^q u(y)$, robust $\ASLL$ again fails, and hence so does robust $\SLL$.

We next verify alignment. Cash-additivity and positive homogeneity give
\begin{equation*}
\mathcal R(X_\pi)
=
\mathcal R(Z)\sigma_\pi+\mathcal R(1)\mu_\pi
=
-\mathcal R(1)\bigl(\kappa_{\mathcal R}\sigma_\pi-\mu_\pi\bigr).
\end{equation*}
Thus, if $\mathcal R(\lambda X_\pi)\le0$ for every $\lambda>0$, then $\mu_\pi\ge\kappa_{\mathcal R}\sigma_\pi>0.$ The preceding calculations imply $\mathcal E_u(\lambda X_\pi)\to\infty$, both when $a > b$ and when $a = b$. Hence $\mathcal R$ is aligned with $\mathcal E_u$ on
$\bar S$.

Theorem \ref{thm:local wp original characterization} therefore shows that weak well-posedness and well-posedness are equivalent to $\mathcal R$ satisfying $\SLL$ on $\bar S$. By Proposition \ref{prop:Gaussian fixed market risk SLL}, this is equivalent to $\SR_{\max}(X)<\kappa_{\mathcal R}$.
\end{proof}

\small
\bibliography{biblio} 

@article{Markowitz1952,
  title={Portfolio selection},
  author={Markowitz, H.},
  journal={J. Finance},
  volume={7},
  number={1},
  pages={77--91},
  year={1952},
}

@article{cheridito2009risk,
  title={Risk measures on {O}rlicz hearts},
  author={Cheridito, P. and Li, T.},
  journal={Math. Finance},
  volume={19},
  number={2},
  pages={189--214},
  year={2009},
  publisher={Wiley Online Library}
}

@article{artzner1999coherent,
  title={Coherent measures of risk},
  author={Artzner, P. and Delbaen, F. and  Eber, J.-M. and  Heath, D.},
  journal={Math. Finance},
  volume={9},
  number={3},
  pages={203--228},
  year={1999},
  publisher={Wiley Online Library}
}

@incollection{follmer2013convex,
  title={Convex risk measures: Basic facts, law-invariance and beyond, asymptotics for large portfolios},
  author={F{\"o}llmer, H. and Knispel, T.},
  booktitle={Handbook of the Fundamentals of Financial Decision Making: Part II},
  pages={507--554},
  year={2013},
  publisher={World Scientific}
}

@article{herdegen2020dual,
  title={Mean-$\rho$ portfolio selection and $\rho$-arbitrage for coherent risk measures},
  author={Herdegen, M. and Khan, N.},
  journal={Math. Finance},
  volume={32},
  number={1},
  pages={226--272},
  year={2022},
  publisher={Wiley Online Library}
}

@incollection{delbaen2002coherent,
  title={Coherent risk measures on general probability spaces},
  author={Delbaen, F.},
  booktitle={Advances in finance and stochastics},
  pages={1--37},
  year={2002},
  publisher={Springer}
}

@article{armstrong2019risk,
	title={Risk managing tail-risk seekers: VaR and expected shortfall vs S-shaped utility},
	author={Armstrong, J. and Brigo, D.},
	journal={J. Bank. Finance},
	volume={101},
	pages={122--135},
	year={2019},
	publisher={Elsevier}
}

@article{castagnoli2021star,
  title={Star-shaped risk measures},
  author={Castagnoli, E. and Cattelan, G. and Maccheroni, F. and Tebaldi, C. and Wang, R.},
  journal={Oper. Res.},
  volume={70},
  number={5},
  pages={2637--2654},
  year={2022},
  publisher={INFORMS}
}

@ARTICLE{MasterFundsRockafellar,
 author={Rockafellar, R. T. and Uryasev, S. and Zabarankin, M.},
 title={Master funds in portfolio analysis with general deviation measures},
 journal={J. Bank. Finance},
 volume={30},
 number = {2},
 pages={743-778},
 year={2006},
}

@article{alexander2002economic,
  title={Economic implications of using a mean-VaR model for portfolio selection: A comparison with mean-variance analysis},
  author={Alexander, G. J. and Baptista, A. M.},
  journal={J. Econ. Dyn. Control},
  volume={26},
  number={7-8},
  pages={1159--1193},
  year={2002},
  publisher={Elsevier}
}

@article{rockafellar2002conditional,
  title={Conditional value-at-risk for general loss distributions},
  author={Rockafellar, R. T. and Uryasev, S.},
  journal={J. Bank. Finance},
  volume={26},
  number={7},
  pages={1443--1471},
  year={2002},
  publisher={Elsevier}
}

@book{follmerschied:2016,
 author={F{\"o}llmer, H. and \vspace{0mm}Schied, A.},
 title={Stochastic Finance},
 edition={Fourth},
 series={de Gruyter Studies in Mathematics},
 volume={27},
 publisher={Walter de Gruyter \& co.},
 address={Berlin},
 year={2016},
}

@article{HKM2024,
  title   = {Risk, Utility and Sensitivity to Large Losses},
  author  = {Herdegen, M. and Khan, N. and Munari, C.},
  journal = {Math. Finance},
  note    = {To appear},
  year    = {2026}
}

@article{herdegen2025rho,
  title={$\rho$-arbitrage and $\rho$-consistent pricing for star-shaped risk measures},
  author={Herdegen, M. and Khan, N.},
  journal={Math. Oper. Res.},
  volume={50},
  number={2},
  pages={1555--1583},
  year={2025},
  publisher={Informs}
}

@article{yaari1987dual,
  title={The dual theory of choice under risk},
  author={Yaari, M. E.},
  journal={Econometrica},
  pages={95--115},
  year={1987},
  publisher={JSTOR}
}

@article{maccheroni2006ambiguity,
  title={Ambiguity aversion, robustness, and the variational representation of preferences},
  author={Maccheroni, F. and Marinacci, M. and Rustichini, A.},
  journal={Econometrica},
  volume={74},
  number={6},
  pages={1447--1498},
  year={2006},
  publisher={Wiley Online Library}
}

@article{landsberger1990lotteries,
  title={Lotteries, insurance, and star-shaped utility functions},
  author={Landsberger, M. and Meilijson, I.},
  journal={J. Econ. Theory},
  volume={52},
  number={1},
  pages={1--17},
  year={1990},
  publisher={Elsevier}
}

@article{kahn1979prospect,
  title={Prospect theory: An analysis of decision under risk},
  author={Kahneman, D. and Tversky, A.},
  journal={Econometrica},
  volume={47},
  number={2},
  pages={363--391},
  year={1979}
}

@article{kramkov1999asymptotic,
  title={The asymptotic elasticity of utility functions and optimal investment in incomplete markets},
  author={Kramkov, D. and Schachermayer, W.},
  journal={Ann. Appl. Probab.},
  pages={904--950},
  year={1999},
  publisher={JSTOR}
}

@article{schachermayer2001optimal,
  title={Optimal investment in incomplete markets when wealth may become negative},
  author={Schachermayer, W.},
  journal={Ann. Appl. Probab.},
  pages={694--734},
  year={2001},
  publisher={JSTOR}
}

@article{Rasonyi2005,
 author = {R{\'a}sonyi, M. and Stettner, L.},
 journal = {Ann. Appl. Probab.},
 number = {2},
 pages = {1367--1395},
 publisher = {Institute of Mathematical Statistics},
 title = {On Utility Maximization in Discrete-Time Financial Market Models},
 volume = {15},
 year = {2005}
}

@article{campbell2001optimal,
  title={Optimal portfolio selection in a Value-at-Risk framework},
  author={Campbell, R. and Huisman, R. and Koedijk, K.},
  journal={J. Bank. Finance},
  volume={25},
  number={9},
  pages={1789--1804},
  year={2001},
  publisher={Elsevier}
}

@article{ciliberti2007feasibility,
  title={On the feasibility of portfolio optimization under expected shortfall},
  author={Ciliberti, S. and Kondor, I. and M{\'e}zard, M.},
  journal={Quant. Finance},
  volume={7},
  number={4},
  pages={389--396},
  year={2007},
  publisher={Taylor \& Francis}
}

@article{basak2001value,
  title={Value-at-risk-based risk management:\ optimal policies and asset prices},
  author={Basak, S. and Shapiro, A.},
  journal={Rev. Financ. Stud.},
  volume={14},
  number={2},
  pages={371--405},
  year={2001},
  publisher={Oxford University Press}
}

@article{gundel2008utility,
  title={Utility maximization under a shortfall risk constraint},
  author={Gundel, A. and Weber, S.},
  journal={J. Math. Econ.},
  volume={44},
  number={11},
  pages={1126--1151},
  year={2008},
  publisher={Elsevier}
}

@article{armstrong2022coherent,
  title={Coherent risk measures alone are ineffective in constraining portfolio losses},
  author={Armstrong, J. and Brigo, D.},
  journal={J. Bank. Finance},
  volume={140},
  pages={106315},
  year={2022},
  publisher={Elsevier}
}

@article{armstrong2024importance,
  title={The importance of dynamic risk constraints for limited liability operators},
  author={Armstrong, J. and Brigo, D. and Tse, A.},
  journal={Ann. Oper. Res.},
  volume={336},
  number={1},
  pages={861--898},
  year={2024},
  publisher={Springer}
}

@article{gabih2009utility,
  title={Utility maximization under bounded expected loss},
  author={Gabih, A. and Sass, J. and Wunderlich, R.},
  journal={Stoch. Models},
  volume={25},
  number={3},
  pages={375--407},
  year={2009},
  publisher={Taylor \& Francis}
}

@article{adam2008spectral,
  title={Spectral risk measures and portfolio selection},
  author={Adam, A. and Houkari, M. and Laurent, J-P.},
  journal={J. Bank. Finance},
  volume={32},
  number={9},
  pages={1870--1882},
  year={2008},
  publisher={Elsevier}
}

@article{bernard2010static,
  title={Static portfolio choice under cumulative prospect theory},
  author={Bernard, C. and Ghossoub, M.},
  journal={Math. Financ. Econ.},
  volume={2},
  number={4},
  pages={277--306},
  year={2010},
  publisher={Springer}
}

@article{he2011portfolio,
  title={Portfolio choice under cumulative prospect theory: An analytical treatment},
  author={He, X. D. and Zhou, X. Y.},
  journal={Manag. Sci.},
  volume={57},
  number={2},
  pages={315--331},
  year={2011},
  publisher={INFORMS}
}

@article{ghossoub2025risk,
  title={Risk-constrained portfolio choice under rank-dependent utility},
  author={Ghossoub, M. and Zhu, M. B.},
  journal={Financ. Stoch.},
  volume={29},
  number={2},
  pages={399--442},
  year={2025},
  publisher={Springer}
}

@article{wei2018risk,
  title={Risk management with weighted VaR},
  author={Wei, P.},
  journal={Math. Finance},
  volume={28},
  number={4},
  pages={1020--1060},
  year={2018},
  publisher={Wiley Online Library}
}

@article{he2015dynamic,
  title={Dynamic portfolio choice when risk is measured by weighted VaR},
  author={He, X. D. and Jin, H. and Zhou, X. Y.},
  journal={Math. Oper. Res.},
  volume={40},
  number={3},
  pages={773--796},
  year={2015},
  publisher={INFORMS}
}

@article{he2011portfolioquantiles,
  title={Portfolio choice via quantiles},
  author={He, X. D. and Zhou, X. Y.},
  journal={Math. Finance},
  volume={21},
  number={2},
  pages={203--231},
  year={2011},
  publisher={Wiley Online Library}
}

@article{owen1983class,
  title={On the class of elliptical distributions and their applications to the theory of portfolio choice},
  author={Owen, J. and Rabinovitch, R.},
  journal={J. Finance},
  volume={38},
  number={3},
  pages={745--752},
  year={1983},
  publisher={Wiley Online Library}
}

@article{chamberlain1983characterization,
  title={A characterization of the distributions that imply mean--Variance utility functions},
  author={Chamberlain, G.},
  journal={J. Econ. Theory},
  volume={29},
  number={1},
  pages={185--201},
  year={1983},
  publisher={Elsevier}
}

@article{schuhmacher2021justifying,
  title={Justifying mean-variance portfolio selection when asset returns are skewed},
  author={Schuhmacher, F. and Kohrs, H. and Auer, B. R.},
  journal={Manag. Sci.},
  volume={67},
  number={12},
  pages={7812--7824},
  year={2021},
  publisher={INFORMS}
}

@article{rockafellar2000optimization,
  title={Optimization of conditional value-at-risk},
  author={Rockafellar, R. T. and Uryasev, S.},
  journal={J. Risk},
  volume={2},
  pages={21--42},
  year={2000}
}

@article{bertsimas2004shortfall,
  title={Shortfall as a risk measure: properties, optimization and applications},
  author={Bertsimas, D. and Lauprete, G. J. and Samarov, A.},
  journal={J. Econ. Dyn. Control},
  volume={28},
  number={7},
  pages={1353--1381},
  year={2004},
  publisher={Elsevier}
}

@article{geissel2022portfolio,
  title={Portfolio optimization with optimal expected utility risk measures},
  author={Geissel, S. and Graf, H. and Herbinger, J. and Seifried, F. T.},
  journal={Ann. Oper. Res.},
  volume={309},
  number={1},
  pages={59--77},
  year={2022},
  publisher={Springer}
}

@article{gao2018fatou,
  title={Fatou property, representations, and extensions of law-invariant risk measures on general Orlicz spaces},
  author={Gao, N. and Leung, D. and Munari, C. and Xanthos, F.},
  journal={Finance Stoch.},
  volume={22},
  number={2},
  pages={395--415},
  year={2018},
  publisher={Springer}
}
\bibliographystyle{plainnat}
 
\end{document}